\documentclass[journal,twoside,web]{ieeecolor}
\usepackage{generic}

\usepackage{cite}

\usepackage{amsmath}
\let\labelindent\relax
 
\usepackage{amsthm}

\usepackage{amssymb,amsfonts}
\usepackage{algorithmic}
\usepackage{graphicx}
\usepackage{algorithm,algorithmic}
\usepackage{hyperref}
\hypersetup{hidelinks=true}
\usepackage{textcomp}
\usepackage{enumitem}
\usepackage{caption}
\usepackage{subcaption}
\usepackage{cleveref}
\usepackage{thmtools}
\usepackage{thm-restate}
\usepackage{shorthands}
\usepackage{array}
\usepackage{tikz}
\usepackage{framed}
\usetikzlibrary{arrows.meta, positioning, fit,calc}
\newcommand{\citep}[1]{\cite{#1}}
\newcommand{\citet}[1]{\cite{#1}}

\newcommand{\nmpparam}{\ensuremath{\xi}}
\def\BibTeX{{\rm B\kern-.05em{\sc i\kern-.025em b}\kern-.08em
    T\kern-.1667em\lower.7ex\hbox{E}\kern-.125emX}}

\newcommand{\TAC}{0}
\newif\ifTACmode
\ifnum\TAC=1
    \TACmodetrue
\else
    \TACmodefalse
\fi

\ifTACmode
\else
\fi
\begin{document}

\title{The Fragility of Learning LQG Controllers}
\author{Bruce D. Lee*, Anastasios Tsiamis*, Nikolai Matni, Manfred Morari, John Lygeros
\thanks{*Equal Contribution}
\thanks{This work is supported in part by an ETH AI Center Postdoctoral Fellowship to Bruce D. Lee. This work is also partially supported by NCCR Automation, grant agreement 51NF40\_225155 from the Swiss National Science Foundation and NSF CAREER award ECCS-
204583.}
\thanks{Bruce D. Lee is with the ETH AI Center. Anastasios Tsiamis, Manfred Morari, and John Lygeros are with the Institute for Automatic Control at ETH Zurich. Nikolai Matni is with the Department of Electrical and Systems Engineering at the University of Pennsylvania. Emails: \tt\small bruce.lee@ai.ethz.ch, \{atsiamis,mmorari,jlygeros\}@control.ee.ethz.ch, nmatni@seas.upenn.edu.}
}

\maketitle
\ifTACmode
\else
    \thispagestyle{empty}
\fi

\begin{abstract}
  Learning methods are  increasingly used to synthesize controllers from data, yet existing sample-complexity characterizations for continuous control are sharp only in the fully observed setting. This paper studies the partially observed case by deriving information-theoretic lower bounds for learning Linear Quadratic Gaussian (LQG) controllers from offline trajectories generated by a (linear) exploration policy. We prove an $\varepsilon$-local minimax excess-cost lower bound that applies to any algorithm mapping the offline dataset to a stabilizing linear controller. The bound is expressed in terms of the Hessian of the LQG cost with respect to model parameters and the inverse Fisher Information induced by the exploration policy. We provide system-theoretic characterizations of these objects, enabling transparent construction of hard instances. Instantiating the bound on classical fragile robust-control examples, including variants of the Doyle LQG fragility counterexample and non-minimum-phase systems, demonstrates when fragile robust control problems translate
into high sample complexity for learning-enabled control.
These results suggest the asymptotic optimality of certainty-equivalent synthesis and motivate the importance of both task-directed experiment design and system co-design for sample-efficient learning in partially observed control.
\end{abstract}

\section{Introduction}
The problem of learning to control partially observed systems remains poorly understood, even for the simple case of linear systems with quadratic costs and Gaussian noise (LQG). Although the robust control community has long recognized that partial observations introduce substantial additional subtleties~\cite{doyle1978guaranteed}, this insight is largely absent from the recent literature on learning-based control and reinforcement learning.
In this paper, we fill this gap by examining the fundamental limits and inherent fragility of learning the optimal LQG controller. We characterize the complexity of this problem through its data requirements, that is, its statistical complexity.

\subsection{Contributions}
We make the following contributions:
\begin{itemize}[noitemsep, nolistsep, leftmargin=*]
    \item We provide an information-theoretic lower bound on the sample complexity of learning LQG controllers from offline data. The bound applies to any algorithm mapping data to a controller, and thus shows that common practical strategies for mitigating partial observability in control and reinforcement learning (e.g., history stacking \citep{mnih2015human}) cannot circumvent the fundamental hardness of the problem.
    \item We show the bound in terms of the underlying dynamical system, allowing the construction of ``hard'' instances.\footnote{Code to evaluate the hardness of any instance is available at \url{https://colab.research.google.com/drive/12BtSCDb8-gYhx-7CwIIw6FaC8Z3OnGmM?usp=sharing} }%
    \item We study classical examples of fragile problems from robust control, including an example inspired by Doyle's counterexample \citet{doyle1978guaranteed} and a non-minimum phase system. We show that such fragile instances result in learning problems with high sample complexity. 
\end{itemize}
These contributions serve as an important step to enable the design of optimal algorithms, e.g. by task-directed experiment design. They also shed light on the importance of system design as a critical step in the engineering pipeline to enable sample efficient learning, a step that is often neglected in reinforcement learning due to the desire to apply black-box algorithms. By exposing hard instances of linear systems, the results additionally enable strategic design of benchmarks for reinforcement learning that have a clear hierarchy of difficulty. 

\subsection{Related Work}

\paragraph{Complexity of Learning-Enabled 
Control} Recent literature has made significant progress in understanding the statistical complexity of learning to control in the setting 
of fully observed linear systems. The authors of \citet{abbasi2011regret} provide a bound on the regret of learning to control a linear system online, while \citet{dean2020sample} bound the excess cost of an offline identification and control scheme. Bounds that appear tight in state and input dimension were introduced by \citet{mania2019certainty} for the offline setting,
and by \citet{simchowitz2020naive} for the online setting. Subsequent work has shown that such bounds neglected terms that could scale exponentially in system dimensions \citep{tsiamis2021linear, tsiamis2022learning, ziemann2022policy, chen2021black}, indicating that the focus should be upon  achieving bounds that are tight in system-theoretic quantities rather than system dimensions. Progress in this direction is presented by \citet{wagenmaker2021task, wagenmaker2023optimal, lee2024active, ziemann2024regret}. However,
these works are restricted to the fully observed setting or special cases of partially observed systems. Upper bounds on learning the LQG are presented in the offline setting by 
\citep{mania2019certainty,zheng2021sample} and in the online setting by \citep{lale2020logarithmic} but without a tight dependence on the underlying system instance. In particular, they rely on bounds for partially observed system identification \citep{oymak2019non, ziemann2023tutorial} which are not asymptotically tight.
We instead give tight lower bounds for partially observed offline LQG learning.  %

\paragraph{Fundamental Limitations of Robust Control} A long line of work in robust control has shown that partially observed systems can exhibit fundamental fragility to model uncertainty.
The seminal paper  \citet{doyle1978guaranteed} was the first to show that LQG controllers could be arbitrarily fragile to model uncertainty, despite the guaraneed inherent robustness of the linear quadratic regulator (LQR). Subsequently, robust alternatives to LQG control were developed, which could mitigate the issue \citep{doyle1988state}. Nonetheless, certain systems that are ill-conditioned \citep{skogestad1988robust} or exhibit a challenging observability structure due to unstable zeros \citep{doyle2013feedback} lack robustness under \emph{any} controller. 
In \citet{leong2016understanding}, the authors provide a case study illustrating such limitations on the toy problem of stabilizing a cartpole system with visual observations. This problem is studied from an empirical perspective in \citet{xu2021learned} by applying a reinforcement learning algorithm to a partially observed cartpole system, and highlighting that the systems which robust control theory predicts are fragile require a large number of samples to learn an effective controller. We provide theoretical justification for such observations by deriving lower bounds on the cost incurred by any learning algorithm applied to control a partially observed linear system.  We instantiate our lower bound on several classical examples of fragile robust control problems, and demonstrate when they correspond to challenging learning-based control problems.

\paragraph{Interplay of System Identification with 
Control}
When the system model must be estimated from data, the fragility phenomena described above place a significant burden on the identification process. In particular, small errors can lead to large degradation in control performance, necessitating very accurate identification~\cite{zeng2023hardness}. This has led to extensive study of the interplay of system identification with control, see e.g. \citet{gevers2005identification} and references therein.
One of the primary takeaways from such literature is that experiment design must be performed in a way that accounts for the downstream control task \citep{rivera2003plant, gevers2005identification, hjalmarsson2005experiment, rojas2007robust, hjalmarsson2009system, bombois2011optimal}. A case study for one sensitive example, controlling a high purity distillation column, is studied by \citet{rivera2007high}. 
In contrast,
our analysis takes an algorithm-agnostic perspective. In particular, we derive lower bounds on the excess cost incurred by \emph{any} algorithm that maps an offline dataset to a controller. These bounds justify identification-for-control: data collection must excite the most task-relevant parameter directions.

\textbf{Notation: } For conforming matrices $A, B, C, D$, we identify $\brac{\begin{array}{c|c}
            A & B \\ \hline  C & D
            \end{array}} $ with the transfer function $C(zI - A)^{-1} B + D$. For a transfer function $\calM$ of dimension $(\mathsf{dz} + \dy) \times (\mathsf{d \eta} + \du) $, denote its feedback interconnection with $K$ of dimension $\du \times \dy$ as
\begin{align*}
    \calF(\calM, K) \triangleq  \calM_{11} + \calM_{12}(I - \calM_{22}K)^{-1} \calM_{21}, 
\end{align*}
where $\bmat{\calM_{11} & \calM_{12} \\ \calM_{21} & \calM_{22}} = \calM$ are size conforming blocks. For a transfer matrix $P$, let $P^*$ denote its complex conjugate. Let $\calR^{m \times n}_p$ and $\calR^{m \times n}_s$ denote the set of real rational transfer functions of dimension $m \times n$ that are  proper and strictly proper, respectively. Denote by $\mathcal{RH}_\infty$ the set of real rational proper transfer functions analytic in $\abs{z} > 1$. For a stable transfer function $P$, the $\calH_2$ norm is defined as $\norm{P}_{\calH_2} = (\frac{1}{2\pi} \int_{-\pi}^{\pi} \trace (P(e^{j\omega})^*P(e^{j\omega})) \; d \omega)^{1/2}$  and the $\calH_{\infty}$ norm is given by $\norm{P}_{\calH_{\infty}} = \max_{\omega \in [-\pi, \pi]} \norm{P(e^{j\omega})}$. 
For matrices $A \in \R^{n\times n},B\in\R^{n\times m},Q \in \R^{n \times n},R \in \R^{m\times m}$, with $Q \succeq 0$ and $(A, Q^{1/2})$ detectable, and $R \succ 0$, $\dare(A,B,Q,R)$ denotes the solution $P$ to the algebraic Riccati equation, 
    $P = A^\top P A - A^\top P B(B^\top P B+ R)^{-1} B^\top P A + Q$
and $\dlyap(A,Q)$ denotes the solution $X$ to the Lyapunov equation 
    $X = A X A^\top + Q. $
The Kronecker product of two matrices $A$ and $B$ is denoted $A \otimes B$. For a matrix $A$, $\mathsf{sym}(A) = A + A^\top$.
For a sequence of random variables $\curly{X_n}_{n=1}^\infty$ and a random variable $X$, $X_n \overset{d}{\to} X$ denotes convergence in distribution. Let $\calN(\mu,\Sigma)$ denote a normal distribution (multivariate) with mean $\mu$ and covariance $\Sigma$. For functions $f,g : D \to R$, the notation $f(x) = O(g(x))$ means that there exists a positive constant $c$ such that $f(x) \leq c g(x)$ for all $x \in D$, while $f(x) = \Omega(g(x))$ means that there exists a constant $c$ such that $f(x) \geq c g(x)$ for all $x \in D$. The notation $f(x) = o(g(x))$ denotes a quantity that satisfies $\lim_{x\to 0}\frac{f(x)}{g(x)} = 0$.   %

\ifTACmode
\textbf{The supplementary material} available with the TAC submission contains omitted proofs, extention to the non-strictly causal setting, and more examples.
\else{ }\fi

\section{Problem formulation}
\label{sec: problem formulation}
\begin{figure}
    \centering
\resizebox{\columnwidth}{!}{
\begin{tikzpicture}[
    block/.style={draw, rounded corners, align=center, minimum width=1.2cm, minimum height=0.6cm},
    blueblock/.style={block, fill=cyan!20, draw=none},
    dashedbox/.style={draw, dashed, minimum width=3cm, minimum height=2.7cm},
    arrow/.style={thick, -{Latex[length=2mm]}},
    myarrow/.style={arrow, line width=1.1mm,-{Latex[length=3.0mm, width=3.4mm]}},
     myarrow2/.style={arrow, double distance=1pt, line width=0.5mm,-{Latex[length=3.5mm, width=4.0mm,fill=white]}},
      dashedbox2/.style={draw, dashed, minimum width=1.5cm, minimum height=2.5cm},
]

\node[blueblock] (ptheta) {$\mathcal P_{\theta}$};
\node[blueblock, below=0.8cm of ptheta] (texp) {$\pi_{\mathrm{exp}}$};

\node[dashedbox, fit=(ptheta)(texp)] (expbox) {};
\node[below=0pt of expbox] (experiments) {experiments};

\draw[arrow] (ptheta.east) -- node[above]{$y$} ++(0.6,0) |- (texp.east);
\draw[arrow] (texp.west) -- node[above]{$u$} ++(-0.6,0) |- ($(ptheta.west)+(0,-0.2)$);

\draw[arrow] ($(ptheta.west)+(-0.62,0.2)$) --node[above]{$d$} ($(ptheta.west)+(0,0.2)$);
\node[right=0.65cm of expbox] (data) {
    $\begin{aligned}
        y_{0:T-1}^{1},& u_{0:T-1}^{1} \\
        &\vdots \\
        y_{0:T-1}^{N},& u_{0:T-1}^{N}
    \end{aligned}$
};
\node[right=50pt of experiments] (datatext) {data $\mathcal D$};

\draw[myarrow] ($(expbox.east)+(0.1,0)$) -- (data.west);

\node[block, fill=orange!20, draw=none, minimum width=1.6cm, minimum height=1cm,
      right=3.5cm of expbox] (A) {$\calA$};
\node[below=0pt of A, align=center] {learning\\algorithm};

\draw[myarrow] (data.east) -- (A.west);

\node[right=2.2cm of A] (policy) {};

\draw[myarrow] (A.east) -- node[below=0pt]{policy} node[above=0pt]{$K_{\mathcal{D}}\in\calR^{\du\times\dy}_s$} (policy.west);

\node[below right=1.5cm and 0 of expbox,blueblock] (deploy_system){$\mathcal{P}_{\theta}$};

\node[blueblock, fill=orange!20, below=0.8cm of deploy_system] (deploy_controller) {$K_{\mathcal{D}}$};

\draw[arrow] (deploy_system.east) -- node[above]{$y$}  ++(0.6,0) |- (deploy_controller.east);
\draw[arrow] (deploy_controller.west) -- node[above]{$u$} ++(-0.6,0) |- ($(deploy_system.west)+(0,-0.2)$);
\draw[arrow] ($(deploy_system.west)+(-0.62,0.2)$) --node[above]{$d$} ($(deploy_system.west)+(0,0.2)$);

\node[dashedbox,fit=(deploy_system)(deploy_controller)](closed_loop_deploy){};
\node[right=2cm of closed_loop_deploy,align=center] (closed_loop_cost) {$J(K_\mathcal{D},\theta)$\\LQG cost};
\draw[myarrow] ($(closed_loop_deploy.east)+(0.2,0)$) -- (closed_loop_cost);

\node[below=0pt of closed_loop_deploy] (closed_loop_deploy_label) {deployment};
\end{tikzpicture}}
	    \caption{Offline reinforcement learning. Input-output data $\mathcal D$ ($N$ independent length-$T$ trajectories) is collected from the unknown system $\mathcal{P}_\theta$ under exploration policy $\pi_{\exp}$. A learning algorithm maps the data to a dynamic linear policy $K_\mathcal{D}$, whose closed-loop deployment incurs cost $J(K_\mathcal{D},\theta)$.
	    }
	    \label{fig: architecture}
        \vspace{-1em}
\end{figure}

We consider the problem of learning to control a partially observed linear dynamical system with an \textbf{unknown} parameter $\theta \in \R^{d_{\theta}}$. The evolution of the system is governed by
\begin{equation}
\label{eq: linsys}
\begin{aligned}
    x_{t+1} &= A(\theta) x_t + B(\theta) u_t + w_t \\
    y_t &= C(\theta) x_t + v_t,
\end{aligned}
\end{equation}
where $x_t \in \R^{\dx}$ denotes the system state, $y_t \in \R^{\dy}$ is the observation, $u_t \in \R^{\du}$ is the control action, and $w_t \in \R^{\dx}$ and $v_t \in \R^{\dy}$ are zero-mean Gaussian noise that are independent across time and from each other with covariance  matrices $\Sigma_w(\theta)$ and $\Sigma_v(\theta)$, respectively. We denote the joint noise signal as $d_t = \bmat{w_t \\ v_t}$. We assume that the initial state $x_0$ is also zero-mean Gaussian with covariance $\Sigma_0(\theta))$ and
is independent from the noise variables. The system matrices $A(\cdot)$, $B(\cdot)$, $C(\cdot)$, $\Sigma_w(\cdot)$, $\Sigma_v(\cdot)$ and $\Sigma_0(\cdot)$ are analytic functions of the unknown parameter $\theta$. We additionally assume $\Sigma_v(\theta) \succ 0$, that the tuple $(A(\theta), B(\theta), C(\theta))$ is stabilizable and detectable, and that the pair $(A(\theta)^\top, \Sigma_w(\theta)^{1/2})$ is detectable. 

\subsection{Objective}
 Let $K$ denote a control policy that maps the available observations $y_{0:t}$ to a control action $u_t$. Let $\calK$ denote the class of available control policies. The nominal control objective is to find a policy $K \in \calK$ that minimizes the following Linear Quadratic Gaussian (LQG) cost
\begin{align}
    \label{eq: control objective}
    J(K, \theta) = \limsup_{T\to\infty} \frac{1}{T} \bfE^{K}_{\theta} \brac{\sum_{t=0}^T x_t^\top Q x_t + u_t^\top R u_t},
\end{align}
where $Q \succeq 0$ with $(A(\theta), Q^{1/2})$ detectable and $R \succ 0$. The subscript in the expectation denotes that the states evolve according to the dynamics \eqref{eq: linsys} with the parameter $\theta$, and the superscript denotes that the inputs are generated according to the policy $K$. We denote the optimal solution to this problem for a fixed parameter $\theta$ as 
\begin{align}
    \label{eq: solution}
    K_{\theta} \in \argmin_{K\in \calK } J(K, \theta).
\end{align}
It is well established that when $\calK$ consists of all measurable functions of the past data, $y_{0:t}$ (causal policies), or of all measurable functions of the past data excluding the most recent measurement, $y_{0:t-1}$ (strictly causal policies), the optimal LQG controller is a linear dynamic feedback policy that can be represented as a proper transfer function $K_\theta\in\calR_p^{\du\times\dy}$, or a strictly proper transfer function $K_\theta\in\calR_s^{\du\times\dy}$, respectively. This policy can be calculated explicitly given $\theta$; see Section~\ref{s: LQG_preliminaries}. 
\ifTACmode
We therefore restrict the policy class to linear dynamical policies represented by $\calK=\calR_{s}^{\du\times \dy}$; in this case, the minimizer in~\eqref{eq: solution} is unique and well-defined. The extension of our result to non-strictly causal policies can be found in the supplementary material.
\else
We therefore restrict the policy class to linear dynamical policies represented by  $\calK=\calR_{p}^{\du\times \dy}$ or $\calK=\calR_{s}^{\du\times \dy}$; in this case, the minimizer in~\eqref{eq: solution} is unique and well-defined. While we present our results for both settings, we focus primarily on strictly causal policies for ease of exposition.
\fi %

Since the system $\theta$ is \textbf{unknown}, the learner must estimate a controller from data to minimize $J(K, \theta)$. We adopt the offline reinforcement learning perspective, visualized in \Cref{fig: architecture}. First, one collects a dataset $\calD = \curly{(y_t^n, u_t^n)}_{t=0, n=1}^{T-1,N}$ of $N$ independent trajectories of length $T$  from \eqref{eq: linsys}, using a history-dependent and potentially stochastic exploration policy $\pi_{\exp}$. This policy is not restricted to lie in $\calK$. We denote the distribution over such datasets as $\rho^{\pi_{\exp}, N,T}(\theta)$. 

Using the dataset $\calD$ collected by the exploration policy, the learner synthesizes a control policy $K_{\calD} \in\calK$.
 Such an estimated controller incurs an \textit{excess cost} over the optimal controller of $J(K_\mathcal{D}, \theta) - J(K_{\theta}, \theta)$. The excess cost depends on the amount and quality of data available, as well as the efficiency of the learning algorithm. We seek to answer the question:

\begin{quote}
\begin{center}
\textit{What are the fundamental limits for the decay in the excess cost incurred by a controller estimated from data?}
\end{center}
\end{quote}

To answer the above question, we pursue lower bounds on the excess cost that are valid for \emph{any} learning algorithm. We consider \emph{instance-specific minimax} lower bounds, which capture the worst case performance of any learning algorithm $\mathcal{A}$ near an instance of interest $\theta^\star$. 
Specifically, we use local minimax theory \citep{duchi2016lecture}, in which we lower bound
 the $\varepsilon$-local minimax excess cost, defined as
 \begin{align*}
     \mathcal{M}(\theta^\star\!, \!\varepsilon, \!\calA) \!\triangleq\!\!\! \!\sup_{\theta \in \calB(\theta^\star, \varepsilon)} \!\!\!\! \E_{\calD \sim \rho^{\pi_{\exp}, N, T}(\theta)} [J(\calA(\calD), \theta) \!-\! \! J(K_\theta,\! \theta)],
 \end{align*}
 where $\varepsilon>0$ sets the level of locality.

Examining a small ball around the nominal instance reveals the hardness of $\theta^\star$, while ruling out degenerate algorithms.
In particular, the supremum ($\varepsilon\neq 0$) rules out algorithms that only work for a single system instance, for example,
algorithms that ignore the data and return the controller $\mathcal{A}(\mathcal{D})=K_{\theta^\star}$.
On the other hand, global minimax lower bounds ($\varepsilon=\infty$) can be overly pessimistic since they may include trivially pathological systems (see, for instance, Figure 1 of \citep{lee2023fundamental}) and are uninformative for $\theta^\star$.

\subsection{Assumptions}
As already stated, we consider linear dynamic policies that are \ifTACmode strictly causal, $\calK = \calR_s^{\du\times\dy}$\else causal or strictly causal, $\calK=\calR_{p}^{\du\times \dy}$ or $\calK = \calR_s^{\du\times\dy}$, respectively\fi. 
\ifTACmode
Note that this class rules out time-varying or adaptive policies, which are outside the scope of the offline setting.
\else
Note that both classes rule out time-varying or adaptive policies, which are outside the scope of the offline setting.
\fi

In the following assumption, we further restrict the complexity of the class $\calK$ by bounding the $\calH_\infty$ norm of the closed-loop transfer matrix of policies in $\calK$ interconnected with the nominal plant $\calP_{\theta^\star}$, where %
\begin{align*}
    \calP_{\theta} \triangleq \bmat{\calP_{\theta}^u & \calP_{\theta}^{d}} = \brac{\begin{array}{c|cc}
    A(\theta) &  B(\theta) & \bmat{\Sigma_w(\theta)^{1/2} & 0}\\[2pt]
    \hline \\[-0.95em] 
    C(\theta)  & 0 & \bmat{0 & \Sigma_v(\theta)^{1/2}}
\end{array}}.
\end{align*}

\begin{assumption}[Internal Stability]
    \label{asmp: closed loop}
    Define the closed loop map of the plant $\calP_{\theta^\star}^u$ interconnected with a controller \ifTACmode $K \in \calR_s^{\du \times \dy}$\else $K \in \calR_p^{\du \times \dy}$\fi as 
    \begin{align*}
        \mathcal{T}_K = \bmat{
            (I-K\calP_{\theta^\star}^u)^{-1} & K(I-\calP_{\theta^\star}^u K)^{-1}  \\
            \calP_{\theta^\star}^u(I-K\calP_{\theta^\star}^u)^{-1} & (I-\calP_{\theta^\star}^u K)^{-1}}.
    \end{align*}
       We assume that the policy class $\calK$ consists of \ifTACmode linear strictly causal controllers, $K\in \calR_s^{\du \times \dy}$\else linear, causal controllers, $K\in \calR_p^{\du \times \dy}$, or of linear strictly causal controllers, $K\in \calR_s^{\du \times \dy}$\fi, such that the closed-loop transfer matrices have uniformly bounded norm, $\sup_{K \in \calK}\norm{\calT_K}_{\calH_\infty}<\infty$.
\end{assumption}
Based on the above, we can define
\begin{equation}\label{eq:alpha} \alpha\triangleq\sup_{K\in\calK}\norm{\calT_K}_{\calH_\infty}-\norm{\calT_{K_{\theta^\star}}}_{\calH_\infty},
\end{equation}
which essentially captures the size of the policy class $\mathcal{K}$.
The block transfer matrix $\mathcal{T}_K$ characterizes the mapping of disturbances before and after the controller in a feedback control loop to measurements before and after the controller. Thus having $\mathcal{T}_K \in \mathcal{RH}_{\infty}$ is necessary and sufficient for internal stability of the closed-loop (Theorem 5.3 \citep{zhou1996robust}). In the special case where the open loop plant $\calP_{\theta^\star}^u$ is stable with a bounded $\calH_{\infty}$ norm, the bound on $\norm{\mathcal{T}_K}_{\calH_\infty}$ reduces to a bound on the $\calH_{\infty}$ norm of the internal model controller $Q = K(I-\calP_{\theta^\star}^u K)^{-1}$ \citep{morari1989robust}.

{\color{magenta}
}

\subsection{Review of LQG optimal control}
\label{s: LQG_preliminaries}
When the system $\theta$ is known, the optimal strictly causal LQG controller $K_\theta\in\calR^{\du\times\dy}_s$ can be computed explicitly in the frequency domain as
\[
K_\theta=F(\theta)(zI-A(\theta) + L(\theta) C(\theta)-B(\theta)F(\theta))^{-1}L(\theta),
\]
where $F(\theta)\in\mathbb{R}^{\du\times\dx}$ is the LQR gain obtained by solving the Linear Quadratic Regulator (LQR) problem, and $L(\theta)\in\mathbb{R}^{\dx\times\dy}$ is the Kalman predictor gain.
The LQR gain is
\begin{align*}
    F(\theta) &= -\Psi(\theta)^{-1} B(\theta)^\top P(\theta) A(\theta); \\
    \Psi(\theta) &= B(\theta)^\top P(\theta) B(\theta) + R;\\
    P(\theta) &= \texttt{DARE}(A(\theta), B(\theta), Q, R),
\end{align*}
while the Kalman predictor gain is  given  by
\begin{align*}
    L(\theta) &= A(\theta) \Sigma(\theta) C(\theta)^\top \Sigma_e(\theta)^{-1};\\
    \Sigma_e(\theta) &= C(\theta) \Sigma(\theta) C(\theta)^\top + \Sigma_v(\theta);\\
    \Sigma(\theta) &= \texttt{DARE}(A(\theta)^\top, C(\theta)^\top, \Sigma_w(\theta), \Sigma_v(\theta)). 
\end{align*}

The optimal LQG controller can also be implemented in the time-domain as
\begin{align}
    \label{eq: lqr control}
    u^\theta_t = F(\theta) \hat x^\theta_{t \vert t-1},
\end{align}
where $\hat x^\theta_{t \vert t-1} \triangleq \bfE_{\theta} \brac{x_t \vert y_{0:t-1}, u_{0:t-1}}$ is the Kalman predictor estimate for the state $x_t$ given the history up to time $t-1$. In particular, 
\begin{align}   
    \label{eq: kf update}
    \hat x^\theta_{t+1\vert t} &= A_{\mathsf{cl}}^o(\theta) \hat x^\theta_{t \vert t-1} + \bmat{L(\theta) & B(\theta)} \bmat{y_{t} \\ u_{t}},
\end{align}
where $A_{\mathsf{cl}}^o(\theta) \triangleq A(\theta) - L(\theta) C(\theta)$ is the closed loop map under the Kalman predictor gain. Under the optimal LQG policy, we set $u_t=u^\theta_t$ in~\eqref{eq: kf update}.
We similarly define the shorthand for the closed loop map under the LQR controller as $A_{\mathsf{cl}}^c(\theta) \triangleq A(\theta) + B(\theta) F(\theta)$. 

\ifTACmode
The optimal causal LQG controller can be expressed similarly; this extension is deferred to the supplementary material.
\else
The optimal causal LQG controller can be expressed similarly, see \Cref{s: non-strictly causal}. 
\fi

Since we are dealing with the average infinite horizon behavior, we can set the initial conditions arbitrarily. A convenient choice is $\hat x_{0\vert -1}^{\theta} = 0$ and $\Sigma_0(\theta) = \Sigma(\theta)$. We will assume this throughout the paper.
Under this choice, the innovation sequence, defined as $e_t=y_t  - C(\theta) \hat x_{t\vert t-1}$, is Gaussian i.i.d., that is, $e_t \overset{iid}{\sim} \calN(0, \Sigma_e(\theta))$. 

\section{Fundamental Limits}

Our main result lower bounds the local minimax excess control cost of any learning algorithm. The lower bound is expressed in terms of two quantities: 
    \paragraph{The Fisher Information} corresponding to the dataset collected by policy $\pi_{\exp}$: 
    \begin{align}
        \label{eq: FI}
        \mathsf{FI}^{\pi_{\exp}}(\theta) &\triangleq \E_{\calD\sim\rho^{\pi_{\exp}, 1,T}(\theta)} \brac{\nabla^2_\theta \mathsf{nll}(\calD; \theta)}.%
    \end{align}
    Here, $\mathsf{nll}$ denotes the negative log likelihood of the dataset for any parameter $\theta$, defined as
    \begin{align*}
       \mathsf{nll}(\calD; \theta)&\!\triangleq\! \frac{1}{2N}\sum_{n=1}^N \sum_{t=0}^{T-1} \norm{y_{t}^n \!-\! \hat y^n_{t}(\theta)}_{\Sigma_e(\theta)^{-1}}^2\!+\!\frac{T}{2} \log \abs{\Sigma_e(\theta)}, \\
        \hat y^n_{t+1}(\theta) &\triangleq
        \sum_{k=0}^t C(\theta) A_{\mathsf{cl}}^o(\theta)^{t-k} \bmat{L(\theta) \!\!\! & B(\theta)} \bmat{y_k^n \\ u_k^n}.
    \end{align*}    
    The Fisher information measures the signal-to-noise ratio and thus captures how easy it is to estimate $\theta$. 
    \paragraph{The Hessian} of the objective under LQG policies:  
    \begin{align}
    \label{eq: Hessian}
        H(\theta) \triangleq  \nabla_{\tilde\theta}^2 J(K_{\tilde \theta}, \theta) \vert_{\tilde \theta = \theta}.%
    \end{align}
    The Hessian captures the sensitivity of the control synthesis problem to the unknown parameters.

\begin{theorem}[Excess Cost Lower Bound]
    \label{thm: excess cost lower bound}
    Let $\calA$ be any algorithm that maps a dataset $\calD$ to a policy $K_{\calD} \in \calK$ for a policy class $\calK$  that obeys Assumption~\ref{asmp: closed loop}. Suppose that $\theta^\star$ and $\varepsilon>0$ are such that all  instances in the ball $B(\theta^\star, \varepsilon)$ satisfy the stabilizability and detectability assumptions of \Cref{sec: problem formulation}. Suppose that the data is collected under an exploration policy $\pi_{\exp}$ such that $\mathsf{FI}^{\pi_{\exp}}(\theta)$ is Lipschitz continuous in $\theta$ over $\calB(\theta^\star, \varepsilon)$ and $\mathsf{FI}^{\pi_{\exp}}(\theta^\star) \succ 0$. 
    There exist constants $c_1(\theta^\star, \pi_{\exp}, T, \alpha)$ and $c_2(\theta^\star, \pi_{\exp}, T)$ that depend on the nominal system, the exploration policy, the length of the experiment rollouts, and the parameter $\alpha$ from Assumption~\ref{asmp: closed loop} such that for %
    $N \geq \max\curly{c_1(\theta^\star, \pi_{\exp}, T, \alpha), c_2(\theta^\star, \pi_{\exp}, T) \varepsilon^{-2}}$, the following bound holds:
    \begin{align}\label{eq:minimax_lower_bounds}
         \mathcal{M}(\theta^\star, \varepsilon, \calA) \geq \frac{1}{4} \frac{\trace(H(\theta^\star) \mathsf{FI}^{\pi_{\exp}}(\theta^\star)^{-1}) }{N}.
    \end{align}
\end{theorem}

The main tool for establishing this result is the van Trees inequality~\citep{gill1995applications}, a Bayesian version of the Cramer-Rao lower bound, which lower bounds the squared error in estimating a smooth function of the unknown parameter from data. To apply this result, we exploit the fact that the excess cost is a quadratic function  of the Youla parameter for the estimated controller.  
The argument is detailed in \Cref{s: lower bound}. 

The theorem lower bounds the excess cost of learning the LQG controller in terms of the Hessian $H(\theta^\star)$ (defined in \eqref{eq: Hessian}) multiplied by the inverse of the Fisher information for the policy used to collect the dataset, $\mathsf{FI}^{\pi_{\exp}}(\theta^\star)$ (defined in \eqref{eq: FI}). The Hessian captures the sensitivity of the control synthesis problem to the unknown parameters, while the inverse of the Fisher information measures the noise-to-signal ratio and thus captures the ``difficulty" of estimating $\theta^\star$. The product of these quantities captures the fact that the excess cost due to learning a controller scales with the difficulty of the identification problem in the directions that are relevant to the cost. Note that the exploration policy should be \emph{persistently exciting}, that is, the Fisher Information should be invertible, in the directions where the Hessian is non-zero. %

\begin{remark}[Dependence on $T$, exploration policies]The Fisher Information matrix depends on the length of the exploration episodes $T$. If the exploration policy $\pi_{\exp}$ is linear and stabilizing, then it satisfies the smoothness requirement for $\mathsf{FI}^{\pi_{\exp}}(\theta)$ (follows by Riccati perturbation arguments, see, for example, Appendix B of \citep{simchowitz2020naive}). If, in addition, the exploration policy is persistently exciting at steady-state, that is, $\lim_T \mathsf{FI}^{\pi_{\exp}}(\theta)/T$ is positive definite, then  the Fisher Information grows linearly with $T$ asymptotically (see \Cref{prop: info}). Thus, in the case of persistently exciting linear stabilizing policies, the lower bound scales inversely proportionally to the total amount of data $N\times T$. Additionally the burn-in time reduces to $N 
\times T \geq \max\curly{c_1(\theta^\star, \pi_{\exp}, \alpha), c_2(\theta^\star, \pi_{\exp}) \varepsilon^{-2}}$ for problem dependent constants $c_1$ and $c_2$ that no longer depend on the episode horizon.
More general exploration policies (e.g. unstable policies, nonlinear policies) might also satisfy the smoothness requirement for $T<\infty$ but we do not have a systematic way of characterizing how their Fisher Information scales with $T$. 
\end{remark}

The lower bound is valid only after $N$ exceeds a burn-in time defined by the quantities $c_1(\theta^\star, \pi_{\exp}, T, \alpha)$, $c_2(\theta^\star, \pi_{\exp}, T)$, and $\varepsilon$. 
For sufficiently small $\varepsilon$, a trivial algorithm that ignores the data and simply returns a fixed robust controller can achieve uniformly 
bounded
excess cost over all instances in the ball around $\theta^\star$. 
When $N$ is larger than the burn-in time, the available data is abundant enough so that an efficient learning algorithm can outperform such trivial baselines, which is the regime we focus on.
Notably, the burn-in time depends on the instance of interest, the exploration policy, the length of the exploration episodes, and the value $\alpha$ from Assumption~\ref{asmp: closed loop}. The exact dependence of the burn-in time on these objects is not tight, as our focus is on the asymptotic behavior. The asymptotic characterization of the bound is tight, as shown in the following remark. %

\begin{remark}[Tightness of lower bound]
One can verify that the lower bound matches the asymptotic rate of decay for the model-based certainty equivalence algorithm. In particular, if the parameter $\hat \theta$ is estimated from the collected dataset via maximum likelihood as $\hat \theta =\argmin_{\theta\in\R^{d_{\theta}}} \mathsf{nll}(\calD, \theta)$, and if $\mathsf{FI}^{\pi_{\exp}}(\theta) \succ 0$, then by Theorem 9.1 of \citet{ljung1998system}
\begin{align*}
     \sqrt{N}(\hat \theta - \theta) \overset{d}{\to}\calN(0, \mathsf{FI}^{\pi_{\exp}}(\theta)^{-1}) \mbox{ as } N\to\infty. 
\end{align*}
If the controller is synthesized based on this estimate and as long as that controller is stabilizing, by a Taylor expansion
\begin{align*}
    J(K_{\hat \theta}, \theta) -J(K_{\theta}, \theta) = \frac{1}{2} \norm{\hat \theta  - \theta}_{H(\theta)}^2 + O\paren{\!\norm{\hat \theta-  \theta}^2\!}. 
\end{align*}
By combining these facts, one achieves the following asymptotic characterization of the excess cost
\begin{equation*}
\begin{aligned}
    \lim_{N\to\infty} \!N \E\brac{\!J(K_{\hat \theta}, \theta) \!-\! J(K_{\theta}, \theta) } =\frac{1}{2}\trace\paren{H(\theta) \mathsf{FI}^{\pi_{\exp}}(\theta)^{-1}}.
\end{aligned}
\end{equation*}
This implies that the result of Theorem~\ref{thm: excess cost lower bound} is asymptotically tight and certainty equivalence is asymptotically optimal. 
\end{remark}

The matching upper and lower bound suggest that to achieve sample efficient control, the exploration policy should be chosen such that $\mathsf{FI}^{\pi_{\exp}}$ is as large as possible in the directions where objective is sensitive to the parameter estimates. 
Practical limitations, such as actuator and energy bounds, typically prevent the exploration policy from injecting large inputs to achieve large signal-to-noise ratio. 
In such cases, the learner should choose an exploration policy that maximizes the relevant information subject to constraints on the experimental procedure. This decision reduces to the canonical problem of experiment design \citep{rivera2003plant, gevers2005identification, hjalmarsson2005experiment, rojas2007robust, hjalmarsson2009system, bombois2011optimal,chatzikiriakos2025high}.  The optimal experiment design formulation using these matching upper and lower bounds has been studied for the fully observed (and potentially nonlinear) setting by \citet{wagenmaker2021task, wagenmaker2023optimal, lee2024active}.

\Cref{thm: excess cost lower bound} recovers the lower bound from the fully observed setting \citep{wagenmaker2021task, wagenmaker2023optimal, lee2023fundamental}
 as a special case. 
In particular, the fundamental limits in the fully-observed setting are also defined by the product of the cost Hessian with the inverse Fisher Information. Consequently, understanding the hardness of the learning-enabled control induced by partial observability is not transparent from the above bound, and we are instead required to study the behavior of the the matrices $H(\theta)$ and $\mathsf{FI}^{\pi_{\exp}}(\theta)$. We provide characterizations of them in the sequel.

\section{Computing the Lower Bound}

We now provide characterizations for the Hessian and Fisher Information from \Cref{thm: excess cost lower bound}. To get an interpretable expression of the Fisher Information, we restrict attention to linear exploration policies with additive probing noise:
\begin{equation}
\label{eq: exploration policy}
\begin{aligned}
    x^{\exp}_{t+1} &= A_{\exp} x^{\exp}_t + B_{\exp} y_t \\
    u_t &= C_{\exp} x^{\exp}_t + D_{\exp}^y y_t + D_{\exp}^\eta \eta_t,
\end{aligned}
\end{equation}
where $\eta_t \in \R^{\du}$ is $i.i.d.$  mean zero Gaussian with identity covariance ($D_{\mathsf{exp}}^{\eta}$ can capture non-isotropic probing noise) and $x_t^{\exp}$ is the state of the exploration policy.

We also restrict our attention to the case of a scalar parameter ($d_{\theta}=1$) $\theta$. The extension of $\theta$ to a multi-dimensional parameter follows by the observation that if $U = \bmat{U_1 & \dots & U_{d_{\theta}}}$ is a $d_{\theta}$-dimensional orthonormal matrix such that $U^\top \mathsf{FI}^{\pi_{\exp}}(\theta) U$ is diagonal, then \begin{align*}
    \trace(H(\theta) \mathsf{FI}^\pi_{\exp}(\theta)^{-1}) = \sum_{i=1}^{d_{\theta}} \frac{U_i^\top H(\theta) U_i}{U_i^\top \mathsf{FI}^{\pi_{\exp}}(\theta) U_i}.
\end{align*} 
The scalar projections of the Hessian defined by $U_i^\top H(\theta) U_i$ and of the Fisher Information defined by $U_i^\top \mathsf{FI}^{\pi_{\exp}}(\theta) U_i$  are equivalent to the second derivatives of $f(t) \triangleq J(K_{\theta+t U_i}, \theta)$ and $g(t) \triangleq \mathsf{nll}(\calD; \theta + t U_i)$ with respect to $t$, at $t=0$.
We drop the dependence of system matrices on the parameter when it is clear from the context. We additionally denote the derivative of a system matrix with respect to the parameter evaluated at $\theta$ as $\dot A = \frac{d}{d\tilde \theta} A(\tilde \theta)\vert_{\tilde\theta=\theta}$. 

The characterizations for the Hessian and Fisher Information matrix rely upon the derivatives of the controller and observer gains with respect to the underlying parameter, which are shown by Lemma B.1 of \citet{simchowitz2020naive} to be %
\begin{align*}
    \dot F &= - \Psi^{-1}(\dot B^\top P A_{\mathsf{cl}}^c + B^\top P (\dot A + \dot B F )
    + B^\top \dot P A_{\mathsf{cl}}^c) \\
    \dot L &= \left(A_{\mathsf{cl}}^o \Sigma \dot C^\top + A_{\mathsf{cl}}^o \dot \Sigma  C^\top + (\dot A - L \dot C) \Sigma C^\top - L \dot \Sigma_v \right)\Sigma_e^{-1} \\
    \dot P &\!=\! \dlyap\paren{\!(A_{\mathsf{cl}}^c)^\top, \mathsf{sym}\paren{(\dot A + \dot B F) P A_{\mathsf{cl}}^c}\!} \\
    \dot \Sigma &\!=\! \dlyap\paren{A_{\mathsf{cl}}^o, \mathsf{sym}\paren{\! A_{\mathsf{cl}}^o \Sigma(\theta) (\dot A - L \dot C) } + \dot \Sigma_w + L \dot \Sigma_v L^\top \!}.
\end{align*}
Both the Hessian and the Fisher Information may now be expressed in terms of these quantities.

\begin{proposition}[Strictly Causal Hessian Characterization]
\label{prop: hessian}
    Consider the strictly causal setting. If $\theta$ is a scalar parameter, then $H = 2\trace(\Psi \bmat{\dot F & F} \Sigma_H \bmat{\dot F & F}^\top)$, where $\Sigma_H = \dlyap(A_{\mathsf{joint}}^H, B_{\mathsf{joint}}^H \Sigma_e (B_{\mathsf{joint}}^H)^\top)$ and
    \begin{align*}
    A_{\mathsf{joint}}^H &= \bmat{A_{cl}^c & 0 \\ \dot A - L \dot C + \dot B F & A_{cl}^o}, \quad
    B_{\mathsf{joint}}^H = \bmat{L \\ \dot L}. 
    \end{align*}
\end{proposition}
\begin{proof}
    By the performance difference lemma (Theorem 11.2 of \citep{soderstrom2012discrete}),
    \begin{align*}
        &J(K_{\tilde \theta}, \theta) - J(K_{\theta}, \theta) \\
        &= \limsup_{T\to\infty}\frac{1}{T}\bfE^{K_{\tilde \theta}}\brac{\sum_{t=0}^T \norm{F(\tilde\theta) \hat x_{t\vert t-1}^{\tilde \theta} - F \hat x_{t\vert t-1}^{ \theta}}_{\Psi}^2}.
    \end{align*}
    Taking the Hessian with respect to $\tilde \theta$ gives
    \begin{align*}
        H(\theta) &= \frac{d^2}{d\tilde\theta^2} J(K_{\tilde \theta}, \theta)\vert_{\tilde \theta = \theta} \\
        &= \!\limsup_{T\to\infty}\frac{1}{T}\bfE^{K_{\tilde \theta}}\!\brac{\sum_{t=0}^T \!\frac{d^2}{d\tilde\theta^2}\!\norm{F(\tilde\theta) \hat x_{t\vert t\!-\!1}^{\tilde \theta} \!-\! F \hat x_{t\vert t\!-\!1}^{ \theta}}_{\Psi}^2\!} 
        \!\bigg\vert_{\tilde \theta = \theta}\!\!.
    \end{align*}
    Note that $\hat x_{t|t-1}^\theta$ also implicitly depends on $\tilde \theta$ through $K(\tilde \theta)$, as the Kalman filter state depends upon the history of inputs and observations up to time $t-1$. 
    Observing that terms involving second order derivatives of $F$, $x_{t|t-1}^{\tilde \theta}$, and $x_{t|t-1}^{\theta}$ are multiplied by quantities that equal zero when evaluated at $\tilde\theta=\theta$, the summand of the above expression may be written
    \begin{align*}
        &\frac{d^2}{d\tilde\theta^2} \norm{  F(\tilde \theta) \hat x_{t\vert t-1}^{\tilde \theta} - F \hat x_{t\vert t-1}^{ \theta}}_{\Psi}^2\vert_{\tilde\theta=\theta} \\
        &= 2  \norm{\dot F \hat x_{t\vert t-1}^{\theta} + F \frac{d}{d\tilde \theta}(\hat x_{t|t-1}^{\tilde\theta} - \hat x_{t|t-1}^{\theta})\vert_{\tilde\theta = \theta}}_{\Psi}^2. 
    \end{align*}
    If we let $r_t^{\tilde \theta} \triangleq \hat x_{t|t-1}^{\tilde\theta} - \hat x_{t|t-1}^{\theta}$, then
    \begin{align*}
        r_{t+1}^{\tilde \theta} &= A_{cl}^o(\tilde \theta) \hat x_{t|t-1}^{\tilde\theta} -A_{cl}^o \hat x_{t|t-1}^{\theta} \\
        &+ (\bmat{L(\tilde \theta) & B(\tilde \theta)} - \bmat{L & B}) \bmat{y_t \\ u_t}. 
    \end{align*}
    Let $\dot r_t \triangleq  \frac{d}{d\tilde \theta} r_{t}^{\tilde \theta}\vert_{\tilde\theta=\theta}$. By the product rule, we can express $\dot r_{t+1}$ recursively as $
       \dot r_{t+1} = \dot A_{cl}^o \hat x_{t|t-1}^{ \theta} + A_{cl}^o\dot r_t+\bmat{\dot L & \dot B}\bmat{ y_t \\ u_t}.$ Note that $y_t$ and $u_t$ also depend on $\theta$, but their derivatives disappear because they are multiplied by $(\bmat{L(\tilde \theta) & B(\tilde \theta)} - \bmat{L & B})\vert_{\tilde \theta = \theta} = 0$. 
       
    The joint state $(\hat x_{t|t-1}^{\theta},\dot r_t)$ then evolves according to
        $\bmat{\hat x_{t+1|t}^{\theta} \\ \dot r_{t+1}} = A_{\mathsf{joint}}^H \bmat{\hat x_{t|t-1}^{\theta} \\ \dot r_t} + B_{\mathsf{joint}}^H e_t,$
    where $e_t \sim \calN(0, \Sigma_e)$ is the innovation sequence. The result follows by writing the Hessian in terms of the stationary covariance of this joint state. 
\end{proof}

The Fisher information can be expressed similarly by writing it as the expected Hessian of the negative log-likelihood. 

\begin{proposition}[Fisher Information Characterization]
\label{prop: info}
Let $\theta$ be a scalar parameter and consider a causal linear exploration policy \eqref{eq: exploration policy} which stabilizes the instance $\theta$. Then the Fisher Information is given by 
\begin{align*}
    \lim_{T} \frac{1}{T} \mathsf{FI}^{\pi_{\exp}}(\theta) &= \norm{\Sigma_{\mathsf{FI}}^{1/2} \bmat{\dot C^\top \\ 0 \\  C^\top} \Sigma_e^{-1/2}}_F^2 \\&+ \frac{1}{2} \trace((\Sigma_e^{-1} \dot \Sigma_e)^2), %
\end{align*}
where 
\begin{align*}
    \Sigma_{\mathsf{FI}} &= \dlyap\paren{A_{\mathsf{joint}}^{\mathsf{FI}}, B_{\mathsf{joint}}^{\mathsf{FI}} \bmat{\Sigma_e\\& I} (B_{\mathsf{joint}}^{\mathsf{FI}})^\top}, \\
    A_{\mathsf{joint}}^{\mathsf{FI}} &\!=\! \bmat{A+BD_{\exp}^y C & B C_{\exp} & 0 \\ B_{\exp} C & A_{\exp} & 0\\ \dot A - L \dot C+\dot B D_{\exp}^y C & \dot B C_{\exp} & A_{\mathsf{cl}}^o},\\
B_{\mathsf{joint}}^{\mathsf{FI}}&\!=\!\bmat{L+BD^y_{\exp} & B D_{\exp}^\eta \\ B_{\exp}  & 0 \\ \dot L+\dot BD^y_{\exp} & \dot B D_{\exp}^\eta}.
\end{align*}
\end{proposition}
\begin{proof}
    The Fisher information is the expected Hessian of the negative log-likelihood: 
    \begin{align*}
        &\mathsf{FI}^{\pi_{\exp}}(\theta) \\
        &= \E\brac{\nabla_{\theta}^2 \paren{\frac{1}{2}\sum_{t=1}^T \norm{y_t - \hat y_t}_{\Sigma_e^{-1}}^2 + \frac{T}{2} \log\abs{\Sigma_e}}}\\
        &= \sum_{t=1}^T\E \brac{D_\theta \hat y_t^\top \Sigma_e^{-1} D_{\theta} \hat y_t} \!+\! \frac{T}{2} \trace((\Sigma_e^{-1} \dot \Sigma_e)^2),
    \end{align*}
    where $\hat y_t = C \hat x_{t\vert t-1}^\theta$. Then by the product rule, $D_{\theta} \hat y_t = \dot C \hat x_{t\vert t-1}^\theta + C \frac{d}{d\theta} \hat x_{t\vert t-1}^\theta$. Taking the derivative of the Kalman filter state update, we arrive at a system with a joint state involving the Kalman filter estimate, the state of the exploration policy, and the derivative of the Kalman filter estimate with respect to $\theta$ (derived analogously to the proof of \Cref{prop: hessian}). This joint state $\xi_t$ evolves according to $\xi_{t+1} = A_{\mathsf{joint}}^{\mathsf{FI}} \xi_t + B_{\mathsf{joint}}^{\mathsf{FI}} \bmat{e_t \\ \eta_t}$, where $e_t \sim \calN(0, \Sigma_e)$ and $\eta_t \sim \calN(0, I)$ are independent across time and from each other. The Fisher Information matrix can be expressed asymptotically as $T\to\infty$ in terms of the stationary covariance of this joint state, resulting in the given characterization.
\end{proof}

For offline identification, the exploration policy need not have any relation to the optimal policy. However, if one considers playing an exploration policy which matches the optimal policy, and with $D_{\exp}^\eta$ potentially nonzero to inject a probing signal, then the structure of the information matrix simplifies, and closely matches that of the Hessian.\footnote{This exploration policy is approximately the choice of greedy exploitation commonly used in online settings, where one wants to maintain a small control cost throughout the learning process \citep{simchowitz2020naive}. %
}

\begin{corollary}[Fisher Information Under Optimal Strictly Causal Policy]
\label{Cor: info under optimal policy}
Let $\theta$ be a scalar parameter. Suppose that the experiment policy \eqref{eq: exploration policy} has $A_{\exp} = A - LC + B F$, $B_{\exp} = L$, $C_{\exp} = F$, $D_{\exp}^y = 0$. Then the Fisher information is characterized as in \Cref{prop: info} with the substitution 
\begin{align*}
    A_{\mathsf{joint}}^{\mathsf{FI}} &= \bmat{A_{\mathsf{cl}}^c & 0 \\ \dot A - L \dot C + \dot B F & A_{\mathsf{cl}}^o}, \quad B_{\mathsf{joint}}^{\mathsf{FI}}&=\bmat{L & B D_{\exp}^\eta \\ \dot L & \dot B D_{\exp}^\eta},
\end{align*}
and with $\bmat{\dot C & C}$ replacing the matrix $\bmat{\dot C & 0 & C}$. 
\end{corollary}

\ifTACmode\else
    \section{Recovering the LQR Setting}
\label{s: lqr setting}
The above characterizations do not immediately recover results from the LQR setting from \citep{lee2023fundamental}, as they assume strictly causal policies. However, the characterizations can also be presented in the causal setting.  
\begin{proposition}[Causal Hessian Characterization]
    Consider the causal setting, and let $\theta$ be a scalar parameter. The Hessian characterization is identical to \Cref{prop: hessian}, with the substitutions:
    \begin{align*}
        \bar L &\triangleq  \Sigma C^\top \Sigma_e^{-1},\,\, I_{\bar LC} \triangleq I - \bar LC\\
         \dot {\bar L} &\triangleq  \paren{I_{\bar LC} \Sigma \dot C^\top + I_{\bar LC} \dot \Sigma  C^\top  - \bar L \dot C \Sigma C^\top} \Sigma_e^{-1} 
        \end{align*}
        and \begin{align*} 
        A_{\mathsf{joint}}^H &\gets \bmat{A+BF & 0 \\ I_{\bar LC}(\dot A + \dot B F) - \bar L \dot C (A+BF) & I_{\bar LC} A} \\
        B_{\mathsf{joint}}^H &\gets \bmat{\bar L \\ \dot {\bar{L}}}
    \end{align*}
\end{proposition}
The proof is the same as the one of~\Cref{prop: hessian}. 
Notably, we recover the fully observed setting of \citep{lee2023fundamental} from the above proposition by letting $C=I$ and taking $\Sigma_v \to 0$. Then, the Hessian becomes $H = 2 \trace\paren{\Psi \dot F \dlyap(A+BF, \Sigma_w) \dot F^\top}$. In \citep{lee2023fundamental}, this characterization was used to derive examples of challenging fully observed problems, including an example where the lower bound on the minimax excess cost is exponential in the state dimension. In the sequel, we present examples of challenging problems enabled by characterizations of the partially observed setting.

\fi
\ifTACmode
    \section{Examples of challenging problems}

Using the lower bound of \Cref{thm: excess cost lower bound} along with the characterization of the Hessian in \Cref{prop: hessian} and the Fisher Information in \Cref{prop: info}, we are now able to examine the fundamental difficulty of learning the LQG controller for particular examples. We begin by examining several cases where varying a single parameter can make the problem arbitrarily fragile. The first of these, in \Cref{s: doyle example}, is inspired by Doyle's LQG example without guaranteed margins \cite{doyle1978guaranteed}. Another, in \Cref{s: nmp}, is a non-minimum phase system. We additionally present an example where there is a design parameter of the system that induces a tradeoff between the cost of control and the excess cost of learning to control in \Cref{s: tradeoffs}.

We consider strictly causal policies in all examples. We study learning performance as we vary certain (known) hyper-parameters, denoted by $\sigma,\xi$ or $s$ to distinguish them from the unknown parameter $\theta$. Proof details are deferred to the supplementary material.

\else
    \section{Examples of challenging problems}

Using the lower bound of \Cref{thm: excess cost lower bound} along with the characterization of the Hessian in \Cref{prop: hessian} and the Fisher Information in \Cref{prop: info}, we are now able to examine the fundamental difficulty of learning the LQG controller for particular examples. We begin by examining several cases where varying a single parameter can make the problem arbitrarily fragile. The first of these, in \Cref{s: doyle example}, is inspired by Doyle's LQG example without guaranteed margins \cite{doyle1978guaranteed}. Two others, in \Cref{s: nmp} and \Cref{s: compounding}, are non-minimum phase systems. We additionally present an example where there is a design parameter of the system that induces a tradeoff between the cost of control and the excess cost of learning to control in \Cref{s: tradeoffs}. Finally, we conclude with a discussion of model-based design with misspecification in \Cref{s: misspecified}. 

We consider strictly causal policies in all examples. We study learning performance as we vary certain (known) hyper-parameters, denoted by $\sigma,\xi$ or $s$ to distinguish them from the unknown parameter $\theta$. Proof details are deferred to the supplementary material.

\fi

\subsection{No margins example}
\label{s: doyle example}

We first consider a discrete time LQG instance inspired by the example from \cite{doyle1978guaranteed}, the first demonstration that a continuous time LQG controller need not be inherently robust. 
Specifically, we consider the LQG instance 
\begin{equation}
\label{eq: doyle example}
\begin{aligned}
    A(\theta) = \bmat{2 & 1 \\ 0 & 2}, B(\theta) = \bmat{0 \\ \theta}, C(\theta)= \bmat{1 & 0}.
\end{aligned}
\end{equation}
with $\Sigma_w(\theta) = Q = \bmat{1 & 1 \\ 1 & 1}$ and $\Sigma_v(\theta) = R= \sigma$, where $\sigma>0$ is a variable hyper-parameter. Suppose that the nominal parameter $\theta^\star=1$. Similarly to~\cite{doyle1978guaranteed}, we investigate the fragility of LQG as parameter $\sigma$ approaches zero.

The control penalty and noise distributions for this example are designed such that both the controller and observer exhibit a near deadbeat response in one direction, resulting in a pole near zero. However, the remaining pole is pushed to the unit circle. Specifically, both the closed-loop state-feedback and the observer error dynamics have an eigenvalue near $1-\sqrt{\sigma}$. 

The matrix $A_{\mathsf{joint}}^H$ defining the Hessian characterization of \Cref{prop: hessian} can be expressed as
\ifTACmode
$A_{\mathsf{joint}}^H=\hat A_{\mathsf{joint}}^H+X$,
\else
\begin{align}
    \label{eq: Ajoint decomp}
    A_{\mathsf{joint}}^H &= \hat A_{\mathsf{joint}}^H + X,
 \end{align}
\fi
where
\begin{align}
    \label{eq: hatAjointH}
    \hat A_{\mathsf{joint}}^H &= \bmat{2  &1 \\ 
    -2 - 2\sqrt{\sigma} & -1 - \sqrt{\sigma} \\
    0 & 0&- 1 - \sqrt{\sigma} & 1 \\ 
    -2 - 2\sqrt{\sigma} & -3 -\sqrt{\sigma} &-2 - 2\sqrt{\sigma} & 2},
\end{align}
and $X = O(\sigma)$.
The Jordan decomposition of $\hat A_{\mathsf{joint}}^H$ is given by $V J V^{-1}$, with
\begin{align}
\label{eq: jordan}
\begin{aligned}
    J &= \bmat{1 - \sqrt{\sigma} & 1 & 0 & 0 \\ 0 & 1- \sqrt{\sigma} &0 & 0\\ 0&0 &0 &1\\ 0 &0 &0 &0}
\end{aligned},
\end{align}
and $V$ is a matrix of generalized eigenvectors. Using this Jordan form, we can compute the Hessian analytically.
\begin{proposition}
    \label{prop: doyle hessian}
    For the LQG instance defined by \eqref{eq: doyle example}, $H(\theta^\star) =  2048 \sigma^{-3/2} + O(\sigma^{-1})$.
\end{proposition}

The $\sigma^{-3/2}$ rate of growth arises due to the Jordan block $\bmat{1 -\sqrt{\sigma} & 1 \\ 0 & 1-\sqrt{\sigma}}$. This rate of growth is faster than the $\sigma^{-1/2}$ rate of growth for a Lyapunov function defined solely in terms of the closed loop matrices under the controller or the observer, $\dlyap(A_{\mathsf{cl}}^c
, Q)$ or $\dlyap(A_{\mathsf{cl}}^o, Q)$, in isolation. In particular, the appearance of the Jordan block arises because the uncertainty in the input channel amplifies the sensitivity of the observer and the controller in the closed-loop response.

We have demonstrated that the Hessian for the example \eqref{eq: doyle example} scales at a rate $\sigma^{-3/2}$. The other term determining the excess cost of the learning problem is the Fisher Information matrix under the experiment policy used to collect the identification dataset.
Consider using a fixed stabilizing policy defined by $A_{\exp}, B_{\exp}, C_{\exp}$ with $D^y_{\exp}=0,D^{\eta}_{\exp}=0$. %
Then the Fisher information converges to a positive constant as $\sigma \to 0$. Consequently, the excess cost due to the use of learned dynamics scales as $\trace(H(\theta^\star) \mathsf{FI}^{\pi_{\exp}}(\theta^\star)^{-1}) \propto \frac{1}{N \sigma^{3/2}}$. Notably, this occurs despite the fact that the optimal cost of control converges to a constant $c$: $J(K_{\theta^\star}, \theta^\star) \to c$, as $\sigma\to 0$.

The example highlights the potential fragility of the LQG problem under control penalties that encourage very aggressive controllers. Here, such controllers lead to nearly marginally stable closed-loop systems that are extremely fragile. The behavior may be counterintuitive: the problem becomes very sample-inefficient when the sensor noise level $\Sigma_v$ is \emph{decreased}. Performance can be recovered by adding ``ficticious noise'' so the controller has robustness to uncertainty \citep{doyle2003robustness}. Such a fix would sacrifice asymptotic performance on the given LQG objective as data becomes abundant.

The example is contrived, without immediate motivating physical applications. However, similar sensitivities are observed more generally in systems which have a high gain in some directions, but not others, i.e., systems with large condition numbers \citep{skogestad1988robust}. The fragility that arises due to the LQG objective motivates alternative objectives, such as $\calH_\infty$ optimal control \citep{bacsar2008h}, or alternative design procedures, such as Internal Model Control \citep{morari1989robust} which can enable more transparent tuning of controller parameters than the cost and noise matrices, $Q$, $R$, $\Sigma_w$ and $\Sigma_v$, of the LQG objective.

\ifTACmode\else
    
Two other variants of the example \eqref{eq: doyle example} which exhibit similar fragility are discussed below. 

\subsubsection{Stable Open-Loop}
  There exist stable variants of this example that can exhibit similar sensitivity. Consider the LQG instance 
    \begin{align*}
        A(\theta) &= \bmat{-0.5 & 0.5 \\ 0 & -0.5}, B(\theta) = \bmat{\theta \\ \theta}, C(\theta) = \bmat{1 & 1}, \\
        Q &= \bmat{1 & 0 \\ 0 & 0}, \Sigma_w(\theta) = \bmat{0 & 0 \\ 0 & 1}, R = \Sigma_v(\theta) = \sigma,
    \end{align*}
    with the nominal parameter $\theta^\star=1$.
    Similar to the LQG instance \eqref{eq: doyle example}, the optimal controller and observer force one eigenvalue of the closed-loop transition matrix to zero, while the other approaches $1-\sqrt{\sigma}$. As with the LQG instance \eqref{eq: doyle example}, the stable variant can become arbitrarily fragile as $\sigma\to 0$. We numerically compute the Hessian of \Cref{prop: hessian} for this stable variant to be: $H(\theta^\star) = 2.53 \sigma^{-3/2} + O(\sigma^{-1})$ obtaining a rate similar to \eqref{eq: doyle example}. 

\subsubsection{Fully-Observed}
    The example constructed by \citet{doyle1978guaranteed} was used to demonstrate that the continuous time LQG controller lacks the gauranteed margins possessed by continuous time LQR.  However, in discrete time there are no guaranteed uniform margins even in the fully observed setting. Consider the LQR controller defined by the system $(A,B,Q,R)$ from \eqref{eq: doyle example}. For this example, it holds that $ \rho(A+ (1+m) \times BF) > 1$ for $m \geq 1 + 2\sqrt{\sigma}$, indicating the upper gain margin is less than  $2\sqrt{\sigma}$, which can become arbitrarily small as $\sigma\to 0$. This sensitivity can be reflected in the lower bound of \Cref{thm: excess cost lower bound} by computing the Hessian and Fisher information for this fully observed example. %
    Following the same derivations we used to compute these objects for the LQG instance \eqref{eq: doyle example}, we numerically compute for this fully observed example that the Hessian is given by $H(\theta^\star) = \frac{9}{\sqrt{\sigma}} + O(\sigma^{-1})$. %

\fi

\subsection{Non-minimum phase dynamics example}
\label{s: nmp}
We now consider an example of a non-minimum phase system. Such systems often exhibit behavior known as inverse response, in which the response to a step input initially moves away from the eventual setting point. This makes non-minimum phase systems notoriously hard to control, with well established fundamental limits characterizing the fragility of any controller to uncertainty in the dynamics \citep{doyle2013feedback}.

We examine the impact of this behavior on the following instance of learning the LQG: 
\begin{equation}
A(\theta)=\begin{bmatrix}
1& 1\\\theta&1
\end{bmatrix},\,B(\theta)=\begin{bmatrix}
0\\1
\end{bmatrix},\,C(\theta)=\begin{bmatrix}
-\nmpparam &1
\end{bmatrix},
\label{eq: nmp example}
\end{equation}
with $\Sigma_w(\theta)=I$, $\Sigma_v(\theta)=1$, $Q=I$, $R=1$, and nominal parameter $\theta^{\star}=0$. Let $\nmpparam>0$ be a variable hyper-parameter.
The nominal system has a non-minimum phase zero at $1+\nmpparam$ and two poles at $1$. We analytically investigate what happens as the zero approaches the poles, that is, $\xi$ goes to $0$. 
\begin{lemma}[Difficulty of non-minimum phase example]
    \label{lem: nmp example}
    Consider the LQG instance defined by \eqref{eq: nmp example} and let the exploration policy be equal to the optimal strictly causal LQG policy. Then,
    \begin{align*}
    H(\theta^\star)=\Omega(\nmpparam^{-7}),\,
        \lim_{T\rightarrow \infty}\tfrac{1}{T} \mathsf{FI}^{\pi_{\exp}}(\theta^\star)=O(\nmpparam^{-3}) 
    \end{align*}
    and $J^\star=\min_{K\in\calK}J(K,\theta^\star)=O(\nmpparam^{-3})$.
\end{lemma}
 Note that as the zero approaches the pole, the system becomes very hard to observe. This, in turn, leads to a very large Hessian that scales with $\nmpparam^{-7}$ and a very large nominal cost that scales with $\nmpparam^{-3}$. The Fisher information is also very large and of the order of $\nmpparam^{-3}$. However, the Fisher information is not large enough to counteract the Hessian. 
 As a result, the lower bound scales as %
    \begin{align*}
        \trace(H(\theta^\star) \mathsf{FI}^{\pi_{\exp}}(\theta^\star)^{-1}) &\propto \frac{\nmpparam^{-4}}{N}.
    \end{align*}
 This shows that the excess cost %
 can be very large and even grow unbounded as $\xi \to 0$. 
 However, this is somewhat unsurprising, as observability is lost in this case. It is more interesting to note that the excess cost can be unbounded \emph{relative to  the nominal cost}:
 \begin{align*}
     \frac{\trace(H(\theta^\star) \mathsf{FI}^{\pi_{\exp}}(\theta^\star)^{-1})}{J^\star} & \propto \frac{\nmpparam^{-1}}{N}.
 \end{align*}
 Thus,  the cost of learning dominates the cost of control.

\ifTACmode\else
    
\subsection{Compounding impact of high sensitivity and low identifiability}
\label{s: compounding}

The previous two examples contained cases where varying a hyper-parameter of the system increased the sensitivity of control synthesis to error in the estimation of the parameter describing the dynamics, as characterized by the Hessian matrix, $H(\theta^\star)$. In these cases, the corresponding Fisher Information either stayed constant or increased with the design parameter. The challenge of learning enabled-control thus arose due to a sensitivity that grew faster than the Fisher Information. We now examine an instance for learning the LQG in which a design parameter increases the sensitivity of the synthesis problem, while also \emph{decreasing} the Fisher Information. This leads to a compounding impact of a very sensitive synthesis problem with dynamics which are challenging to identify. 

We consider the instance of learning the LQG defined by 
\begin{equation}
\label{eq: compouding}
\begin{aligned}
    A(\theta) = \bmat{\theta& 1 \\ 0 & 1}, B(\theta) =\bmat{0 \\ 1}, C(\theta) = \bmat{1 & s},
\end{aligned}
\end{equation}
with $\Sigma_w(\theta) = \bmat{1 & 1 \\ 1 & 1}$, $\Sigma_v(\theta) = 1$, $Q = I$, $R = 1$, and the nominal parameter $\theta^\star=1$. Let $s>0$ be a hyper-parameter.

We study experiment policies that are equal to the optimal causal LQG policy, with varying levels of noise injected, with $D^{\eta}_{\exp} = \eta I$. In \Cref{tab: compounding}, we numerically compute the asymptotic behavior of the optimal control cost, the Hessian, and the Fisher Information with different levels of probing noise as the  hyper-parameter $s$ becomes large; 
we rely upon numerical evaluation of the characterizations in \Cref{prop: hessian} and \Cref{prop: info}. 

\begin{table}[h]
\centering
\begin{tabular}{|c|c|c|c|c|c|}
\hline
$J^\star$ & $H$ & $\mathsf{FI}^{\pi_{\exp}^0}$ & $\mathsf{FI}^{\pi_{\exp}^1}$  & $\mathsf{FI}^{\pi_{\exp}^{10}}$ & $\mathsf{FI}^{\pi_{\exp}^{100}}$  \\ \hline
$28$ & $86s$   & $18.6s^{-2}$  & $20.4s^{-2}$    & $194s^{-2}$    & $17600s^{-2}$    \\ \hline
\end{tabular}
\caption{Asymptotic behavior of the optimal cost $J^\star= J(K_{\theta^\star}, \theta^\star)$, Hessian $H = H(\theta^\star)$, and Fisher Information $\mathsf{FI}^{\pi_{\exp}} = \mathsf{FI}^{\pi_{\exp}}(\theta^\star)$ for the instance \eqref{eq: compouding} as the design parameter $s \to \infty$. The exploration policy $\pi_{\exp}^\eta$ denotes the optimal LQG policy with probing noise injected as $D^\eta_{\exp}=\eta I$. } 
\label{tab: compounding}
\end{table}

From \Cref{tab: compounding}, we observe that as $s$ increases, the Hessian grows unbounded (with a rate if $O(s)$), while at the same time the Fisher Information goes to zero (with a rate of $O(s^{-2})$) independently of the probing noise level. This occurs despite the convergence of the optimal control cost to a constant.
Increasing the probing noise mitigates but does not eliminate the identifiability issue, that is, the Fisher Information still goes to zero.

\fi
\ifTACmode
    
\subsection{Tradeoffs}
\label{s: tradeoffs}
The above examples illustrated situations where varying a hyper-parameter towards certain values can lead to challenging learning problems. In this example, we change perspective and treat the hyper-parameter as a design variable. We then investigate how to optimally select this design variable to balance the cost of learning against the cost of control.
\begin{figure*}
    \centering
    \begin{subfigure}[t]{0.24\textwidth}
        \centering
        \includegraphics[width=\textwidth]{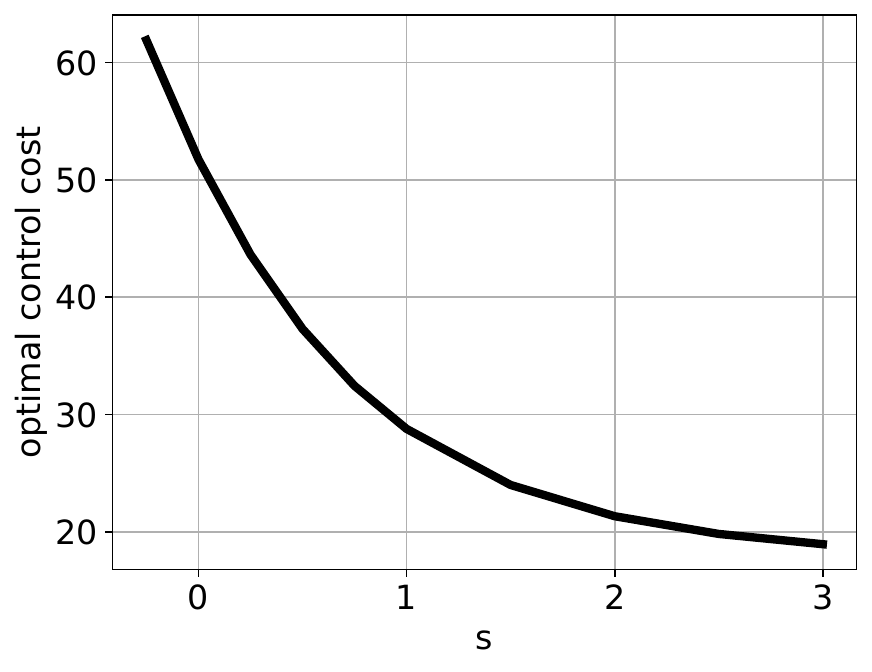}
        \caption{Optimal control cost}
        \label{fig:tradeoffs_cost}
    \end{subfigure}
    \hfill
    \begin{subfigure}[t]{0.24\textwidth}
        \centering
        \includegraphics[width=\textwidth]{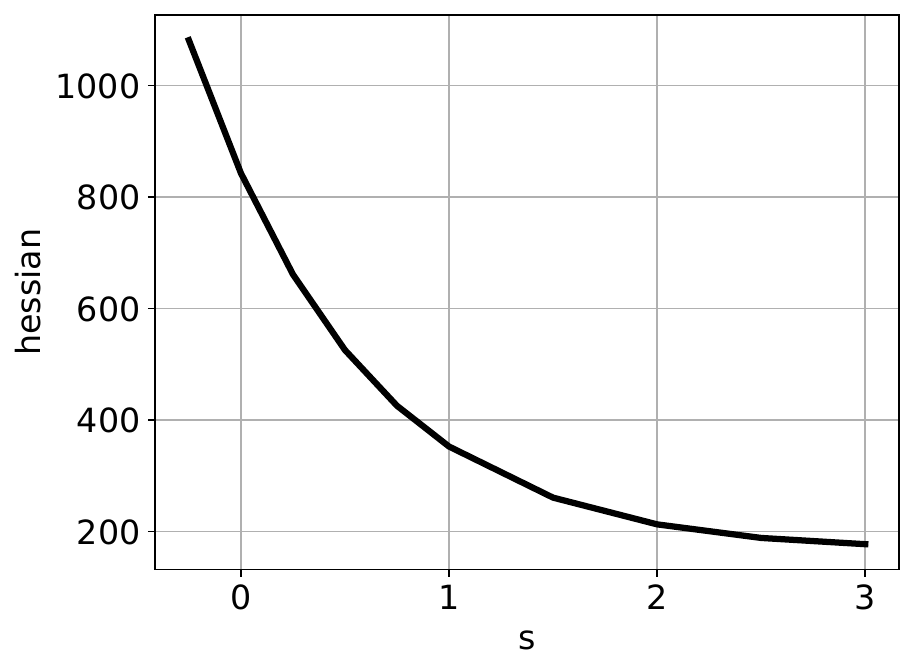}
        \caption{Hessian}
        \label{fig:tradeoffs_hessian}
    \end{subfigure}
    \hfill
	    \begin{subfigure}[t]{0.24\textwidth}
	        \centering
	        \includegraphics[width=\textwidth]{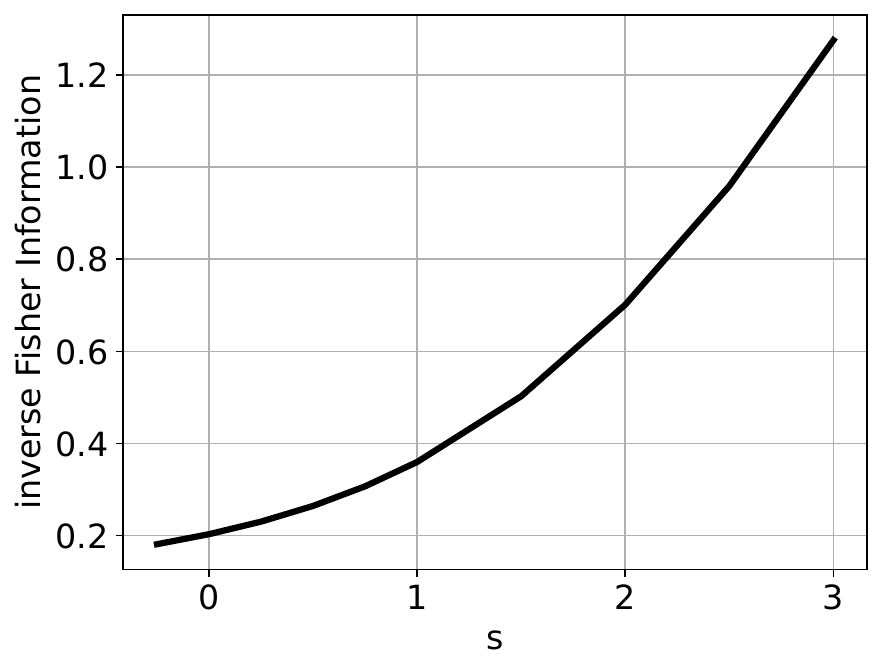}
	        \caption{Inverse asymptotic Fisher info.}
	        \label{fig:tradeoffs_fi}
	    \end{subfigure}
     \begin{subfigure}[t]{0.24\textwidth}
        \centering
        \includegraphics[width=\textwidth]{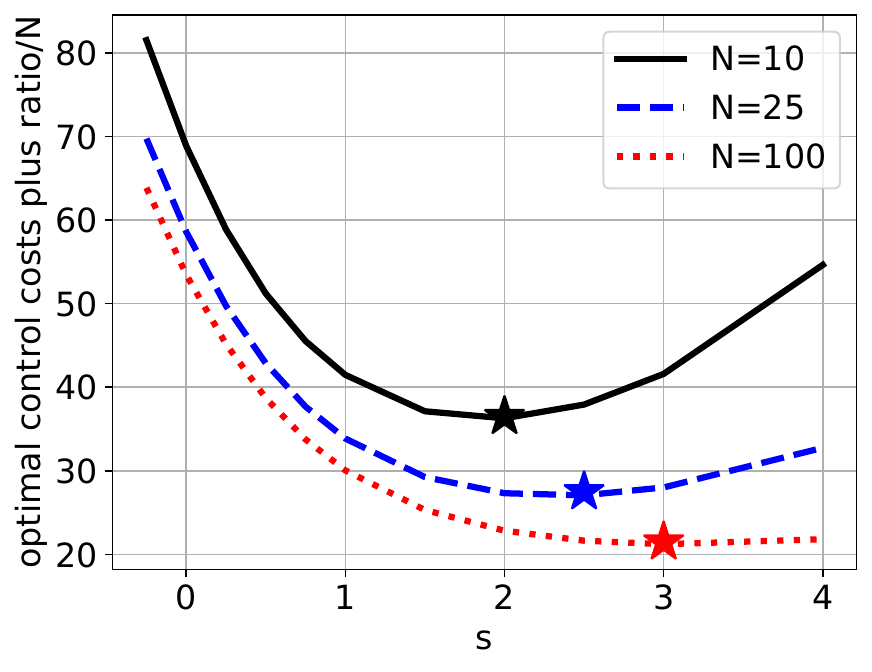}
        \caption{Co-design objective of \eqref{eq: co-design}}
        \label{fig: tradeoffs total cost}
     \end{subfigure}
	    \caption{\textbf{a-c) } Optimal cost, Hessian, and inverse asymptotic Fisher information for \eqref{eq: tradeoffs}. \textbf{d) } Co-design objective \eqref{eq: co-design} for various amounts of data $N$. Stars denote minimizers.}
	    \label{fig:tradeoffs}
        \vspace{-1.25em}
\end{figure*}
 In particular, consider the LQG instance defined by
\begin{equation}
    \label{eq: tradeoffs}
    \begin{aligned}
    A(\theta) = \bmat{ \theta & 1 \\ 0 & 1}, B(\theta) = \bmat{0\\ 1}, C(\theta) = \bmat{1 & s},
\end{aligned}
\end{equation}
with $Q = \Sigma_w(\theta) = I; R = \Sigma_v(\theta)= 1$ and nominal parameter $\theta^\star=1$. The exploration policy is taken to be the LQG policy with additional probing noise defined by $D^{\eta}_{\exp} = I$. The hyper-parameter $s$ can now be viewed as a parameter that determines the design of the sensor. %

In \Cref{fig:tradeoffs}, we examine the optimal control cost, cost Hessian, and inverse asymptotic Fisher Information as the sensor parameter $s$ varies from $0$ to $3$. As $s$ increases, the optimal controller incurs lower cost and becomes less sensitive to the unknown parameter.
However, the learner gathers information about the unknown parameter more slowly (the inverse Fisher Information increases) as $s$ increases. This induces a tradeoff between control and identification. We may therefore ask: where is the sweet spot for this design parameter, minimizing control cost while enabling sample-efficient identification? Our asymptotic lower bound, and the matching upper bound for certainty equivalence, suggest that the learner asymptotically suffers average cost scaling as
\begin{align*}
    J(K_{\theta^\star}, \theta^\star) + \frac{\trace(H(\theta^\star) \mathsf{FI}^{\pi_{\exp}}(\theta^\star)^{-1})}{N}.
\end{align*}
The sensor design parameter $s$ affects all three quantities that occur in the above expression. 
The problem of optimal sensor design could therefore be posed as
\begin{align}
    \label{eq: co-design}
    \min_{s} J(K_{\theta^\star}(s), \theta^\star, s) + \frac{\trace(H(\theta^\star,s ) \mathsf{FI}^{\pi_{\exp}}(\theta^\star, s)^{-1})}{N}.
\end{align}
In \Cref{fig: tradeoffs total cost}, we show how the objective~\eqref{eq: co-design} behaves for the LQG problem defined by \eqref{eq: tradeoffs} for various amounts of data. The star locations denote the corresponding minimizers of the co-design objective. As the amount of data available to the learner increases, the learner should increase the preferred value for the parameter $s$. 

There is a fundamental design challenge that we do not resolve here and defer to future work. To obtain the optimal system design in~\eqref{eq: co-design}, the learner would need access to the true, but unknown, parameter $\theta^\star$. In practical settings, this difficulty can be mitigated by selecting a worst-case value of the unknown parameter from an appropriate set.

\else
    
\subsection{Tradeoffs}
\label{s: tradeoffs}
The above examples illustrated situations where varying a hyper-parameter towards certain values can lead to challenging learning problems. In this example, we change perspective and treat the hyper-parameter as a design variable. We then investigate how to optimally select this design variable to balance the cost of learning against the cost of control.
\begin{figure*}
    \centering
    \begin{subfigure}[t]{0.24\textwidth}
        \centering
        \includegraphics[width=\textwidth]{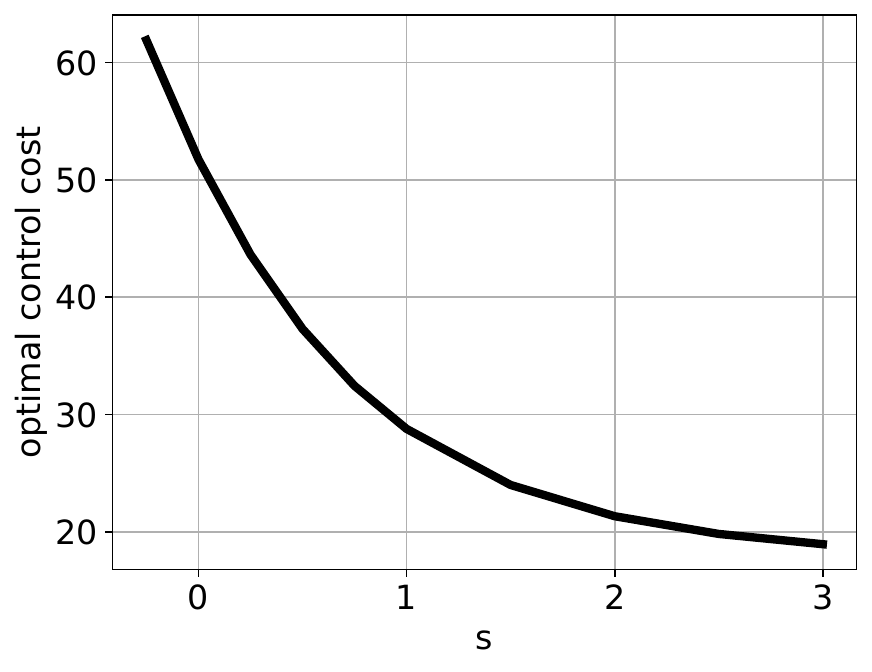}
        \caption{Optimal control cost}
        \label{fig:tradeoffs_cost}
    \end{subfigure}
    \hfill
    \begin{subfigure}[t]{0.24\textwidth}
        \centering
        \includegraphics[width=\textwidth]{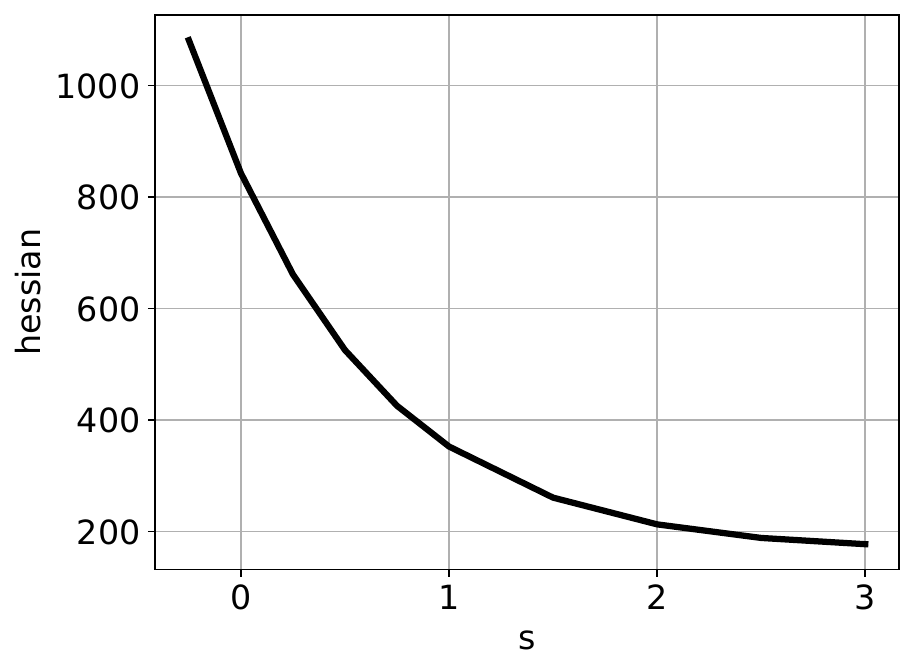}
        \caption{Hessian}
        \label{fig:tradeoffs_hessian}
    \end{subfigure}
    \hfill
	    \begin{subfigure}[t]{0.24\textwidth}
	        \centering
	        \includegraphics[width=\textwidth]{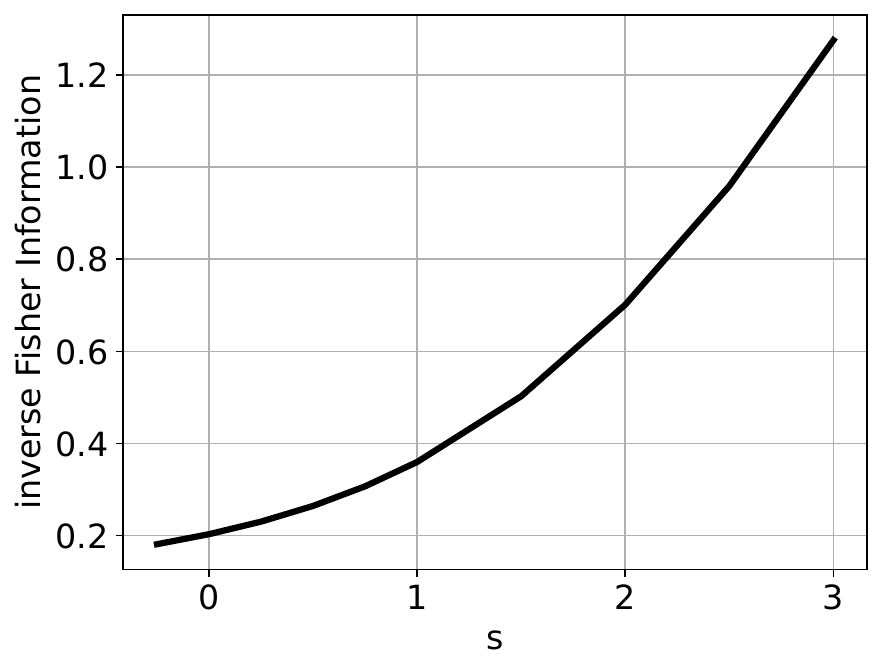}
	        \caption{Inverse asymptotic Fisher info.}
	        \label{fig:tradeoffs_fi}
	    \end{subfigure}
     \begin{subfigure}[t]{0.24\textwidth}
        \centering
        \includegraphics[width=\textwidth]{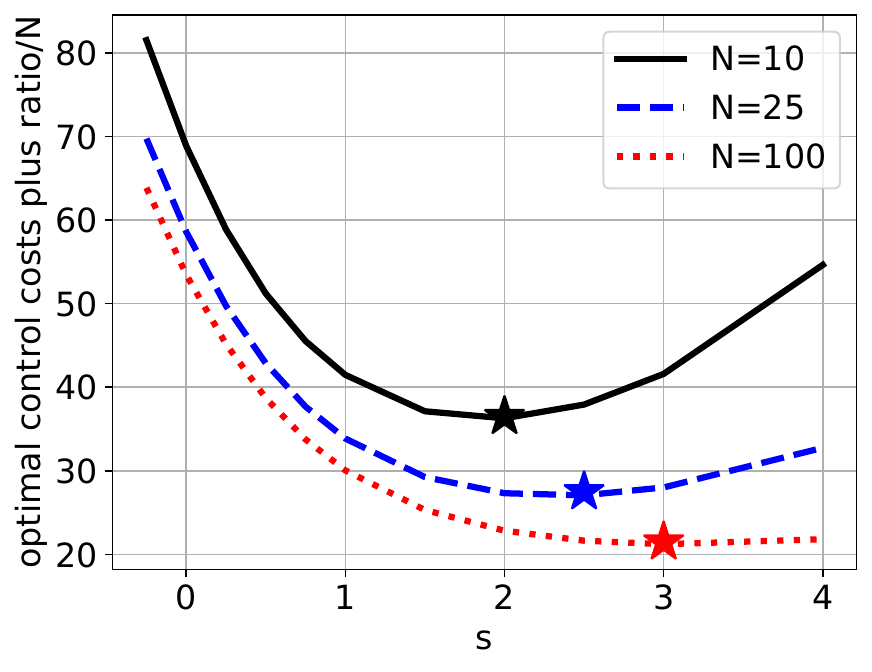}
        \caption{Co-design objective of \eqref{eq: co-design}}
        \label{fig: tradeoffs total cost}
     \end{subfigure}
	    \caption{\textbf{a-c) } Optimal cost, Hessian, and inverse asymptotic Fisher information for \eqref{eq: tradeoffs}. \textbf{d) } Co-design objective \eqref{eq: co-design} for various amounts of data $N$. Stars denote minimizers.}
	    \label{fig:tradeoffs}
        \vspace{-1.25em}
\end{figure*}
 In particular, consider the LQG instance defined by
\begin{equation}
    \label{eq: tradeoffs}
    \begin{aligned}
    A(\theta) = \bmat{ \theta & 1 \\ 0 & 1}, B(\theta) = \bmat{0\\ 1}, C(\theta) = \bmat{1 & s},
\end{aligned}
\end{equation}
with $Q = \Sigma_w(\theta) = I; R = \Sigma_v(\theta)= 1$ and nominal parameter $\theta^\star=1$. The exploration policy is taken to be the LQG policy with additional probing noise defined by $D^{\eta}_{\exp} = I$. The hyper-parameter $s$ can now be viewed as a parameter that determines the design of the sensor. %

In \Cref{fig:tradeoffs}, we examine the optimal control cost, cost Hessian, and inverse asymptotic Fisher Information as the sensor parameter $s$ varies from $0$ to $3$. As $s$ increases, the optimal controller incurs lower cost and becomes less sensitive to the unknown parameter.\footnote{Note the different process noise and range of $s$ from \Cref{s: compounding}.}
However, the learner gathers information about the unknown parameter more slowly (the inverse Fisher Information increases) as $s$ increases. This induces a tradeoff between control and identification. We may therefore ask: where is the sweet spot for this design parameter, minimizing control cost while enabling sample-efficient identification? Our asymptotic lower bound, and the matching upper bound for certainty equivalence, suggest that the learner asymptotically suffers average cost scaling as
\begin{align*}
    J(K_{\theta^\star}, \theta^\star) + \frac{\trace(H(\theta^\star) \mathsf{FI}^{\pi_{\exp}}(\theta^\star)^{-1})}{N}.
\end{align*}
The sensor design parameter $s$ affects all three quantities that occur in the above expression. 
The problem of optimal sensor design could therefore be posed as
\begin{align}
    \label{eq: co-design}
    \min_{s} J(K_{\theta^\star}(s), \theta^\star, s) + \frac{\trace(H(\theta^\star,s ) \mathsf{FI}^{\pi_{\exp}}(\theta^\star, s)^{-1})}{N}.
\end{align}
In \Cref{fig: tradeoffs total cost}, we show how the objective~\eqref{eq: co-design} behaves for the LQG problem defined by \eqref{eq: tradeoffs} for various amounts of data. The star locations denote the corresponding minimizers of the co-design objective. As the amount of data available to the learner increases, the learner should increase the preferred value for the parameter $s$. 

There is a fundamental design challenge that we do not resolve here and defer to future work. To obtain the optimal system design in~\eqref{eq: co-design}, the learner would need access to the true, but unknown, parameter $\theta^\star$. In practical settings, this difficulty can be mitigated by selecting a worst-case value of the unknown parameter from an appropriate set.

\fi
\ifTACmode\else
    
\subsection{Loss of Persistent Excitation}

The lower bounds and the above examples highlight the importance of the exploration policy. In particular, for the excess cost to decay to zero with an increasing number of experiments, the Fisher Information induced by the exploration policy must be positive in the directions where the Hessian is positive. Here, we review a situation where this is not the case. In particular, consider the fully observed setting by choosing $C(\theta) = I$ and $\Sigma_v(\theta) \to 0$.

Consider playing a static exploration policy defined by the LQR gain $A_{\exp} = 0$, $B_{\exp} =0$, $C_{\exp} = 0$, $D_{\exp}^y = \tilde F$ and $D_{\exp}^\eta =0$. In this case, the Fisher information of \Cref{prop: hessian} is 
\begin{align*}
    &\mathsf{FI}^{\pi_{\exp}} =\\
    &\trace\paren{\!(\!\dot A \!+\! \dot B \tilde F)\dlyap(A\!+\!B \tilde F, \Sigma_w)\!(\dot A \!+\! \dot B \tilde F)^\top \!\Sigma_w^{-1}\! }
\end{align*}
and $\mathsf{FI}^{\pi_{\exp}} = 0$ whenever $\dot A+ \dot B \tilde F  =0$. One scenario that this can occur is when $\tilde F = F + \Delta F$ with $\dot A = -\dot B F$ and $\dot B \Delta F = 0$.  

 By the discussion of \Cref{s: lqr setting}, the Hessian in the fully observed setting is given by 
\[
    H(\theta) = 2 \trace\paren{\Psi \dot F \dlyap(A+BF, \Sigma_w) \dot F^\top}.
\]
This Hessian can be lower bounded as $H(\theta) \geq \lambda_{\min}(R) \lambda_{\min}(\Sigma_w) 2 \norm{\dot F}_F^2$. Substituting $\dot A + \dot B F = 0$, leads to $\dot F = - \Psi^{-1} \dot B P A_{\mathsf{cl}}^c$, which can be nonzero for particular choices of $\dot B$. 

This is an instance where identifiability of parameters relevant for control is lost  by playing a near optimal policy \citep{polderman1986necessity}. In particular, playing a certainty equivalent policy synthesized using an estimated parameter may not provide sufficient information to identify the optimal policy. The issue can be resolved by including an appropriate probing term in the exploration policy with $D_{\exp}^\eta \neq 0$. 

\fi
\ifTACmode\else
    
\subsection{Misspecification}
    \label{s: misspecified}

    The above examples studied classical  fragile systems in the face of the excess cost lower bounds of \Cref{thm: excess cost lower bound}. Critically, the lower bound elucidates the unavoidable excess cost of learning only in the presence of a well-specified problem: one where the parametric dynamics model captures the data generating process of \eqref{eq: linsys}. If one engages in model-based design, e.g. by following the certainty equivalent synthesis procedure, then they may only attain the stated bounds if the model class is appropriately selected. The primary issue of robust control is the uncertainty due to dynamics which are challenging to model, at least via linear systems, inducing some misspecification into the model class. To understand the role such misspecification can play, observe that if the search space is restricted to a subset $\Theta \subseteq \R^{d_{\theta}}$ such that $\theta^\star \notin \Theta$, then the identification procedure will asymptotically converge to the best in class estimate $\hat \theta \in \Theta$ that maximizes the likelihood under the exploration policy $\pi_{\exp}$ [Theorem 8.2, \cite{ljung1998system}]. 
Such an estimate will differ from the true underlying parameter and the error might align with a sensitive direction of the control cost, as measured by $H(\theta^\star)$. This in turn will lead to a poorly performing controller due to an incorrect parametric modeling assumption. 

Consider, for example, a variant of the LQG instance of \eqref{eq: doyle example}, where  the actuation matrix is incorrectly assumed to be $\hat{B}(\theta)=\begin{bmatrix}\theta\\1+\varepsilon\end{bmatrix}$,  while the true system evolves under $B(\theta)=\begin{bmatrix}\theta\\1\end{bmatrix}$. 
   Let the true unknown parameter be equal to $\theta^\star=0$. 
   In this case, the learner only fits the parameter for the first component of the actuation matrix using least squares. The estimation for $\theta$ turns out to be asymptotically unbiased, and therefore the resulting dynamics model is correct up to the misspecified second element in the actuation matrix. A Taylor expansion reveals that the learner suffers an excess cost that scales as $\varepsilon^2 H$, for the Hessian $H$ as in \Cref{s: doyle example}. This excess cost can be arbitrarily poor as the value of $\sigma$ approaches zero.

\fi
\section{Proof of Lower Bound}
\label{s: lower bound}

The proof~\Cref{thm: excess cost lower bound} (\Cref{s: lower_bound_proof}) relies on the celebrated Youla parameterization~\cite{zhou1996robust,zheng2020equivalence} and the notion of doubly co-prime factorizations, which we review first in \Cref{s: lb prelim}. The Youla parameter uniquely parameterizes all stabilizing controllers. Moreover, the excess cost is a non-degenerate quadratic function of the Youla parameter (\Cref{s: performance_difference_lemma}). It also enables a clean characterization of the performance gap between two stabilizing linear controllers applied to the same system (\Cref{s: technical_results_Youla}). With these results in hand, we proceed to prove the lower bound using the van Trees inequality (\Cref{s: van trees proof}).
\ifTACmode
The proof is presented for the strictly causal setting; the extension to the non-strictly causal setting can be found in the supplementary material.
\else
While the proof is presented for the strictly causal setting, in \Cref{s: non-strictly causal}, we extend the result to the non-strictly causal setting.  
\fi

\subsection{Review of Youla parameterization}
\label{s: lb prelim}
Doubly co-prime factorizations enable us to parameterize all possible stabilizing controllers. They are defined as follows.
\begin{definition}[Doubly co-prime factorization]
    Let $P\in\calR^{\dy\times\du}_p$ be any system. We say that $M, U, N, V, \tilde M, \tilde U, \tilde N, \tilde V\in\calR\calH_\infty$ are a doubly co-prime factorization of $P$ if and only if 
    \begin{itemize}
        \item $P = N M^{-1} = \tilde M^{-1} \tilde N$ 
        \item $\tilde V M - \tilde U N = I,  \tilde M V - \tilde N U = I$ \quad (co-prime identity).
        \item $\tilde V U - \tilde U V = 0, \tilde M N - \tilde N M = 0$ \quad (orthogonality).
    \end{itemize}
\end{definition}
The parameterization of the stabilizing controllers is achieved via the celebrated Youla parameter.
\begin{lemma}
    \label{lem: youla existence and uniqueness}
    Consider any proper controller $K\in\calR_p^{\du\times\dy}$ that internally stabilizes a strictly proper system $P\in\calR^{\dy\times\du}_s$. Suppose that $M, U, N, V, \tilde M, \tilde U, \tilde N, \tilde V\in\calR\calH_\infty$ is a doubly co-prime factorization of $P$.
    Then there exists a unique stable transfer function $\calQ$, called the Youla parameter, that is proper and satisfies $K = (U + M \calQ)(V + N \calQ)^{-1} = (\tilde V + \calQ \tilde N)^{-1} (\tilde U + \calQ \tilde M)$. In particular, it holds that
     \begin{align}\label{eq: explicit youla}
        \calQ = \bmat{\tilde V & -\tilde U }\bmat{K \\ I} (I - P K)^{-1} \tilde M^{-1}. 
    \end{align}
    If in addition $U$ is strictly proper, then $K$ is strictly proper if and only if $\calQ$ is strictly proper.
\end{lemma}

\begin{proof}
    Let $U',V'\in\calR\calH_{\infty}$ be any right co-prime  factorization (chapter 5 in~\cite{zhou1996robust}) of the controller, that is, $K = U' V'^{-1}$.  Following the same steps as in Theorem 12.17 of~\cite{zhou1996robust}, it follows that $K$ is internally stabilizing if and only if there exist $\calQ\in\calR\calH_\infty,\tilde{\calQ}\in\calR\calH_\infty$ such that $K=(U + M \calQ)(V + N \calQ)^{-1}=(\tilde V +  \tilde{\calQ}\tilde N)^{-1}(\tilde U + \tilde \calQ\tilde M )$. 
Multiplying by $V + N \calQ$ and $\tilde V +  \tilde{\calQ}\tilde N$ from the right and the left respectively, we obtain
       \begin{align*}
       &  (\tilde V + \tilde \calQ \tilde N)(U + M \calQ) =  (\tilde U+ \tilde \calQ \tilde M)(V + N\calQ) \\
        \implies &  \tilde V M \calQ  - \tilde U N \calQ + \tilde \calQ \tilde N U - \tilde \calQ \tilde M V = 0 \quad \mbox{(orthogonality)} \\
        \implies & \calQ - \tilde \calQ  = 0 \quad \mbox{(co-prime identity)}.
    \end{align*}
This implies that $\tilde \calQ = \calQ$. It also implies uniqueness.

    Define $Z = \tilde M V' - \tilde N U'$. Then, by inverting $U'V'^{-1}=K=(\tilde V +  \calQ\tilde N)^{-1}(\tilde U + \calQ\tilde M )$, we obtain $\calQ=(\tilde VU'-\tilde UV')Z^{-1}$.
    Note that $Z^{-1}$ is well-defined and in $\mathcal{RH}_{\infty}$ (see Lemma 5.10 in~\cite{zhou1996robust}). Then, \eqref{eq: explicit youla} follows from the identities
    \[
\begin{bmatrix}
    U'\\V'
\end{bmatrix}Z^{-1}=\begin{bmatrix}
    K\\I
\end{bmatrix}V'Z^{-1},\,V'Z^{-1}=(I-PK)^{-1}\tilde{M}^{-1}.
    \]

   Finally, let now $U$ be strictly proper. Taking the limit $z\rightarrow \infty$, we obtain $M(\infty)\calQ(\infty)V^{-1}(\infty)=K(\infty)$.
    $M(\infty)$ and $V^{-1}(\infty)$ are both nonzero and finite by the co-prime property and the fact that both $P$ and $K$ are strictly proper. 
    Hence, $\calQ(\infty)=0$ if and only if $K(\infty)=0$.
\end{proof}

The optimal strictly causal LQG controller induces a doubly co-prime factorization. 
Specifically, define the following transfer functions
\begin{equation}
\label{eq: nominal co-prime}
\begin{aligned}
   \bmat{M_{\theta} & U_\theta \\ N_\theta & V_\theta} &\triangleq \left[\begin{array}{c|cc}
              A_{\mathsf{cl}}^c(\theta) & B(\theta) & L(\theta) \\ \hline 
              F(\theta) & I & 0 \\ C(\theta) & 0 & I 
            \end{array}\right] \\
     \bmat{\tilde V_{\theta} & -\tilde U_\theta \\ -\tilde N_\theta & \tilde M_{\theta}} &\triangleq \left[\begin{array}{c|cc}
             A_{\mathsf{cl}}^o(\theta) & -B(\theta) & -L(\theta) \\ \hline 
              F(\theta) & I & 0 \\ C(\theta) & 0 & I 
            \end{array}\right], 
\end{aligned}
\end{equation}
where we define the closed loop matrix under the LQR optimal controller as $A_{\mathsf{cl}}^c(\theta) \triangleq A(\theta) + B(\theta) F(\theta)$ and recall the shorthand for the closed loop matrix under the Kalman predictor gain $A_{\mathsf{cl}}^o(\theta) = A(\theta) - L(\theta) C(\theta)$. It turns out that the above selection of transfer functions is a doubly co-prime factorization for system $\calP^u_\theta$~\cite{zhou1996robust}.
We will exploit these doubly co-prime factorizations to define families of Youla parameters for every instance $\theta$. 

\begin{definition}Let $K$ be the transfer function for a strictly causal linear policy that stabilizes the system defined by instance $\theta$. Denote by $\calQ_{\theta}^K$ the Youla parameter corresponding to the controller $K$ for the above co-prime factorization. In particular, $\calQ_{\theta}^K$ is a stable, strictly proper real rational transfer function such that
\begin{equation}
     \label{eq: youla}
\begin{aligned}
    K &= (U_\theta + M_\theta \calQ_{\theta}^{K}) (V_{\theta} + N_{\theta} \calQ_{\theta}^{K})^{-1} \\
    &= (\tilde V_{\theta} + \calQ_{\theta}^K \tilde N_\theta)^{-1} (\tilde U_{\theta} + \calQ_{\theta}^K \tilde M_{\theta}).
\end{aligned}
\end{equation}
\end{definition}

\subsection{Performance Difference Lemma}\label{s: performance_difference_lemma}
We now proceed to study the performance gap between an arbitrary stabilizing controller and the optimal LQG policy. Such results are known in the literature as performance difference lemmas. Our first result expresses the performance gap in terms of control policies mapping a history of both measurements and inputs to future control inputs. Our second result shows that the gap is simply equal to a non-degenerate quadratic function of the Youla parameter. In order to state this result, we define $\tilde \calP_\theta \triangleq \bmat{0 & I\\ \calP^{d}_\theta & \calP_\theta^u }$ and $K_{\theta}^{y,u} \triangleq \bmat{\tilde U_{\theta} & I - \tilde V_{\theta}}$.

\begin{lemma}[Performance Difference Lemma]
\label{lem: performance difference lemma}
Consider strictly proper stabilizing linear controller $K$. Let $K^{y,u}$ be any transfer matrix of dimension $\du \times (\dy + \du)$ satisfying 
\begin{align*}
    K = K^{y,u} \bmat{I \\ K}. 
\end{align*}
Then the excess cost of applying controller $K$ to system \eqref{eq: linsys} with parameter $\theta$ is given by 
    \begin{align}
        &J(K, \theta) - J(K_{\theta}, \theta)\nonumber \\
        &= \norm{\Psi(\theta)^{1/2} (K^{y,u} -  K_{\theta}^{y,u}) \bmat{I \\ K} \calF(\tilde \calP_{\theta}, K)}_{\calH_2}^2\label{eq: performance_difference_K_version}\\
        &= \norm{\Psi(\theta)^{1/2}\calQ^K_{\theta}\Sigma^{1/2}_e(\theta)}_{\calH_2}^2. \label{eq: performance_difference_Q_version}
    \end{align}
    As a result,
      \begin{align*}
        J(K, \theta) - J(K_{\theta}, \theta) \geq \mu \norm{\calQ_\theta^K}_{\calH_2}^2,
    \end{align*}
    where 
       $\mu = {\lambda_{\min}(B(\theta)^\top P(\theta) B(\theta) + R)} 
        {\lambda_{\min}(\Sigma_e(\theta))}  > 0.$
\end{lemma}
\begin{proof}
   The lower bound involving $\mu$ follows immediately from~\eqref{eq: performance_difference_Q_version}. We only need to prove~\eqref{eq: performance_difference_K_version},~\eqref{eq: performance_difference_Q_version}.
   
   \textbf{Part a): proof of~\eqref{eq: performance_difference_K_version}.} By the performance difference lemma (Theorem 11.2 of \citep{soderstrom2012discrete}), it holds that 
    \begin{align*}
        &J(K, \theta) - J(K_{\theta}, \theta) \\
        &= \limsup_{T} \frac{1}{T}\E_{\theta}^K \sum_{t=0}^T \norm{\Psi(\theta)^{1/2} (u_t - u_t^\theta)}^2,
    \end{align*}
        where $u_t^\theta$ is the action taken by the LQG controller for the system with dynamics $\theta$ conditioned on the observed history of actions and observations $u_0, y_0, \dots, u_{t-1}, y_{t-1}$:
    \begin{align*}
        u_t^\theta = F(\theta) \bfE_{\theta}\brac{x_t \vert y_{0:t-1}, u_{0:t-1}}.
    \end{align*}
    By contrast, $u_t$ is the control action taken by controller $K$, and the expectation is under the distribution induced by $K$.

    We may write this in frequency domain by denoting $\bfy$ as the signal generated by the feedback interconnection of the plant $\tilde \calP_{\theta}$ with the controller $K$: $\bfy = \calF(\tilde \calP_{\theta}, K) \boldsymbol d$, for noise signal $\boldsymbol{d} = \bmat{\boldsymbol{w} \\ \boldsymbol{v}}$. Then the control input sequence may be expressed as $\bfu = K \calF(\tilde \calP_{\theta}, K) \boldsymbol d$. The history of measurements and inputs can therefore be written as $\bmat{\bfy \\ \bfu} = \bmat{I \\ K} \calF(\tilde \calP_{\theta}, K) \boldsymbol d$. Consequently, 
        $\bfu = K^{y,u} \bmat{I \\ K} \calF(\tilde \calP_{\theta}, K)\boldsymbol{d}$
    for any $K^{y,u}$ satisfying $K^{y,u} \bmat{I \\K} = K$. 
    For the sequence of optimal LQG actions for system $\theta$ conditioned on the observations, it holds that
    \begin{align*}
        \bfu^{\theta} %
        &= F(\theta) (z I - A(\theta) + L(\theta) C(\theta))^{-1} (L(\theta) \bfy + B(\theta) K \bfy) \\
        &= \bmat{\tilde U_\theta & I - \tilde V_\theta} \bmat{I \\ K} \calF(\tilde \calP_{\theta}, K)\boldsymbol{d}.
    \end{align*}
    The performance gap can then be expressed by Parseval's theorem as in~\eqref{eq: performance_difference_K_version}.
    
     \textbf{Part b): proof of~\eqref{eq: performance_difference_Q_version}.}
         Select $K^{y,u} = \bmat{0 & I}$. Then the gap can be written as 
    \begin{align*}
        &J(K, \theta) - J(K_\theta, \theta)\\ &= \norm{\Psi(\theta)^{1/2}\bmat{ -\tilde U_{\theta} &  \tilde V_{\theta} } \bmat{I \\ K} (I - \calP_\theta^u K)^{-1} \calP^d_\theta}_{\calH_2}^2, \\
        &= \norm{\Psi(\theta)^{1/2}\calQ_{\theta}^K \tilde M_{\theta} \calP^d_\theta}_{\calH_2}^2
    \end{align*}
    where the second equality follows from~\Cref{eq: explicit youla}.

    Recall that $$\calP^{d}_\theta = \bmat{C(\theta)(zI - A(\theta))^{-1} \Sigma_w(\theta)^{1/2} & \Sigma_v(\theta)^{1/2}},$$ and 
    $\tilde M_{\theta} = -C(\theta) (zI - A_\mathsf{cl}^o(\theta))^{-1} L(\theta) + I$. We drop the dependence of the system matrices on $\theta$ in the following. After some algebraic manipulations, it follows that 
    \begin{align*}
        \tilde M_{\theta} \calP^{d}_{\theta} &=\! \bmat{C(zI \!-\!A_{\mathsf{cl}}^o)^{-1} \Sigma_w^{1/2} \!& \!(I \!-\! C(zI -A_{\mathsf{cl}}^o)^{-1} L )\Sigma_v^{1/2} }. 
    \end{align*}
     Define $\calR = (zI - A_{\mathsf{cl}}^o)^{-1}$ and $\tilde \calR = \calR^*$. It holds that 
    \begin{align*}
        (\tilde M_{\theta} \calP^{d}_{\theta})(\tilde M_{\theta} P^{d}_{\theta})^* &= C \calR \Sigma_w \tilde \calR C^\top + C \calR L\Sigma_v L^\top \tilde \calR C^\top \\
        &-C\calR L\Sigma_v - \Sigma_v L^\top \tilde \calR C^\top +  \Sigma_v. 
    \end{align*} 
    Recall that $L \Sigma_v = A_{\mathsf{cl}}^o\Sigma C^\top$ and $\Sigma_v L^\top = C\Sigma A_{\mathsf{cl}}^o$. Also note that $\calR A_{\mathsf{cl}}^o = -I + \calR z$. Then 
        $C \calR L \Sigma_v = C  \calR A_{\mathsf{cl}}^o \Sigma C^\top = -C \Sigma C^\top +  z C\calR \Sigma C^\top.$
    Additionally,
    \begin{align*}
        &C \calR (\Sigma_w + L \Sigma_v L^\top) \tilde \calR C^\top \\
        &= C \calR (\Sigma - A_{\mathsf{cl}}^o \Sigma (A_{\mathsf{cl}}^o)^\top) \tilde \calR C^\top  \\
        &= - C\Sigma C^\top\!+\!z C\calR \Sigma C^\top \!+\! z^* C\Sigma \tilde \calR C^\top \!+\! C \calR \Sigma \tilde \calR C^\top (1 \!-\! z z^*). 
    \end{align*}
    Then 
        $(\tilde M_{\theta} P^{d}_{\theta})(\tilde M_{\theta} P^{d}_{\theta})^*\vert_{z=e^{j\omega}} = C\Sigma C^\top + \Sigma_v = \Sigma_e.$ 
    The proof of~\eqref{eq: performance_difference_Q_version} now follows from the definition of $\calH_2$ norm and the cyclical invariance property of the trace operator. 
\end{proof}
The quadratic structure of the performance gap is key to applying the van Trees inequality, which is suitable for obtaining lower bounds in quadratic estimation problems.

\subsection{Technical results on the Youla parameterization}\label{s: technical_results_Youla}
The section contains several auxiliary technical results for proving Theorem~\ref{thm: excess cost lower bound}. First, we establish bounds on the $\calH_\infty$ norm of the Youla parameter of any controller belonging to the class specified in Assumption~\ref{asmp: closed loop}. Second, we prove local Lipschitz properties for the co-prime parameterizations. Third, we characterize the gap between closed-loop responses of two different controllers.
\begin{lemma}[Bounded Youla Parameter]
    \label{lem: bounded controller norm bounded Youla}
Let Assumption~\ref{asmp: closed loop} hold. It holds that
\[
\norm{\calQ^{K}_{\theta^\star}}_{\calH_{\infty}}\le \beta, 
\]
for $\beta\!=\!\norm{\begin{bmatrix}
    \tilde V_{\theta^\star}&\!-\tilde U_{\theta^\star}
\end{bmatrix}}_{\mathcal{H}_\infty}(\alpha\!+\!2\norm{\mathcal{T}_{K_{\theta^\star}}}_{\mathcal{H}_{\infty}})
\norm{\begin{bmatrix}
    -U_{\theta^\star}\\ V_{\theta^\star}
\end{bmatrix}}_{\mathcal{H}_\infty}.$
\end{lemma}
\begin{proof} By Lemma~\ref{lem: affinity of closed-loop map} and the properties of the co-prime factorization, we have
$\calQ^{K}_{\theta^\star}=\begin{bmatrix}
    \tilde V_{\theta^\star}&-\tilde U_{\theta^\star}
\end{bmatrix}(\mathcal{T}_K-\mathcal{T}_{K_{\theta^\star}}) \begin{bmatrix}
    -U_{\theta^\star}\\ V_{\theta^\star}
\end{bmatrix}.$
Taking the $\calH_\infty$ norm of this quantity and applying Assumption~\ref{asmp: closed loop} along with submultiplicativity provides the given value of $\beta$. 
\end{proof}

\begin{lemma}[Closed-loop map is affine]\label{lem: affinity of closed-loop map}
Recall the definition of the closed loop map $\mathcal{T}_K$ of the plant $\mathcal{P}_{\theta^\star}^u$ under controller $K$. Let $\calQ^K_{\theta^\star}$ be the corresponding Youla parameter. Let $\mathcal{T}_{K_{\theta^\star}}$ be the closed-loop map under the nominal controller $U_{\theta^\star}V_{\theta^\star}^{-1}$ with $\calQ^{K_{\theta^\star}}_{\theta^\star}=0$. Then, we have
\[
\mathcal{T}_{K}=\mathcal{T}_{K_{\theta^\star}}+\begin{bmatrix}M_{\theta^\star}\\N_{\theta^\star}\end{bmatrix}\calQ^K_{\theta^\star}\begin{bmatrix}\tilde N_{\theta^\star}&\tilde M_{\theta^\star}\end{bmatrix},
\]
where $\tilde M_{\theta^\star},\tilde M_{\theta^\star},\tilde M_{\theta^\star},\tilde M_{\theta^\star}$ are stable, proper and $\mathcal{P}_{\theta^\star}^u=\tilde M^{-1}_{\theta^\star}\tilde N_{\theta^\star}=N_{\theta^\star}M^{-1}_{\theta^\star}$.
\end{lemma}
\begin{proof}
Let
\begin{align*}
    \tilde V&=\tilde V_{\theta^\star}+\calQ^{K}_{\theta^\star}\tilde N_{\theta^\star},& \tilde U&=\tilde{U}_{\theta^\star}+\calQ^{K}_{\theta^\star}\tilde{M}_{\theta^\star},\\
     V&=V_{\theta^\star}+N_{\theta^\star}\calQ^{K}_{\theta^\star},& U&=U_{\theta^\star}+M_{\theta^\star}\calQ^{K}_{\theta^\star}.
\end{align*}
Recall that
$K=UV^{-1}=\tilde{V}^{-1}\tilde U$
with $\mathcal{P}_{\theta^\star}^u=N_{\theta^\star}M^{-1}_{\theta^\star}=\tilde M^{-1}_{\theta^\star}\tilde N_{\theta^\star}$. 
By orthogonality and the co-prime identity, it follows that
    $\tilde V M_{\theta^\star}-\tilde UN_{\theta^\star}=I, \tilde M_{\theta^\star}V -\tilde N_{\theta^\star} U=I.$
Thus, 
\begin{align*}
 (I-K\calP_{\theta^\star}^u)^{-1}&=M_{\theta^\star}(\tilde V M_{\theta^\star}-\tilde UN_{\theta^\star})^{-1}\tilde{V}=M_{\theta^\star}\tilde{V},\\
 \calP_{\theta^\star}^u(I-K\calP_{\theta^\star}^u)^{-1}&=N_{\theta^\star}\tilde{V}.
\end{align*}
Working similarly, we obtain that
 $(I-\calP_{\theta^\star}^u K)^{-1}=V\tilde M_{\theta^\star}$ and $K (I-\calP_{\theta^\star}^u K)^{-1}=U\tilde M_{\theta^\star}.$
Finally, expanding $V,U,\tilde{V},\tilde{U}$:
\[
\mathcal{T}_{K}=\mathcal{T}_{K_{\theta^\star}}+\begin{bmatrix}M_{\theta^\star}\calQ^K_{\theta^\star}\tilde N_{\theta^\star}& M_{\theta^\star}\calQ^K_{\theta^\star}\tilde M_{\theta^\star}\\ N_{\theta^\star}\calQ^K_{\theta^\star}\tilde N_{\theta^\star} & N_{\theta^\star}\calQ^K_{\theta^\star}\tilde M_{\theta^\star}\end{bmatrix},
\]
which completes the proof.
\end{proof}

Throughout our proof, we leverage the fact that various system quantities, including the transfer matrices defining the co-prime factorization, the Youla parameter corresponding to any LQG policy, and are smooth in the unknown parameter. The following lemmas establish these results.
\begin{lemma}[Co-prime Perturbations]
    \label{lem: coprime perturbations}
    There exists constants $ c_{\varepsilon}^X(\theta^\star)$ and $L_X(\theta^\star)$ depending only on the system instance $\theta^\star$ such that for $\varepsilon \leq  c_{\varepsilon}^X(\theta^\star)$ and any $\theta_1, \theta_2 \in \calB(\theta^\star, \varepsilon)$ the transfer matrices of the co-prime factorization \eqref{eq: nominal co-prime} satisfy $\norm{X_{\theta_i} - X_{\theta_2}}_{\mathcal{H}_{\infty}} \leq L_X(\theta^\star) \varepsilon$
    for $X$ denoting any of the transfer matrices $M,N,U,V,\tilde M, \tilde N, \tilde U, \tilde V$.   
\end{lemma}
\begin{proof}
    This result follows from the results on the smoothness of the Riccati equation solution with respect to the system parameters, as shown in
    Appendix B of\citep{simchowitz2020naive} and the assumption that the system matrices are analytic functions of $\theta$. 
\end{proof}

\begin{lemma}
    \label{lem: stability}
    There exists a constant $c_{\varepsilon}^{\mathsf{stab}}(\theta^\star, \alpha)$ such that if $\varepsilon \leq c_{\varepsilon}^{\mathsf{stab}}(\theta^\star, \alpha)$, then for any policy $K \in \Pi$ from Assumption~\ref{asmp: closed loop} and any system instance $\theta \in \calB(\theta^\star, \varepsilon)$, $K$ stabilizes the system $\calP_{\theta}^u$. 
\end{lemma}
\begin{proof}
    Let $K = U V^{-1}$ with $U= U_{\theta^\star}+M_{\theta^\star}\calQ^K_{\theta^\star}$, $V=V_{\theta^\star}+N_{\theta^\star}\calQ^K_{\theta^\star}$. By Lemma 5.10 of \citep{zhou1996robust}, a necessary and sufficient condition for the interconnection of $\calP_{\theta}^u$ with $K$ to be internally stable is that $\tilde M_{\theta} V - \tilde N_{\theta} U$ is invertible in $\mathcal{RH}_{\infty}$. This quantity may be written $\tilde M_{\theta} V - \tilde N_{\theta} U = (\tilde M_{\theta^\star} + \Delta \tilde M) V - (\tilde N_{\theta^\star} +  \Delta \tilde N) U$, for $\Delta \tilde N = \tilde N_{\theta} - \tilde N_{\theta^\star}$ and $\Delta \tilde M = \tilde M_{\theta} - \tilde M_{\theta^\star}$. As $\tilde M_{\theta^\star} V - \tilde N_{\theta^\star} U=I$ it suffices to have
    \begin{align*}
        \norm{(\Delta \tilde M\, V - \Delta \tilde N\, U)}_{\calH_\infty} \leq 1/2.
    \end{align*}
    Then it also suffices to have  
    \begin{align*}
        \norm{\bmat{\Delta \tilde M & \Delta \tilde N}}_{\calH_{\infty}} \leq \frac{1}{2\norm{\bmat{V \\ U} }_{\calH_{\infty}}}.
    \end{align*}
    It holds that \begin{align*}
 &   \norm{\bmat{V \\ U} }_{\calH_{\infty}}\le \norm{\bmat{V_{\theta^\star} \\ U_{\theta^\star}} }_{\calH_{\infty}}+\norm{\bmat{N_{\theta^\star} \\ M_{\theta^\star}} }_{\calH_{\infty}}\norm{\calQ^K_{\theta^\star} }_{\calH_{\infty}}\\
 & \norm{\bmat{V \\ U} }_{\calH_{\infty}}\le \underbrace{\norm{\bmat{V_{\theta^\star} \\ U_{\theta^\star}} }_{\calH_{\infty}}+\norm{\bmat{N_{\theta^\star} \\ M_{\theta^\star}} }_{\calH_{\infty}}\beta}_{\gamma(\alpha)}.
    \end{align*}

    Then for stability of the interconnection, it suffices to have 
     \begin{align*}
        \norm{\bmat{\Delta \tilde M & \Delta \tilde N}}_{\calH_{\infty}}\leq \frac{1}{2(\gamma(\alpha))}.
    \end{align*}
    By \Cref{lem: coprime perturbations}, there exists bounds on $\norm{\Delta \tilde N}_{\calH_\infty}$ and $\norm{\Delta \tilde N}_{\calH_{\infty}}$ that are linear in $\varepsilon$ as long as $\varepsilon \leq c_\varepsilon^X(\theta^\star)$. Then there exists constant $c_{\varepsilon}^{\mathsf{stab}}$ depending on the system and the value of $\alpha$ such that if $\varepsilon \leq c_{\varepsilon}^{\mathsf{stab}}$, the above bound on $\norm{\bmat{\Delta \tilde M & \Delta \tilde N}}_{\calH_{\infty}}$ is satisfied.
\end{proof}

\begin{lemma}[Smoothness of Youla Parameter]
\label{lem: smoothness of youla parameter}
    Suppose Assumption~\ref{asmp: closed loop} holds. Consider a strictly proper transfer function $K \in \Pi$.
 There exist constants $ c_{\varepsilon}^\calQ(\theta^\star, \alpha)$ and $ L_\calQ(\theta^\star, \alpha)$ such that for $\varepsilon \leq c_{\varepsilon}^\calQ(\theta^\star, \alpha)$ and any two $\theta_1, \theta_2 \in \calB(\theta^\star, \varepsilon)$,
    \begin{align*}
        \norm{\calQ_{\theta_2}^K -\calQ_{\theta_1}^K}_{\calH_\infty} \leq L_\calQ(\theta^\star, \alpha) \norm{\theta_2 - \theta_1}.
    \end{align*}
\end{lemma}
\begin{proof}
    By \Cref{lem: stability}, $\varepsilon$ can be taken small enough that $K$ stabilizes the systems corresponding to $\theta_1$ and $\theta_2$. 
    We begin by expressing $\calQ_{\theta_1}^K$ in terms of $\calQ_{\theta_2}^K$. Specifically, recall that $\calQ_{\theta_1}^K$ is such that the controller can be expressed in terms of the right co-prime factorization:
        $K = (U_{\theta_1} + M_{\theta_1} \calQ_{\theta_1}^K)(V_{\theta_1} + N_{\theta_1} \calQ_{\theta_1}^K)^{-1}. $
    Similarly, $\calQ_{\theta_2}^K$ is such that the controller can be expressed in terms of the left co-prime factorization as 
        $K = (\tilde V_{\theta_2} + \calQ_{\theta_2}^K \tilde N_{\theta_2})^{-1}(\tilde U_{\theta_2} + \calQ_{\theta_2}^K \tilde M_{\theta_2}).$ 
    Setting these two expressions equal, we can express $\calQ_{\theta_2}^K$ as 
    \begin{align*}
        \calQ_{\theta_2}^K &= \brac{\tilde V_{\theta_2}(U_{\theta_1} + M_{\theta_1} \calQ_{\theta_1}^K) - \tilde U_{\theta_2} (V_{\theta_1} + N_{\theta_1} \calQ_{\theta_1}^K)} \\
        &\times\brac{\tilde M_{\theta_2}(V_{\theta_1} + N_{\theta_1}\calQ_{\theta_1}^K)- \tilde N_{\theta_2}(U_{\theta_1} + M_{\theta_1} \calQ_{\theta_1}^K)}^{-1}.
    \end{align*}
    Let $\tilde V_{\theta_2} =\tilde V_{\theta_1} + \Delta \tilde V$, $\tilde U_{\theta_2} = \tilde U_{\theta_1} + \Delta \tilde U$, $\tilde M_{\theta_2} = \tilde M_{\theta_1} + \Delta \tilde M$, $\tilde N_{\theta_2} = \tilde N_{\theta_1} + \Delta \tilde N$. Then it holds that the numerator in our expression for $\calQ_{\theta_2}^K$ can be expressed as 
    \begin{align*}
        &\tilde V_{\theta_2}(U_{\theta_1} + M_{\theta_1} \calQ_{\theta_1}^K) - \tilde U_{\theta_2} (V_{\theta_1} + N_{\theta_1} \calQ_{\theta_1}^K) \\
        &= \tilde V_{\theta_1} U_{\theta_1} - \tilde U_{\theta_1} V_{\theta_1} + (\tilde V_{\theta_1} M_{\theta_1} -\tilde U_{\theta_1} N_{\theta_1})\calQ_{\theta_1}^K + \Delta \tilde V U_{\theta_1} \\
        &- \Delta \tilde U V_{\theta_1} + (\Delta \tilde V M_{\theta_1} - \Delta \tilde U N_{\theta_1}) \calQ_{\theta_1}^K \\
        &= \calQ_{\theta_1}^K + \Delta_1 + \Delta_2 \calQ_{\theta_1}^K,
    \end{align*}
    where $\Delta_1 = \Delta \tilde V U_{\theta_1} - \Delta \tilde U V_{\theta_1}$ and $\Delta_2 = \Delta \tilde V M_{\theta_1} - \Delta \tilde U N_{\theta_1}$. The final equality follows from orthogonality and the co-prime identity. The denominator in our expression for $\calQ_{\theta_2}^K$ similarly simplifies as
    \begin{align*}
        &\tilde M_{\theta_2}(V_{\theta_1} + N_{\theta_1}\calQ_{\theta_1}^K)- \tilde N_{\theta_2}(U_{\theta_1} + M_{\theta_1} \calQ_{\theta_1}^K) 
        \\&= \tilde M_{\theta_1} V_{\theta_1} - \tilde N_{\theta_1} U_{\theta_1} + (\tilde M_{\theta_1} N_{\theta_1} - \tilde N_{\theta_1} M_{\theta_1})\calQ_{\theta_1}^K + \Delta \tilde M V_{\theta_1} \\
        &- \Delta \tilde N U_{\theta_1} + (\Delta \tilde M N_{\theta_1} - \Delta \tilde N M_{\theta_1})\calQ_{\theta_1}^K \\
        &= I + \Delta_3 + \Delta_4 \calQ_{\theta_1}^K,
    \end{align*}
    where $\Delta_3 = \Delta \tilde M V_{\theta_1} - \Delta \tilde N U_{\theta_1}$ and $\Delta_4 = \Delta \tilde M N_{\theta_1} - \Delta \tilde N M_{\theta_1}$. Combining these facts, we find that 
    \begin{align*}
        \calQ_{\theta_2}^K = (\calQ_{\theta_1}^K+  \Delta_1 + \Delta_2 \calQ_{\theta_1}^K)(I + \Delta_3 + \Delta_4 \calQ_{\theta_1}^K)^{-1}.
    \end{align*}
    Then
    \begin{align*}
       &\norm{\calQ_{\theta_2}^K - \calQ_{\theta_1}^K}_{\calH_\infty} \\
       &= \norm{(\calQ_{\theta_1}^K+  \Delta_1 + \Delta_2 \calQ_{\theta_1}^K)(I + \Delta_3 + \Delta_4 \calQ_{\theta_1}^K)^{-1} - \calQ_{\theta_1}^K}_{\calH_\infty} \\
       &= \big\|(\Delta_1 + \Delta_2 \calQ_{\theta_1}^K + \calQ_{\theta_1}^K \Delta_3 + \calQ_{\theta_1}^K \Delta_4 \calQ_{\theta_1}^K) \\
       &\times (I + \Delta_3 + \Delta_4 \calQ_{\theta_1}^K)^{-1}\big\|_{\calH_\infty} \\
       &\leq \sum_{i=1}^4 \norm{  \Delta_i}_{\calH_\infty}\!\!(1 + \norm{\calQ_{\theta_1}^K}_{\calH_\infty}^2\!\!)  \norm{(I + \Delta_3 + \Delta_4 \calQ_{\theta_1}^K)^{-1}}_{\calH_\infty}\!\!. 
    \end{align*}
    By \Cref{lem: bounded controller norm bounded Youla}, $\norm{\calQ_{\theta_1}^K}_{\calH_\infty}$ is bounded by a quantity $\beta$ defined in terms of the parameter $\alpha$ of Assumption~\ref{asmp: closed loop} and system constants defined in terms of the nominal instance $\theta^\star$.  Combining this fact with the co-prime perturbations bounds of \Cref{lem: coprime perturbations} ensures the existence of a constant $c_{\varepsilon}^\calQ(\theta^\star, \alpha)$ such that if $\varepsilon\leq  c_{\varepsilon}^\calQ(\theta^\star, \alpha)$, then $\norm{\Delta_3 + \Delta_4 \calQ_{\theta_1}^K}_{\calH_\infty} \leq \frac{1}{2}$. This in turn implies that $\norm{(I + \Delta_3 + \Delta_4 \calQ_{\theta_1}^K)^{-1}}_{\calH_\infty} \leq 2$. %
 Then we are left with the inequality 
        $\norm{\calQ_{\theta_2}^K - \calQ_{\theta_1}^K}_{\calH_\infty} \leq 2 \sum_{i=1}^4 \norm{\Delta_i}_{\calH_\infty}(1 + \beta^2)$.
    We achieve the lemma statement by again appealing to the co-prime perturbation arguments of \Cref{lem: coprime perturbations} and defining $L_\calQ(\theta^\star, \varepsilon)$ appropriately.
\end{proof}

\begin{lemma}[Closeness of LQG solutions]
    \label{lem: closeness of LQG solutions}
    There exists a constant $ c_{\varepsilon}^{LQG}(\theta^\star, \alpha)$ depending only on the system instance $\theta^\star$ and the value of $\alpha$ in Assumption~\ref{asmp: closed loop} such that for $\varepsilon \leq c_{\varepsilon}^{LQG}(\theta^\star, \alpha)$, the policy class $\Pi$ contains the optimal LQG policy corresponding to every instance $\theta \in \calB(\theta^\star, \varepsilon)$.
\end{lemma}

\begin{proof}
    By the Riccati perturbation arguments in Appendix B of \citep{simchowitz2020naive}, there exists a constant $ c_{\varepsilon}^{LQG}(\theta^\star, \alpha)$ dependent on the nominal system and $\alpha$ such that if $\varepsilon < c_{\varepsilon}^{LQG}(\theta^\star, \alpha)$, then for any $\theta \in \calB(\theta^\star, \varepsilon)$, the LQG controller $K_{\theta}$ satisfies the inequality $\norm{\mathcal{T}_{K_{\theta}} - \mathcal{T}_{K_{\theta_\star}}}_{\calH_{\infty}
    } \leq \alpha$.
    Then the optimal LQG controller corresponding to any $\theta \in \calB(\theta^\star, \varepsilon)$ is within the policy class $\Pi$ of Assumption~\ref{asmp: closed loop}. 
\end{proof}
Next, we present a bound characterizing the gap between the response of two different controllers applied to two different systems, using the following definition of the closed loop gap: 
\begin{equation}
    \label{eq: closed loop gap}
\begin{aligned}
\mathsf{CL\_gap}(K, \theta_1, \theta_2)\!\triangleq\! \bmat{I \\ K} (\calF(\tilde \calP_{\theta_2}, K) \!-\! \calF(\tilde \calP_{\theta_1}, K)).
\end{aligned}
\end{equation}

\begin{lemma}
    \label{lem: closed loop perturbations}
    Consider a controller $K \in \Pi$ such that $K$ stabilizes every system instance in $ \calB(\theta^\star, \varepsilon)$. Let $\theta \in \calB(\theta^\star, \varepsilon)$. Recall $K_{\theta}^{y,u} = \bmat{\tilde U_{\theta} & I - \tilde V_{\theta}}$ and choose
        $K^{y,u} \triangleq \bmat{\tilde U_{\theta} + \calQ_{\theta}^K \tilde M_{\theta} & I - \tilde V_{\theta} - \calQ_{\theta}^K \tilde N_{\theta}}$.

    There exists a constant $L_{CL}(\theta^\star, \alpha)$ that depends only on the nominal system instance such that 
    for $\varepsilon \leq c_{\varepsilon}^{X}(\theta^\star)$ (as in \Cref{lem: coprime perturbations}), 
    \begin{align*}
       &\norm{(K^{y,u} \!\!- \!K_{\theta}^{y,u})\mathsf{CL\_gap}(K, \theta,  \theta^\star\!)}_{\calH_2}\!  \leq \!\norm{\calQ_{\theta}^{K}}_{\calH_2} \!L_{CL}(\theta^\star,\alpha) \varepsilon.
    \end{align*}
\end{lemma}
\begin{proof}
    It holds that $K^{y,u} - K_{\theta}^{y,u} = \calQ_{\theta}^{K} \bmat{ \tilde M_{\theta} & -\tilde N_{\theta}}$. Furthermore, note that for any $\tilde \theta$
    \begin{align*}
    \bmat{I \\ K} \calF(\tilde \calP_{\tilde \theta}, K)
        &= \bmat{V_{\tilde \theta} + N_{\tilde \theta} Q_{\tilde \theta}^K \\ U_{\tilde \theta} + M_{\tilde \theta} Q_{\tilde \theta}^K}(\tilde M_{\tilde \theta} V_{\tilde \theta} - \tilde N_{\tilde \theta} U_{\tilde \theta})^{-1} \tilde M_{\tilde \theta}\calP^{d}_{\tilde \theta} \\
        &= \bmat{V_{\tilde \theta} + N_{\tilde \theta} Q_{\tilde \theta}^K \\ U_{\tilde \theta} + M_{\tilde \theta} Q_{\tilde \theta}^K} \tilde M_{\tilde \theta}\calP^{d}_{\tilde \theta}.
    \end{align*}
    Then 
    \begin{align*}
        \mathsf{CL\_gap}(K, \theta, \theta^\star) 
        &= \bmat{
        V_{\theta}  - V_{\theta^\star}+ N_{\theta} \calQ_{\theta}^{K} - N_{\theta^\star} \calQ_{\theta^\star}^K\\ U_{\theta} - U_{\theta^\star} + M_{\theta} \calQ_{\theta}^{K} - M_{\theta^\star} \calQ_{\theta^\star}^K} \tilde M_{\theta} \calP^{d}_{\theta} \\
        &+ \bmat{  V_{\theta^\star} +N_{\theta^\star}\calQ_{\theta^\star}^K \\ U_{\theta^\star} + M_{\theta^\star} Q_{\theta^\star}^K} \brac{\tilde M_{\theta} \calP^{d}_{\theta} - \tilde M_{\theta^\star} \calP^{d}_{\theta^\star}}. 
    \end{align*}
    Now consider taking the $\calH_2$ norm of the transfer function $(K^{y,u} - K_{\theta}^{y,u})\mathsf{CL\_gap}(K, \theta, \theta^\star)$ and applying the triangle inequality along with mixed $\calH_\infty$,$\calH_2$ submultiplicativity. By the co-prime perturbation bounds of \Cref{lem: coprime perturbations} for  $  \varepsilon \leq c_{\varepsilon}^X(\theta^\star)$, the above quantity can be bounded as  by a linear function of $\norm{\calQ_{\theta}^K}_{\calH_2} \varepsilon$. Denoting the coefficient for this linear function by $L_{CL}(\theta^\star,\alpha)$ provides the inequality in the lemma. 
\end{proof}

\subsection{Proof of \Cref{thm: excess cost lower bound}}\label{s: lower_bound_proof}

For any $\varepsilon > \tilde \varepsilon$, we can lower bound $\calM(\calA, \theta^\star, \varepsilon)$ by $\calM(\calA, \theta^\star, \tilde \varepsilon)$. We therefore prove that the desired bound holds for $\varepsilon$ sufficiently small. 

By \Cref{lem: closeness of LQG solutions}, there exists a constant  $c_{\varepsilon}^{LQG}(\theta^\star, \alpha)$ such that for $\varepsilon \leq c_{\varepsilon}^{LQG}(\theta^\star, \alpha)$, the optimal LQG controller corresponding to any system $\theta \in \calB(\theta^\star, \varepsilon)$ is in the policy class $\Pi$ defined by \Cref{asmp: closed loop}. 
Let $\lambda$ be a prior density over $\calB(\theta^\star,\varepsilon)$ which is absolutely continuous, and satisfies $\lambda(\theta) > 0$ for $\theta \in \calB(\theta^\star, \varepsilon)$ and $\lambda(\theta) \to 0$ as $\norm{\theta - \theta^\star}\to \epsilon$. Specifically, consider the density $\rho(x) = \frac{1}{Z}(1 - \norm{x}^2) \mathbf{1}(\norm{x} \leq 1)$, where $Z$ is a normalization constant, and let $\lambda(x) = \frac{1}{\varepsilon} \rho\paren{\frac{\theta - \theta^\star}{\varepsilon}}$. 
By the fact that the weighted average over a set lower bounds the supremum,
\begin{align*}
    &\sup_{\theta \in \calB(\theta^\star, \varepsilon)} \bfE_{\calD \sim \rho^{\pi_{\exp}, N, T}(\theta)} \brac{ J(K_{\calD}, \theta) - J(K_{\theta}, \theta)} \\
    &\geq \E_{\Theta \sim \lambda} \bfE_{\calD \sim \rho^{\pi_{\exp}, N, T}(\Theta)} \brac{ J(K_{\calD}, \Theta) - J(K_{\Theta}, \Theta)}.
\end{align*}
By \Cref{lem: stability}, there exists a constant $c_{\varepsilon}^{\mathsf{stab}}(\theta^\star, \alpha)$ such that for $\varepsilon\leq c_{\varepsilon}^{\mathsf{stab}}(\theta^\star, \alpha)$, any member of $\Pi$, including $K_{\calD}$, stabilizes all systems in the ball $\calB(\theta^\star, \varepsilon)$. 
Then, by \Cref{lem: performance difference lemma}, the above gap is equal to
\begin{align*}
    \E_{\Theta, \calD} \norm{\Psi^{1/2}(K_{\calD}^{y,u} - K_{\Theta}^{y,u}) \bmat{I \\ K_{\calD}} \calF(\tilde \calP_{\Theta}, K_{\calD})}_{\calH_2}^2, 
\end{align*}
where $K_{\calD}^{y,u}$ is any transfer function such that $K_{\calD} = K_{\calD}^{y,u} \bmat{I \\ K_{\calD}}$.  We will  consider the choice 
\begin{align}
    \label{eq: Kdyu}
        K_{\calD}^{y,u} &=\bmat{ \tilde U_{\Theta} + \calQ_{\Theta}^{K_{\calD}} \tilde M_{\Theta}& I - \tilde V_{\Theta} - \calQ_{\Theta}^{K_{\calD}} \tilde N_{\Theta}}.
    \end{align}

Recall the definition of $\mu$ from \Cref{lem: performance difference lemma} and of $L_\calQ(\theta^\star, \alpha)$ from \Cref{lem: smoothness of youla parameter}. Define the event 
\begin{align}
    \label{eq: calE definition}
    \calE \!=\! \curly{\!\norm{\calQ_{\theta^\star}^{K_{D}}\!}_{\calH_2}^2\!\! \!\leq\! \frac{\trace\paren{H(\theta^\star) \mathsf{FI}^{\pi_{\exp}}(\theta^\star)^{\!-1}\!}}{ N \mu (1-2^{-1/8})}\!\! +\!  2  L_\calQ(\theta^\star, \alpha)^2 \varepsilon^2},
\end{align}
and let $\mathbf{1}(\calE)$ be the corresponding indicator. It holds by a triangle inequality that 
\begin{align}
    \label{eq: main lb triangle inequality}
    &\E_{\Theta, \calD} \norm{ \Psi^{1/2} (K_{\calD}^{y,u} - K_{\Theta}^{y,u}) \bmat{I \\ K_{\calD}} \calF(\tilde \calP_{\Theta}, K_{\calD})}_{\calH_2}^2 \\
    \nonumber
    &\geq \E_{\Theta, \calD} \brac{\norm{ \Psi^{1/2} (K_{\calD}^{y,u} - K_{\Theta}^{y,u}) \bmat{I \\ K_{\calD}} \calF(\tilde \calP_{\Theta}, K_{\calD})}_{\calH_2}^2 \mathbf{1}(\calE)} \\\nonumber
    &\geq \E_{\Theta, \calD} \Bigg[\Bigg(\norm{ \Psi^{1/2} (K_{\calD}^{y,u} \!-\! K_{\Theta}^{y,u}) \bmat{I \\ K_{\calD}} \calF(\tilde \calP_{\theta^\star}, K_{\calD})}_{\calH_2} \!\!\!\!\mathbf{1}(\calE) \\ \nonumber& - \!\norm{ \Psi^{1/2} (K_{\calD}^{y,u} \!-\! K_{\Theta}^{y,u}) \mathsf{CL\_gap}(K_{\calD}, \Theta, \theta^\star)}_{\calH_2} \!\! \mathbf{1}(\calE) \Bigg)^2\Bigg],
\end{align}
where we recall the definition of the closed loop gap from \eqref{eq: closed loop gap}. The first term will be lower bounded using the van Trees inequality, proven in \Cref{s: van trees proof}. 
This implies that the first term in our lower bound \eqref{eq: main lb triangle inequality} possesses the desired dependence on the Hessian, Fisher Information and $N$. 

To conclude the proof of \Cref{thm: excess cost lower bound}, it remains to  derive an upper bound on the negative quantity in \eqref{eq: main lb triangle inequality}. To do so, let $\bar \Psi$ be an upper bound on $\Psi$ in Loewner order for all $\theta\in B(\theta^\star, \varepsilon)$. 
It holds by \Cref{lem: closed loop perturbations} that there exists $c_{\varepsilon}^X(\theta^\star)$ and $L_{CL}(\theta^\star, \alpha)$ such that for $\varepsilon \leq c_{\varepsilon}^X(\theta^\star)$, 
\begin{equation*}
\begin{aligned}
    & \norm{ \Psi^{1/2} (K_{\calD}^{y,u} - K_{\Theta}^{y,u}) \mathsf{CL\_gap}(K_{\calD}, \Theta,  \theta^\star)}_{\calH_2}  \mathbf{1}(\calE)\\
    &\leq \bar \Psi \norm{\calQ_{\Theta}^{K_\calD}}_{\calH_2} L_{CL}(\theta^\star, \alpha) \varepsilon  \mathbf{1}(\calE).
\end{aligned}
\end{equation*}
By definition of the event $\calE$, the above quantity can be bounded by 
\begin{equation}
\label{eq: negative term upper bound}
\begin{aligned}
    &\bar \Psi \paren{\!\sqrt{\frac{\trace\paren{H(\theta^\star) \mathsf{FI}^{\pi_{\exp}}(\theta^\star)^{-1}}}{ N \mu (1-2^{-1/8})} \!+ \! 2 L_\calQ(\theta^\star, \alpha)^2\varepsilon^2} } \!L_{CL}(\theta^\star\!,\! \alpha) \varepsilon.
\end{aligned}
\end{equation}

To conclude our proof from these facts, we observe that by the Cauchy-Schwarz Inequality
    $\E[(a - b)^2] %
    \geq (\sqrt{\E[a^2]} - \sqrt{\E[b^2]})^2.$
Identifying 
\begin{align*}
    a &\gets \norm{\Psi^{1/2} (K_{\calD}^{y,u} - K_{\Theta}^{y,u}) \bmat{I \\ K_{\calD}} \calF(\tilde \calP_{\theta^\star}, K_{\calD})}_{\calH_2} \!\!\mathbf{1}(\calE) \\
    b &\gets \norm{ \Psi^{1/2} (K_{\calD}^{y,u} \!-\! K_{\Theta}^{y,u}) \mathsf{CL\_gap}(K_{\calD}, \Theta, \theta^\star)}_{\calH_2} \!\! \mathbf{1}(\calE), 
\end{align*}
we apply \Cref{lem: vt bound} to lower bound $\E[a^2]$ and the upper bound of \eqref{eq: negative term upper bound} to upper bound $\E[b^2]$. Then there exists constants $\tilde c_{\varepsilon}(\theta^\star, \pi_{\exp}, T, \alpha)$ and $c_N(\theta^\star, \pi_{\exp}, T)$ such that for $\varepsilon \leq \tilde c_{\varepsilon}(\theta^\star, \pi_{\exp}, T, \alpha)$, and $N\geq c_N(\theta^\star, \pi_{\exp}, T) \varepsilon^{-2}$ the following lower bound holds:
\begin{align*}
    &\E_{\Theta, \calD} \norm{ \Psi^{1/2} (K_{\calD}^{y,u} - K_{\Theta}^{y,u}) \bmat{I \\ K_{\calD}} \calF(\tilde \calP_{\Theta}, K_{\calD})}_{\calH_2}^2 \\
    &\geq \left(\sqrt{\mathsf{vTLB}} - \sqrt{\frac{\mathsf{vTLB}}{100} + 2 L_\calQ(\theta^\star, \alpha)^2 L_{CL}(\theta^\star, \alpha)^2 \varepsilon^4} \right)^2.
\end{align*}
In particular, it suffices to take $\tilde c_{\varepsilon}(\theta^\star, \pi_{\exp}, T \alpha) = \min\curly{c_{\varepsilon}^{VT}(\theta^\star, \pi_{\exp}, T, \alpha), c_{\varepsilon}^X(\theta^\star), \sqrt{\frac{\mu (1 - 2^{-1/8})}{\sqrt{2} \cdot 100 L_{CL}(\theta^\star, \alpha)}}}$.

For $\varepsilon \leq \paren{\frac{\mathsf{vTLB}}{200 L_\calQ(\theta^\star, \alpha)^2 L_{CL}(\theta^\star, \alpha)^2}}^{1/4}$, the above quantity is lower bounded by $\mathsf{vTLB}/\sqrt{2}$. The condition 
$\varepsilon \leq \paren{\frac{\mathsf{vTLB}}{200 L_\calQ(\theta^\star, \alpha)^2 L_{CL}(\theta^\star, \alpha)^2}}^{1/4}$ can be expressed as $\varepsilon \leq  c_{\varepsilon}(\theta^\star, \pi_{\exp}, T, \alpha) N^{-1/4}$ for some constant $c_{\varepsilon}(\theta^\star, \pi_{\exp}, T, \alpha)$ that depends only on the nominal system instance 
$\theta^\star$, the exploration policy, the length of the exploration episodes, and the parameter $\alpha$. Then desired bound holds for $N \geq c_N(\theta^\star, \pi_{\exp}, T)\varepsilon^{-2}$ and $\varepsilon \leq c_{\varepsilon}(\theta^\star, \pi_{\exp}, T, \alpha)N^{-1/4}$. Consider now any positive $\varepsilon$, and some $N$ that satisfies 
$N \geq \max\bigg\{c_N(\theta^\star, \pi_{\exp}, T)\varepsilon^{-2}$, $c_N(\theta^\star, \pi_{\exp}, T)^2 c_{\varepsilon}(\theta^\star, \pi_{\exp}, T, \alpha)^{-4}\Bigg\}$. Let $\tilde \varepsilon = c_{\varepsilon}(\theta^\star, \pi_{\exp}, T, \alpha) N^{-1/4}$. Consider two cases. First suppose that $\varepsilon \leq \tilde \varepsilon$. Then the condition for the bound holds, and we are done. Next, suppose that $\varepsilon \geq \tilde \varepsilon$. Then the bound holds for $\tilde \varepsilon$, and $\calM(\theta^\star, \varepsilon, \calA) \geq \calM(\theta^\star, \tilde \varepsilon, \calA)$. Thus the bound also holds for $\varepsilon$.

\subsection{Application of van Trees Inequality}

The standard version of the van Trees argument presented by \citep{gill1995applications} does not allow the inclusion of a data-dependent event $\mathbf{1}(\calE)$ or an estimator that depends on the unknown parameter $\Theta$, as with our choice of the estimator $K_{\calD}^{y,u}$ in \eqref{eq: Kdyu}. We present an extension that covers these cases below.

\begin{lemma}[van Trees inequality with an event]
\label{lem: van trees}
Let $g(\calD,\Theta)\in \mathbb R^m$ be differentiable in $\Theta\in\mathbb R^d$, and let
$\calE$ be $\calD$-measurable. Let
$\calI_{\lambda} \triangleq \int\frac{\nabla \lambda(\theta)}{\lambda(\theta)}\frac{\nabla \lambda(\theta)}{\lambda(\theta)}^\top \lambda(\theta) \; d\theta$
and $J := N \E_{\Theta}\mathsf{FI}^{\pi_{\exp}}(\Theta)+\calI_{\lambda},$
and assume the usual regularity conditions for integration by parts, including vanishing boundary terms. Define $
    \tilde g
    :=
    \mathbf E_{\Theta,\calD}\brac{
        D_{\Theta}g(\calD,\Theta)\mathbf 1(\calE)
    }.$
Then $\mathbf{E}_{\Theta,\calD}\brac{\norm{g(\calD,\Theta)}^2}
    \geq
    \trace\paren{
        \tilde g J^{-1}\tilde g^\top
    } .$
\end{lemma}

\begin{proof} We drop the subscript on the expectation for simplicity. 
Let $S(\calD,\Theta)
    :=
    \nabla_{\Theta}\log p(\calD,\Theta)$
denote the joint score. Under the stated model,
\begin{align}
    \label{eq: info shorthand}
    \mathbf E \brac{SS^\top}
    =
    N \E_{\Theta}\mathsf{FI}^{\pi_{\exp}}(\Theta)+\calI_{\lambda}
    =J,
\end{align}
where $\calI_{\lambda}$ measures the prior density concentration.
Set $X=g(\calD,\Theta),
    Y=\mathbf 1(\calE)S(\calD,\Theta).$
Since $\mathbf 1(\calE)^2\leq 1$, $\mathbf E[YY^\top]
    =
    \mathbf E\brac{\mathbf 1(\calE)SS^\top}
    \preceq
    \mathbf E\brac{SS^\top}
    =
    J.$
Integration by parts gives $\mathbf E\brac{XY^\top} = -\tilde g$. 
Applying matrix Cauchy--Schwarz \citep{tripathi1999matrix} concludes the proof.
\end{proof}
We apply this to the first term in the lower bound of \eqref{eq: calE definition}.
\begin{lemma}[Application of van Trees inequality]
    \label{lem: vt bound}
    Let $\calA$ be any algorithm mapping a dataset $\calD$ to a policy $K_{\calD} \in \calK$ for a policy class $\calK$  that obeys obeys Assumption~\ref{asmp: closed loop}. Suppose all  instances in the ball $B(\theta^\star, \varepsilon)$ satisfy the stabilizability and detectability assumptions of \Cref{sec: problem formulation}. Suppose that the data is collected under an exploration policy $\pi_{\exp}$ such that $\mathsf{FI}^{\pi_{\exp}}(\theta)$ is Lipschitz continuous in $\theta$ over $\calB(\theta^\star, \varepsilon)$ and $\mathsf{FI}^{\pi_{\exp}}(\theta^\star) \succ 0$. There exist constants depending on the system, exploration policy, and length of the exploration episodes, $c_{\varepsilon}^{VT}(\theta^\star, \pi_{\exp}, T, \alpha)$ and $c_N(\theta^\star, \pi_{\exp}, T)$ such that for $\varepsilon \leq  c_{\varepsilon}^{VT}(\theta^\star, \pi_{\exp}, T, \alpha)$ and $N \geq  c_N(\theta^\star, \pi_{\exp}, T) \varepsilon^{-2}$, the following bound holds:
    \begin{align*}
        &\E_{\Theta, \calD} \brac{\norm{\Psi^{1/2} (K_{\calD}^{y,u} - K_{\Theta}^{y,u}) Z_{\theta^\star}^{K_{\calD}}}_{\calH_2}^2  \mathbf{1}(\calE)} \geq \mathsf{vTLB},
    \end{align*}
    where
    $Z_{\theta}^K \!=\! \bmat{I \\ K} \! \calF(\tilde \calP_{\theta}, K)$ and 
    $\mathsf{vTLB} \triangleq \frac{\trace\paren{\!H \paren{\mathsf{FI}^{\pi_{\exp}}}^{-1} \!}}{2 \sqrt{2} N}.$ 
\end{lemma}

\label{s: van trees proof}

\begin{proof}
    Let $\underline \Psi$ lower bound $\Psi(\theta)$ for all $\theta \in \calB(\theta^\star, \varepsilon)$ in Loewner order. 
    It holds by definition of the $\calH_2$ norm that
    \begin{align*}
         &\norm{ \underline \Psi^{1/2} (K_{\calD}^{y,u} - K_{\Theta}^{y,u}) Z_{\theta^\star}^{K_{\calD}}}_{\calH_2}^2= \\
         &\frac{1}{2\pi} \!\int_{-\pi}^\pi\!\norm{ \!\paren{\underline \Psi^{1/2} \!(K_{\calD}^{y,u} \!-\! K_{\Theta}^{y,u}) Z_{\theta^\star}^{K_{\calD}} \!}(e^{j\omega})}_F^2 d\omega. %
    \end{align*}
    We now apply \Cref{lem: van trees} to the integrand of the above expression multiplied by $\mathbf{1}(\calE)$. Note that by the fact that $K_{\calD}^{y,u} \bmat{I \\K_{\calD}} = K_{\calD }$, the Jacobian of the vectorized quantity in the above norm with respect to $\theta$ is 
    $D_{\theta} \VEC K_{\theta}^{y,u}(e^{j\omega})\vert_{\theta=\Theta} \paren{\underline \Psi^{1/2} \otimes Z_{\theta^\star}^{K_{\calD}}(e^{j\omega})}.$
    Consequently, by \Cref{lem: van trees} it holds that for any $\omega \in [-\pi, \pi]$ 
    \begin{align*}
    & \E \brac{\mathbf{1}(\calE)\norm{\underline \Psi^{1/2} (K_{\calD}^{y,u} - K_{\Theta}^{y,u})(e^{j\omega}) Z_{\theta^\star}^{K_{\calD}}(e^{j\omega}) }_F^2}\geq \\
     & \norm{ J^{-\frac{1}{2}} \E\brac{D_{\theta}\! \VEC K_{\theta}^{y,u}(e^{j\omega})\vert_{\theta=\Theta}  \mathbf{1}(\calE) \paren{\underline \Psi^{1/2} \!\otimes \! Z_{\theta^\star}^{K_{\calD}}(e^{j\omega})}}}_F^2, 
    \end{align*}
    with $J$ as in \eqref{eq: info shorthand}. 
    The integral of the above quantity over $\omega$ results in an expression of the form 
    \begin{align*}
        \norm{\E \brac{\mathcal{G}_{\Theta} (\underline \Psi^{1/2} \otimes Z_{\theta^\star}^{K_{\calD}}) \mathbf{1}{(\calE)}}}_{\calH_2}^2.
    \end{align*}
    By a triangle inequality, it holds that  
    \begin{align*}
        &\norm{\E \brac{\mathcal{G}_{\Theta} (\underline \Psi^{1/2} \otimes Z_{\theta^\star}^{K_{\calD}}) \mathbf{1}{(\calE)}}}_{\calH_2} \\
        &\geq  \norm{\E \brac{\calG_{\theta^\star} (\Psi(\theta^\star)\otimes Z_{\theta^\star}^{K_{\theta^\star}} \mathbf{1}{(\calE)}}}_{\calH_2} \\
        &-\norm{\E \brac{\calG_{\Theta} (\underline \Psi \otimes (Z_{\theta^\star}^{K_{\calD}} - Z_{\theta^\star}^{K_{\theta^\star}})) \mathbf{1}{(\calE)}}}_{\calH_2} \\
        &- \norm{\E \brac{\paren{\calG_{\Theta} (\underline \Psi \otimes Z_{\theta^\star}^{K_{\theta^\star}})) -\calG_{\theta^\star} (\Psi(\theta^\star) \otimes Z_{\theta^\star}^{K_{\theta^\star}}))}\mathbf{1}{(\calE)}}}_{\calH_2}.
    \end{align*}
    We upper bound the second term by observing that
    \begin{align*}
        Z_{\theta^\star}^{K_{\calD}} - Z_{\theta^\star}^{K_{\theta^\star}} = \bmat{N_{\theta^\star} \\ M_{\theta^\star}}  Q_{\theta^\star}^{K_{\calD}} \tilde M_{\theta^\star} \calP_{\theta^\star}^d.
    \end{align*}
    Consequently, the second term is upper bounded as
    \begin{align*}
        &\norm{\E \brac{\calG_{\Theta}(Z_{\theta^\star}^{K_{\calD}} - Z_{\theta^\star}^{K_{\theta^\star}}) \mathbf{1}{(\calE)}}}_{\calH_2} \\
        &\leq L(\theta^\star) \max_{\theta \in \calB(\theta^\star, \varepsilon)} \norm{\calG(\theta)}_{\calH_{\infty}} \E\brac{\norm{\calQ_{\theta^\star}^{K_{\calD}}}_{\calH_2} \mathbf{1}(\calE)},
    \end{align*}
    where $L(\theta^\star)$ bounds the $\calH_\infty$ norms of the matrices $N_{\theta^\star}$, $M_{\theta^\star}$, and $\tilde M_{\theta^\star} \calP_{\theta^\star}^d$. The above quantity 
    can be made small for $\varepsilon$ small and $N$ large by mirroring the analysis following \eqref{eq: negative term upper bound}.
    In particular,  substitute the bound on $\norm{\calQ_{\theta^\star}^{K_{\calD}}}$ under the event $\calE$ and observe that $\max_{\theta \in \calB(\theta^\star,\varepsilon)}\norm{\calG_{\theta}}_{\calH_{\infty}}$ decays at a rate $\frac{1}{\sqrt{N}}$ (due to $J^{-1/2})$. We can therefore select an instance dependent upper bound on $\varepsilon$ and a lower bound on $N$ such that the second term becomes negligible relative to the first. The third term can be upper bounded by observing that $\Psi(\theta)$ and $D_{\theta} \VEC K_{\theta}^{y,u}$ are smooth with respect to $\theta$. Specifically, $K_{\theta}^{y,u} (e^{j\omega})  = F(\theta) (e^{j\omega} - A_{\mathsf{cl}}^o(\theta))^{-1} \bmat{L(\theta) & B(\theta)}.$ 
    For any frequency $\omega \in [-\pi,\pi]$,
    we can express the resulting derivative in terms of the derivatives of the system matrices, the LQR gain, and the Kalman gain. By the stability of $A_{\mathsf{cl}}^o(\theta^\star)$ and the smoothness of the derivatives $D_{\theta} F(\theta)$ and $D_{\theta}L(\theta)$ from their characterizations in Lemma B.1 of \citep{simchowitz2020naive} along with the perturbation arguments of Theorem 5 of \citep{simchowitz2020naive}, it holds that for any frequency $\omega \in [-\pi, \pi]$, $D_{\theta} \VEC K_{\theta}^{y,u}(e^{j\omega})$ is a smooth function of $\theta$ for a neighborhood around $\theta^\star$. %
    Therefore the impact of the discrepancy between $\underline \Psi$ and $\Psi(\theta^\star)$ and $\calG_{\Theta}$ and $\calG_{\theta^\star}$ can be made arbitrarily small relative to the first term by taking $\varepsilon$ sufficiently small and $N$ sufficiently large. 
    
    We now lower bound the term $ \norm{\E \brac{\calG_{\theta^\star}(\Psi(\theta^\star) \otimes Z_{\theta^\star}^{K_{\theta^\star}}) \mathbf{1}{(\calE)}}}_{\calH_2}$.
     It is shown by Corollary 2.1  of \citep{lee2023fundamental} that $\norm{\calI_{\lambda}} = \frac{1}{\varepsilon^2} \norm{\calI_{\rho}}$. Therefore, there exists a constant $c_N(\theta^\star, \pi_{\exp}, T)$ such that for $N \geq  c_N(\theta^\star, \pi_{\exp}, T) \varepsilon^{-2}$, the contribution of $\calI_{\lambda}$ to $J$ can be made arbitrarily small. %
    By the assumption that $\mathsf{FI}^{\pi_{\exp}}(\theta)$ is a smooth function of $\theta$ for over $\calB(\theta^\star,\varepsilon)$ 
    and since 
    $\mathsf{FI}^{\pi_{\exp}}(\theta^\star) \succ 0$, there exists a constant $ c_{\varepsilon}(\theta^\star, \pi_{\exp}, T)$ such that for  $\varepsilon \leq c_{\varepsilon}(\theta^\star, \pi_{\exp}, T)$, the impact of the gap between $\E_{\Theta} \mathsf{FI}^{\pi_{\exp}}(\Theta)$ and $\mathsf{FI}^{\pi_{\exp}}(\theta^\star)$ can also be made arbitrarily small. 
         Let $c_{\varepsilon}(\theta^\star, \pi_{\exp}, T)$ be small enough and $c_N(\theta^\star, \pi_{\exp}, T)$ large enough that the slack introduced by from the above substitutions is less than $\frac{1}{2^{1/4}}$. Then 
    \begin{align*}
        &\E_{\Theta, \calD} \brac{\norm{ \underline \Psi^{1/2} (K_{\calD}^{y,u} - K_{\Theta}^{y,u}) Z_{\theta^\star}^{K_{\calD}}}_{\calH_2}^2 \mathbf{1}(\calE)} \geq \frac{\bfP[\calE]^2}{2^{1/4}}\\
        &\times \!\norm{D_{\theta} \VEC K_{\theta}^{y,u}\vert_{\theta=\theta^\star} (\Psi(\theta^\star)^{\frac{1}{2}} \!\otimes \!Z_{\theta^\star}^{K_{\theta^\star}}) (N \mathsf{FI}^{\pi_{\exp}}(\theta^\star))^{-\frac{1}{2}}}_{\calH_2}^2\!\!.
    \end{align*}
    From \Cref{lem: performance difference lemma} along with the definition of the $\calH_2$ norm,
    \begin{align*}
        &\frac{1}{2\pi}\int_{-\pi}^\pi D_{\theta} \VEC K_{\theta}^{y,u}(e^{j\omega})\vert_{\theta=\theta_\star}\paren{\Psi(\theta^\star) \otimes \Sigma_{\theta^\star}^{K_{\theta^\star}}(e^{j\omega})} \\
        &\times D_{\theta}\VEC K_{\theta}^{y,u}(e^{j\omega})^*\vert_{\theta=\theta^\star}\; d\omega = \frac{1}{2} H(\theta^\star). 
    \end{align*}
    Making this substitution, we find that
    \begin{equation}
    \begin{aligned}
        \label{eq: vt partial}
        &\E_{\Theta, \calD} \brac{\norm{ \underline \Psi^{1/2} (K_{\calD}^{y,u} - K_{\Theta}^{y,u}) Z_{\theta^\star}^{K_{\calD}}}_{\calH_2}^2 \mathbf{1}(\calE)} \\ %
        &\geq \frac{\bfP[\calE]^2}{2^{1/4}} \frac{1}{2}
        \trace\paren{ H(\theta^\star) \paren{N \mathsf{FI}^{\pi_{\exp}}(\theta^\star)}^{-1}}.
    \end{aligned}
    \end{equation}
    It remains only to lower bound $\Pr[\calE]$. 

    \textbf{Lower bound of $\bfP[\calE]$: } 
    Restrict attention to learning algorithms for which 
    \begin{align*}
        \E_{\Theta, \calD} \brac{J(K_{\calD}, \Theta) - J(K(\Theta), \Theta)} \leq \frac{\trace\paren{H(\theta^\star) \mathsf{FI}^{\pi_{\exp}}(\theta^\star)^{-1}}}{4N},
    \end{align*}
    as other algorithms immediately satisfy the lower bound. By the quadratic lower bound of \Cref{lem: performance difference lemma}, the performance gap upper bounds the squared controller parameter error: 
    \begin{align*}
        \E_{\Theta, \calD} \brac{J(K_{\calD}, \Theta) - J(K(\Theta), \Theta)} \geq \mu \E_{\Theta, \calD}\brac{\norm{\calQ_{\Theta}^{K_{\calD}}}_{\calH_2}^2}.
    \end{align*} 
    By the triangle inequality, it holds that for any $\theta \in \calB(\theta^\star, \varepsilon)$, 
    \begin{align*}
        \norm{\calQ_{\theta^\star}^{K_{\calD}}}_{\calH_2}^2 \leq  2\norm{\calQ_{\theta}^{K_\calD}}_{\calH_2}^2 + 2\norm{\calQ_{\theta}^{K_{\calD}} - \calQ_{\theta^\star}^{K_{\calD}}}_{\calH_2}^2.
    \end{align*}
    From the perturbation result of \Cref{lem: smoothness of youla parameter}, there exists an $c_{\varepsilon}^\calQ(\theta^\star, \alpha)$ such that if $\varepsilon \leq c_{\varepsilon}^\calQ(\theta^\star, \alpha)$, then for any $\theta \in \calB(\theta^\star, \varepsilon)$, $2\norm{\calQ_{\theta}^{K_{\calD}} - \calQ_{\theta^\star}^{K_{\calD}}}_{\calH_2}^2 \leq 2 L_\calQ(\theta^\star, \alpha)^2 \varepsilon^2$.  By selecting $\theta$ as the value that minimizes $\norm{\calQ_{\theta}^{K_{\calD}}}_{\calH_2}$ it holds with probability at least $1-\delta$ that
    \begin{align*}
        2\norm{\calQ_{\theta}^{K_{\calD}}}_{\calH_2}^2 \leq 2\norm{\calQ_{\Theta}^{K_{\calD}}}_{\calH_2}^2 &\leq 2 \frac{\E_{\Theta, \calD}\brac{\norm{\calQ_{\Theta}^{K_{\calD}}}_{\calH_2}^2}}{\delta}\\
        & 
        \leq \frac{\trace\paren{H(\theta^\star) \mathsf{FI}^{\pi_{\exp}}(\theta^\star)^{-1}}}{ TN \mu \delta},
    \end{align*}
    where the second inequality follows from Markov's inequality. Combining these inequalities, $\calE$ from \eqref{eq: calE definition} holds with probability at least $\frac{1}{2^{1/8}}$. To conclude the proof of \Cref{lem: vt bound}, we can substitute this bound into \eqref{eq: vt partial}, and let $ c_{\varepsilon}^{VT}(\theta^\star, \pi_{\exp}, T, \alpha)$ be the minimum of the aforementioned sufficient upper bounds on $\varepsilon$. %
\end{proof}

\ifTACmode\else
    \ifTACmode
\section{Non-strictly Causal Extension}
\else
\subsection{Non-strictly causal extension}
\fi
\label{s: non-strictly causal}
The optimal non-strictly causal LQG controller has the following form
\begin{align}
    \label{eq: lqr control_non_strictly}
    u^\theta_t = F(\theta) \hat x^\theta_{t \vert t},
\end{align}
where $\hat{x}^\theta_{t\vert t}\triangleq \bfE_{\theta} \brac{x_t \vert y_{0:t}, u_{0:t-1}}$ is the Kalman estimate for the state $x_t$ given the full history \emph{including} time $t$. Let the \emph{modified} Kalman gain $\bar L(\theta)$ be given by
\[
  \bar L(\theta) = \Sigma(\theta) C(\theta)^\top \Sigma_e(\theta)^{-1}.
\]
Notice that $A(\theta)\bar L(\theta)=L(\theta)$. Then, the Kalman filter estimate $\hat{x}^{\theta}_{t|t}$ is given by
\begin{align}   \label{eq: kf update_posterior} 
    \hat x^\theta_{t\vert t} &=  \hat x^\theta_{t \vert t-1} + \bar L(\theta)(y_t-C(\theta)\hat x^\theta_{t \vert t-1}).
\end{align}
Thus, under any control input $\bfu$ and output $\bfy$, the  representation of the Kalman filter estimate in the frequency domain becomes
\begin{align*}    
   \bfx^\theta&=(I-\bar L(\theta)C(\theta))(zI-A^o_{\mathsf{cl}}(\theta))^{-1}(B(\theta)\bfu+L(\theta)\bfy)\\&+\bar L(\theta)\bfy,
\end{align*}
while the optimal LQG controller is equal to
\[
K_\theta=F(\theta)(I-\bar{L}(\theta)C(\theta))(zI-\bar{A}(\theta))^{-1}A^c_{\mathsf{cl}}\bar{L}(\theta)+F(\theta)\bar{L}(\theta)
\]
where
\[
\bar A(\theta)=A_{\mathsf{cl}}^o(\theta)+B(\theta)F(\theta)(I-\bar L(\theta)C(\theta)).
\]
We can repeat the same steps as in the strictly causal case.
\subsubsection{Co-prime factorizations}
Consider the doubly co-prime factorizations
\begin{equation*}
\begin{aligned}
   \bmat{M_{\theta} & U_\theta \\ N_\theta & V_\theta} &\triangleq \left[\begin{array}{c|cc}
              A_{\mathsf{cl}}^c(\theta) & B(\theta) & L(\theta)+B(\theta)F(\theta)\bar L(\theta) \\ \hline \rule{0pt}{2.6ex}
              F(\theta) & I & F(\theta)\bar{L}(\theta) \\ C(\theta) & 0 & I 
            \end{array}\right] \\
     \bmat{\tilde V_{\theta} & -\tilde U_\theta \\ -\tilde N_\theta & \tilde M_{\theta}} &\triangleq \left[\begin{array}{c|cc}
             A_{\mathsf{cl}}^o(\theta) & -B(\theta) & -L(\theta) \\ \hline \rule{0pt}{2.6ex}
              F(\theta)I_{\bar LC}(\theta) & I & -F(\theta)\bar{L}(\theta) \\ C(\theta) & 0 & I 
            \end{array}\right],
\end{aligned}
\end{equation*}
where $I_{\bar L C}(\theta)=I-\bar L(\theta)C(\theta)$.
We can verify that this parameterization remains doubly co-prime while
\[
K_\theta=\tilde V_\theta^{-1}\tilde U_\theta=U_\theta V_\theta^{-1}.
\]
Note that $M_\theta,N_\theta,\tilde M_\theta,\tilde N_{\theta}$ remain unchanged from~\eqref{eq: nominal co-prime}.
\subsubsection{Extension of main result} The only part of the proof that changes is the intermediate expression of $\bfu^\theta$ in the performance difference lemma. Following the notation of the proof of Lemma~\ref{lem: performance difference lemma}, we have
\begin{align*}
\bfu^\theta&=F(\theta)\bfx^\theta=F(\theta)(I-\bar L(\theta)C(\theta))(zI-A^o_{\mathsf{cl}}(\theta))^{-1}(B(\theta)\bfu\\&+L(\theta)\bfy)+F(\theta)\bar L(\theta)\bfy\\
&=\tilde{U}_\theta \bfy+(I-\tilde{V}_\theta)\bfu=\bmat{\tilde U_\theta & I - \tilde V_\theta} \bmat{I \\ K} \calF(\tilde \calP_{\theta}, K)\bfy.
\end{align*}
Everything else remains the same. As a result of using a different co-prime factorization, the system specific constants in Lemmas~\ref{lem: coprime perturbations}-\ref{lem: closeness of LQG solutions},~\ref{lem: closed loop perturbations}, Theorem~\ref{thm: excess cost lower bound} will be different from those of the strictly causal case. 

\fi
\section{Conclusion}

This work provides a Cramer-Rao style lower bound for the excess cost achieved by an LQG controller learned from an offline set of interaction data from the system. The lower bound is expressed in terms of the Fisher Information for the offline dataset and the Hessian of the control objective under a certainty equivalent control policy, with respect to the system parameter estimate defining the certainty equivalent policy. Classic examples of fragile systems are examined in light of this lower bound, demonstrating when classically fragile control problems translate into problems of learning-enabled control with a large sample complexity. 

The results suggest the asymptotic optimality of certainty-equivalent based synthesis policies, and also suggest that an objective for task-weighted optimal experiment design used for the fully observed case by \citet{wagenmaker2021task} is also effective in the partially observed setting.  By highlighting the potential tradeoff between the cost of control and the excess cost of learning to control, this work additionally suggests a theoretically-informed manner of system co-design for learning-enabled control. Such a strategy is previewed in \Cref{s: tradeoffs}. Further investigation is left for future work. Finally, this work may provide an effective way to construct benchmarks for reinforcement learning applied to continuous control tasks with a rigorous understanding of the complexity of the benchmark.

\section*{REFERENCES}
\bibliographystyle{IEEEtran}
\bibliography{refs}

\ifTACmode
\raggedbottom
\makeatletter
\def\@IEEEBIOskipN{1\baselineskip}
\makeatother
\begin{IEEEbiography}[{\includegraphics[width=1in,height=1.25in,clip,keepaspectratio]{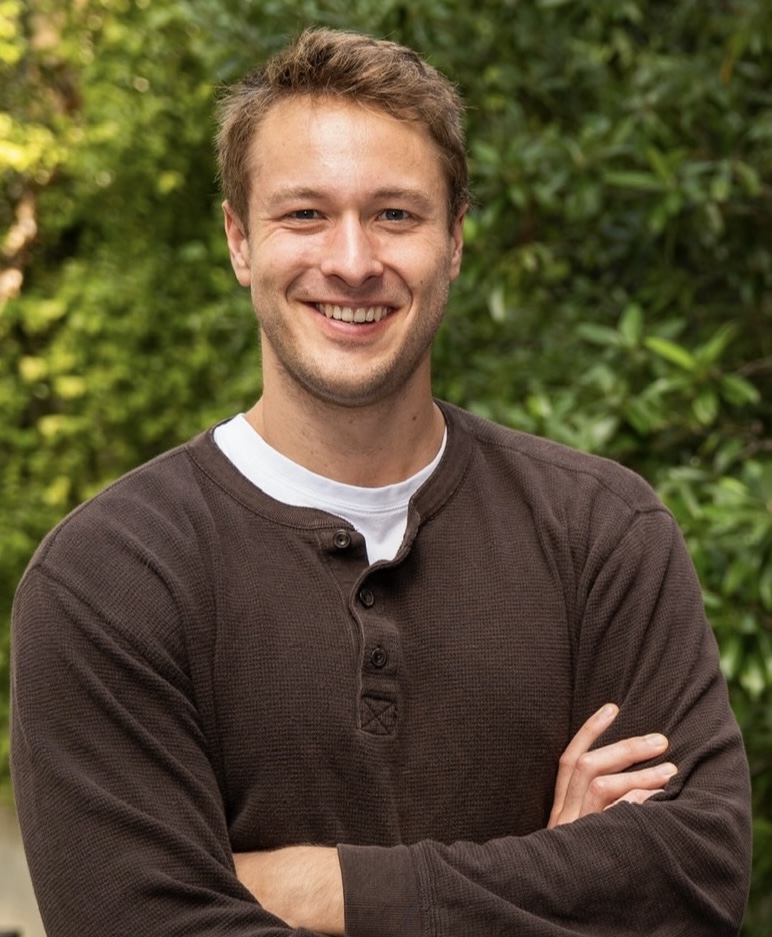}}]{Bruce D. Lee} is a Postdoctoral Fellow with the ETH AI Center. He is also affiliated with the Learning and Adaptive Systems Group, Institute for Automatic Control, and the Institute for Dynamic Systems and Control at ETH Zürich. He received his Ph.D. in Electrical and Systems engineering from the University of Pennsylvania in June, 2025. He also holds a B.E.E degree in Electrical Engineering from the University of Minnesota.
His research is at the intersection of reinforcement learning, robotics, and control. Bruce is recipient of the DoD NDSEG Doctoral Fellowship, and the ETH AI Center Postdoctoral Fellowship. His work was recognized as a Best Paper Finalist at the Sixth Annual Learning for Dynamics \& Control Conference. 
\end{IEEEbiography}

\begin{IEEEbiography}[{\includegraphics[width=1in,height=1.25in,clip,keepaspectratio]{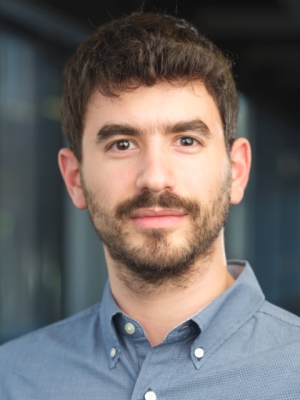}}]{Anastasios Tsiamis} is a Senior Scientist at the Automatic Control Laboratory, ETH Zurich, Switzerland. He obtained his Ph.D. at the Department of Electrical and Systems Engineering, University of Pennsylvania, USA in 2022. He received the Diploma degree in Electrical and Computer Engineering from the National Technical University of Athens, Greece, in 2014. His research interests include statistical and online learning in the setting of control systems, as well as robust and risk-aware control. Anastasios Tsiamis was a finalist for the IFAC Young Author Prize in IFAC 2017 World Congress and a finalist for the Best Student Paper Award in ACC 2019. He is a coauthor to the paper that has won the Best Student Paper Award in CDC 2022.
\end{IEEEbiography}

\begin{IEEEbiography}[{\includegraphics[width=1in,height=1.25in,clip,keepaspectratio]{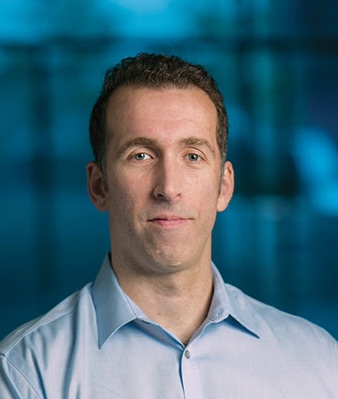}}]{Nikolai Manti} is an Associate Professor in the Department of Electrical and Systems Engineering at the University of Pennsylvania, where he is also a member of the GRASP Lab and the PRECISE Center. 
Prior to joining Penn, Nikolai was a postdoctoral scholar in EECS at UC Berkeley. He also held a postdoctoral position in the Computing and Mathematical Sciences Department at Caltech. He received his Ph.D. in Control and Dynamical Systems from Caltech in June 2016. He also holds B.A.Sc. and M.A.Sc. degrees in Electrical Engineering from the University of British Columbia in Vancouver, Canada.
His research interests broadly encompass the use of learning, optimization, and control in the design and analysis of safety-critical, data-driven autonomous systems. Nikolai is a recipient of the AFOSR YIP Award (2024), NSF CAREER Award (2021), a Google Research Scholar Award (2021), the 2024 IEEE Transactions on Control of Network Systems Best Paper Award, the 2021 IEEE CSS George S. Axelby Award, and the 2013 IEEE CDC Best Student Paper Award.  He is also a co-author on papers that have won the 2022 IEEE CDC and the 2017 IEEE ACC Best Student Paper Award.
\end{IEEEbiography}

\begin{IEEEbiography}[{\includegraphics[width=1in,height=1.25in,clip,keepaspectratio]{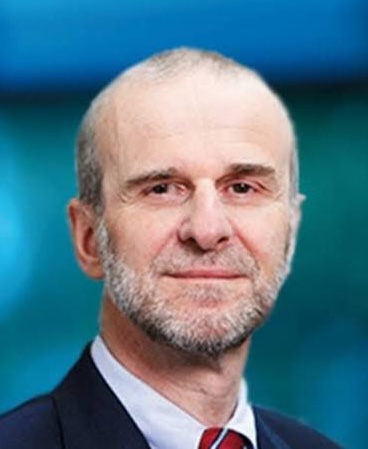}}]{Manfred Morari} Manfred Morari received the Diploma degree in chemical engineering from ETH Zürich, Zürich, Switzerland, and the Ph.D. degree in chemical engineering from the University of Minnesota, Minneapolis, MN, USA.
He is a Professor Emeritus at ETH Zürich, where he previously seved as Head of the Department of Information Technology and Electrical Engineering. He also served as a Professor of Electrical and Systems Engineering at the University of Pennsylvania, and the McCollum–Corcoran Professor of Chemical Engineering and Executive Officer of Control and Dynamical Systems at the California Institute of Technology (Caltech), Pasadena, CA, USA. Earlier, he was a Professor at the University of Wisconsin–Madison, WI, USA. 
He supervised more than 80 Ph.D. students. Dr. Morari is a Fellow of AIChE, IFAC, and the U.K. Royal Academy of Engineering, and a member of the U.S. National Academy of Engineering. He has received numerous awards, including the Eckman, Ragazzini, and Bellman Awards from the American Automatic Control Council (AACC); the Colburn, Professional Progress, and CAST Division Awards from the American Institute of Chemical Engineers (AIChE); the Control Systems Technical Field Award and the Bode Lecture Prize from IEEE; the Nyquist Lectureship and the Oldenburger Medal from the American Society of Mechanical Engineers (ASME); and the IFAC High Impact Paper Award. He also served as President of the European Control Association and on the technical advisory boards of several major corporations.
\end{IEEEbiography}

\begin{IEEEbiography}[{\includegraphics[width=1in,height=1.25in,clip,keepaspectratio]{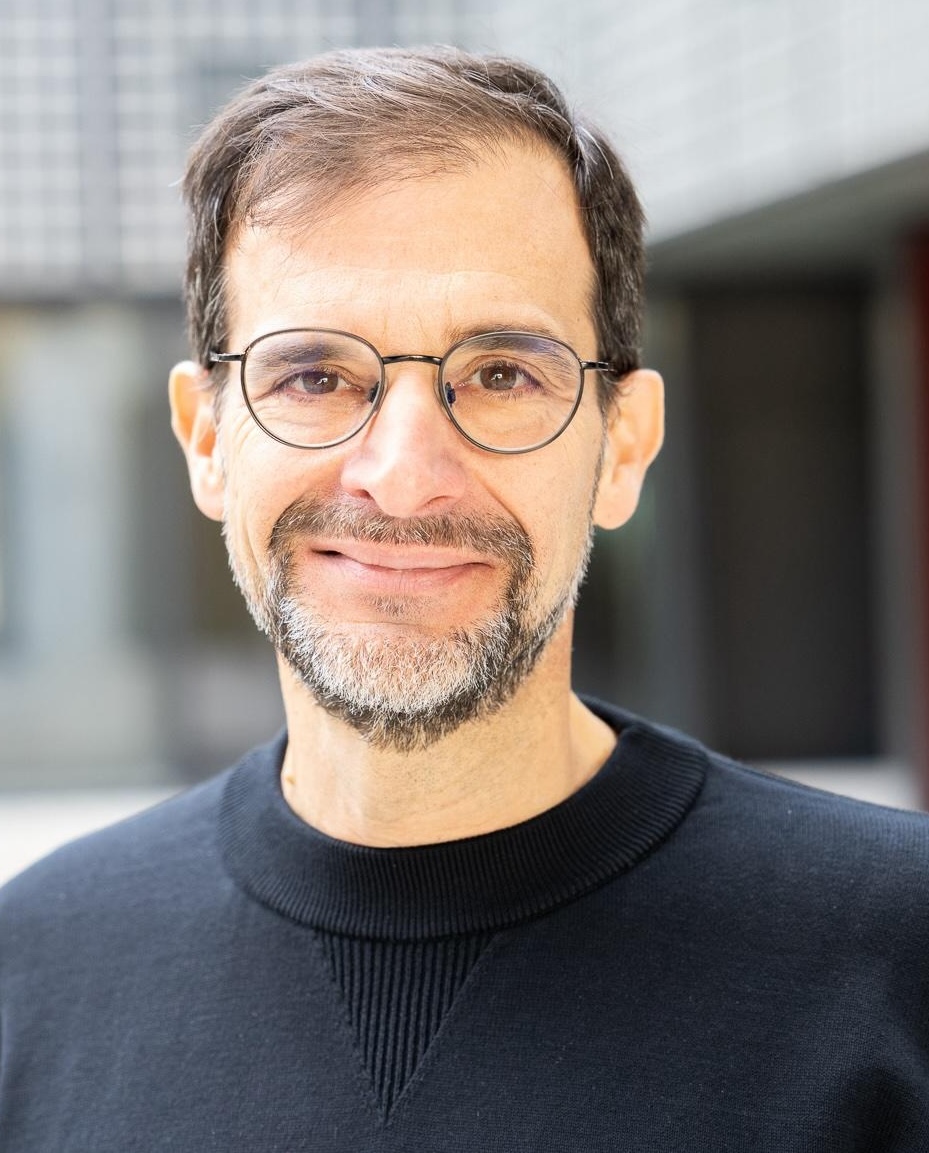}}]{John Lygeros} completed a B.Eng. degree in electrical engineering in 1990 and an M.Sc. degree in Systems Control in 1991, both at the Imperial College of Science Technology and Medicine, London, U.K.. In 1996 he obtained a Ph.D. degree from the Electrical Engineering and Computer Sciences Department, University of California, Berkeley. In the period 1996-2000 he held a series of research appointments at the National Automated Highway Systems Consortium, M.I.T., and U.C. Berkeley. In parallel, he also worked as a part-time research engineer at SRI International, Menlo Park, California, and as a Visiting Professor at the Department of Mathematics of the Université de Bretagne Occidentale, Brest, France. Between July 2000 and March 2003 he was a University Lecturer at the Department of Engineering, University of Cambridge, U.K., and a Fellow of Churchill College. Between March 2003 and July 2006 he was an Assistant Professor at the Department of Electrical and Computer Engineering, University of Patras, Greece. In July 2006 he joined the Automatic Control Laboratory at ETH Zurich where he is currently serving as the Professor for Computation and Control and the Head of the laboratory.
His research interests include modeling, analysis, and control of large scale dynamical systems, with applications to biochemical networks, energy systems, transportation, and advanced manufacturing. John Lygeros is a Fellow of the IEEE, and a member of the IET and the Technical Chamber of Greece. Between 2012 and 2015 he served in the Board of Governors of the IEEE Control System Society and in the Scientific Steering Committee of the Newton Institute, while between 2015 and 2018 he served as the Head of the Department of Infomation Technology and Electrical Engineering of ETH Zurich. Between 2013 and 2023 he served as the Vice-President Finances and a Council Member of the International Federation of Automatic Control (IFAC) and as a Board Member of the IFAC Foundation. Since 2020 he is serving as the Director NCCR Automation.
\end{IEEEbiography}
\fi

\ifTACmode
    \clearpage
    \appendices
    \section*{Supplementary Material for ``The Fragility of Learning LQG Controllers''}
    \section{Additional Examples}
    \subsection{Variants of Doyle's Example}

    \section{Recovering the LQR Setting}
\label{s: lqr setting}
The above characterizations do not immediately recover results from the LQR setting from \citep{lee2023fundamental}, as they assume strictly causal policies. However, the characterizations can also be presented in the causal setting.  
\begin{proposition}[Causal Hessian Characterization]
    Consider the causal setting, and let $\theta$ be a scalar parameter. The Hessian characterization is identical to \Cref{prop: hessian}, with the substitutions:
    \begin{align*}
        \bar L &\triangleq  \Sigma C^\top \Sigma_e^{-1},\,\, I_{\bar LC} \triangleq I - \bar LC\\
         \dot {\bar L} &\triangleq  \paren{I_{\bar LC} \Sigma \dot C^\top + I_{\bar LC} \dot \Sigma  C^\top  - \bar L \dot C \Sigma C^\top} \Sigma_e^{-1} 
        \end{align*}
        and \begin{align*} 
        A_{\mathsf{joint}}^H &\gets \bmat{A+BF & 0 \\ I_{\bar LC}(\dot A + \dot B F) - \bar L \dot C (A+BF) & I_{\bar LC} A} \\
        B_{\mathsf{joint}}^H &\gets \bmat{\bar L \\ \dot {\bar{L}}}
    \end{align*}
\end{proposition}
The proof is the same as the one of~\Cref{prop: hessian}. 
Notably, we recover the fully observed setting of \citep{lee2023fundamental} from the above proposition by letting $C=I$ and taking $\Sigma_v \to 0$. Then, the Hessian becomes $H = 2 \trace\paren{\Psi \dot F \dlyap(A+BF, \Sigma_w) \dot F^\top}$. In \citep{lee2023fundamental}, this characterization was used to derive examples of challenging fully observed problems, including an example where the lower bound on the minimax excess cost is exponential in the state dimension. In the sequel, we present examples of challenging problems enabled by characterizations of the partially observed setting.

    \section{Non-minimum phase example}
Consider system
\[
A=\begin{bmatrix}
1& 1\\\theta&1
\end{bmatrix},\,B=\begin{bmatrix}
0\\1
\end{bmatrix},\,C=\begin{bmatrix}
-\nmpparam &1
\end{bmatrix}.
\]
The system has a non-minimum phase zero at $1+\nmpparam$.
Let $\Sigma_w=I$, $\Sigma_v=1$, $Q=I$, $R=1$. Let $\theta^{\star}=0$. 
\begin{lemma}[Difficulty of non-minimum phase example]\label{lem:example_nmp_difficulty}
    Recall the definition of the minimax excess cost $\mathcal{M}(\theta^\star,\varepsilon,\mathcal{A})$. We have the following lower bounds \begin{align*}
        \liminf_{T\rightarrow \infty} T\mathcal{M}(\theta^\star,\varepsilon,\mathcal{A})&\ge \frac{1}{N}O(\nmpparam^{-4})\\
         \liminf_{T\rightarrow \infty} \frac{T\mathcal{M}(\theta^\star,\varepsilon,\mathcal{A})}{J^\star(\theta^\star)}&\ge \frac{1}{N}O(\nmpparam^{-1}),
    \end{align*}
    where $J^\star(\theta^\star)=\min_{\pi\in\Pi}J(\pi,\theta)$.
\end{lemma}
We observe that the dominant term  of the excess cost blows up as the system becomes unobservable and we approach a pole zero cancellation. The same is true relative to the optimal LQG cost.
\subsection{Computation of Riccati}
First, we characterize the Riccati solution as a function of $\nmpparam$, as $\nmpparam$ goes to zero.

\begin{lemma}[Riccati rate]\label{lem:example_nmp_Riccati_rates}
    Let $\Sigma=\begin{bmatrix} p_1&p\\p&p_2\end{bmatrix}$ be the solution to the Riccati equation for the Kalman filter. Define $\sigma^2_e=C\Sigma C^\top+\Sigma_v$, $\sigma_e>0$ and let $r=\frac{1+\sqrt{5}}{2}$. Then, we have
    \begin{equation}\label{eq:example_nmp_asymptotic_Riccati}
    \begin{aligned}
      \tilde p_1&\triangleq  \nmpparam^3 p_1=2+(4r+1)\nmpparam+O(\nmpparam^2)\\
      \tilde p&\triangleq  \nmpparam^2 p=2+(2r+1)\nmpparam+O(\nmpparam^2)\\
      \tilde p_2&\triangleq  \nmpparam p_2=2+(r+1)\nmpparam+O(\nmpparam^2)\\
      \sigma_e&=r+\frac{r^2}{2r-1}\nmpparam+O(\nmpparam^2).
    \end{aligned}
    \end{equation}
\end{lemma}

\subsection{Computation of Hessian}
The Kalman gain and the respective closed-loop matrix are
\begin{align}
    L&=A\Sigma C^\top (C\Sigma C^\top+\Sigma_v)^{-1}\nonumber\\
    &=\begin{bmatrix} 1&1\\0&1\end{bmatrix}\begin{bmatrix} -\nmpparam p_1+p\\-\nmpparam p+p_2\end{bmatrix}\sigma_e^{-2}\nonumber\\
    &=\begin{bmatrix} -p_2/\sigma_e\\-1/\sigma_e\end{bmatrix}\label{eq:example_nmp_Kalman_gain}
\end{align}
and
\begin{align}
\begin{bmatrix}
   a_1&a_2\\a_3&a_4 
\end{bmatrix}\triangleq A_{\mathsf{cl}}^o= A-LC=\begin{bmatrix}
      1-\nmpparam \frac{p_2}{\sigma_e}& 1+\frac{p_2}{\sigma_e}\\-\frac{\nmpparam}{\sigma_e}&1+\frac{1}{\sigma_e}
  \end{bmatrix}.\label{eq:example_nmp_closed_loop_matrix}
\end{align}

The first step towards computing the Hessian, is calculating the derivative of the Kalman gain $\dot L$.
\begin{lemma}[Derivative of Kalman gain]\label{lem:example_nmp_gain_derivative}
Recall that $r=\frac{1+\sqrt{5}}{2}$. We have
\begin{align*}
    \dot L_1&=-2r^{-1}\nmpparam^{-3}+O(\nmpparam^{-2})\\
    \dot L_2&=-2r^{-1}\nmpparam^{-2}+O(\nmpparam^{-1}).
\end{align*}
As a result
\[
\dot L\Sigma_e\dot L^\top=4\nmpparam^{-4}\begin{bmatrix}
    \nmpparam^{-2}&\nmpparam^{-1}\\\nmpparam^{-1}&1
\end{bmatrix}+\begin{bmatrix}
    O(\nmpparam^{-5})&O(\nmpparam^{-4})\\O(\nmpparam^{-4})&O(\nmpparam^{-3})
\end{bmatrix}.
\]
Furthermore, $\dot\Sigma_e=O(\nmpparam^{-1})$.
\end{lemma}

The following lemma provides the expression for the Hessian.
\begin{lemma}[Hessian]\label{lem:example_nmp_Hessian}
    Let 
    \[
\Sigma_H=\begin{bmatrix}
    \Sigma_{H,1}&\Sigma_{H,12}\\\Sigma_{H,12}^\top&\Sigma_{H,2}
\end{bmatrix}.
    \]
    Then, we have
    \begin{align*}
        \Sigma_{H,1}&=\begin{bmatrix}
            O(\nmpparam^{-2})&O(\nmpparam^{-2})\\O(\nmpparam^{-2})&O(\nmpparam^{-2})
        \end{bmatrix},\,\Sigma_{H,12}=\begin{bmatrix}
            O(\nmpparam^{-4})&O(\nmpparam^{-3})\\O(\nmpparam^{-4})&O(\nmpparam^{-3})
        \end{bmatrix}\\
            \Sigma_{H,2}&=2\begin{bmatrix}
            \nmpparam^{-7}&\nmpparam^{-6}\\\nmpparam^{-6}&\nmpparam^{-5}
        \end{bmatrix}+\begin{bmatrix}
            O(\nmpparam^{-6})&O(\nmpparam^{-5})\\O(\nmpparam^{-5})&O(\nmpparam^{-4})
        \end{bmatrix}.
    \end{align*}
    Moreover, $H=\Omega(\nmpparam^{-7})$ and $C\Sigma_{H,2}C^\top=O(\nmpparam^{-3})$.
\end{lemma}
\subsection{Fisher Information under optimal policy}
Under the optimal policy, it follows from Corollary~\ref{Cor: info under optimal policy} that $  \Sigma_H=\Sigma_{\mathsf{FI}} $. Hence, we obtain the following rate.
\begin{lemma}[Fisher Information]\label{lem:example_nmp_FI}
    We have
    \[
\lim_{T\rightarrow \infty}\frac{1}{T}\mathsf{FI}(\theta)=O(\nmpparam^{-3}).
    \]
\end{lemma}
\section*{Proofs}
\subsection*{Proof of Lemma~\ref{lem:example_nmp_Riccati_rates}}
\begin{proof}
\textbf{  Part a) positive definite.}  First, we verify that under~\eqref{eq:example_nmp_asymptotic_Riccati} and sufficiently small $\nmpparam$, matrix
    \[
\begin{bmatrix} p_1&p\\p&p_2\end{bmatrix}
    \]
   is positive definite. Both $\tilde p_1$ and $\tilde p_2$ are guaranteed to be positive for small $\nmpparam$. We only need to verify positivity of the determinant $p_1p_2-p^2$ or equivalently $\tilde p_1\tilde p_2-\tilde p^2$. We have
   \begin{align*}
\tilde p_1\tilde p_2&=4+(2(r+1)+2(4r+1))\nmpparam+O(\nmpparam^2)\\{\tilde p}^2&=4+4(2r+1)\nmpparam+O(\nmpparam^2).
   \end{align*}
   As a result, for small enough $\nmpparam$ we have $\tilde p_1\tilde p_2-{\tilde p}^2=2r\nmpparam+O(\nmpparam^2)>0$.

\textbf{  Part b) satisfaction of DARE.} 
By Lemma~\ref{lem:example_nmp_DARE_reformulation}, we have that the DARE is equivalent to
\begin{equation*}\begin{aligned}
    2\tilde p+\tilde p_2 \nmpparam+\nmpparam^2&=\tilde p_2^2\\
    -(\tilde p_2-\nmpparam-1)\omega +1&=\omega^2\\
    (\tilde p_2-\nmpparam)\nmpparam\omega&=\tilde p-\tilde p_1\\
    \tilde p_2-\tilde p&=\nmpparam \omega.
\end{aligned}
\end{equation*}
We will exploit the implicit function theorem to define a solution at $\nmpparam=0$ and extend the solution to $\nmpparam>0$.
Define
\begin{equation*}g(\tilde p_1,\tilde p,\tilde p_2,\omega,\nmpparam)\triangleq \begin{bmatrix}
    2\tilde p+\tilde p_2 \nmpparam+\nmpparam^2-\tilde p_2^2\\
    -(\tilde p_2-\nmpparam-1)\omega +1-\omega^2\\
    (\tilde p_2-\nmpparam)\nmpparam\omega-\tilde p+\tilde p_1\\
    \tilde p_2-\tilde p-\nmpparam \omega.
\end{bmatrix}
\end{equation*}
Consider the equations at $\nmpparam=0$
\begin{equation*}\begin{aligned}
    2\tilde p_0-\tilde p_{20}^2&=0\\
    \omega_0^2+(\tilde p_{20}-1)\omega_0 -1&=0\\
    \tilde p_0-\tilde p_{10}&=0\\
    \tilde p_{20}-\tilde p_{0}&=0.
\end{aligned}
\end{equation*}
Observe that $\tilde p_{10}=\tilde p_0=\tilde p_{20}=2$, $\omega_0=-r$ is one solution. Equivalently, $g(2,2,2,-r,0)=0$. The Jacobian with respect to $\tilde p_1,\tilde p,\tilde p_2,\omega$ around the solution is equal to
\[
J\triangleq\frac{\partial g(\tilde p_1,\tilde p,\tilde p_2,\omega,\nmpparam)}{\partial(\tilde p_1,\tilde p,\tilde p_2,\omega)}|_{(2,2,2,-r,0)}=\begin{bmatrix}
0&2&-4&0\\0&0&r&-1+2r\\1&-1&0&0\\0&-1&1&0
\end{bmatrix}
\]
which is full rank. Hence, there exists a neighborhood of $(2,2,2,-r,0)$ where a solution to the DARE is well defined as a function of $\nmpparam$. Moreover, we can compute the gradient of $\tilde p_1,\tilde p,\tilde p_2,\omega$ with respect to $\nmpparam$
\begin{multline}
\frac{\partial}{\partial \nmpparam}\left.\begin{bmatrix}
p_1\\\tilde p\\\tilde p_2\\\omega
\end{bmatrix}\right\vert_{\nmpparam=0}=-J^{-1}\frac{\partial g(\tilde p_1,\tilde p,\tilde p_2,\omega,\nmpparam)}{\partial \nmpparam}|_{(2,2,2,-r,0)}\\=\begin{bmatrix}
0.5&0&-1&2\\0.5&0&0&2\\0.5&0&0&1\\\tfrac{r}{2-4r}&\tfrac{2}{2-4r}&0&\tfrac{2r}{2-4r}
\end{bmatrix}\begin{bmatrix}2\\-r\\-2r\\r
\end{bmatrix}\\=\begin{bmatrix}
    1+4r\\1+2r\\1+r\\r^2/(-2r+1)
\end{bmatrix}=\begin{bmatrix}
    1+4r\\1+2r\\1+r\\-\tfrac{3r+1}{5}
\end{bmatrix}
\end{multline}
Since $g$ is infinite times differentiable, the analytic implicit function theorem implies that all higher order derivatives exist and we can take the Taylor expansion around $(2,2,2,-r,0)$. 
Hence, there exists a neighborhood of $(2,2,2,-r,0)$ such that all solutions to~\eqref{eq:example_nmp_DARE_reformulation} are equal to
  \begin{equation*}
    \begin{aligned}
      \tilde p_1(\nmpparam)&=2+(4r+1)\nmpparam+O(\nmpparam^2)\\
      \tilde p(\nmpparam)&=2+(2r+1)\nmpparam+O(\nmpparam^2)\\
      \tilde p_2(\nmpparam)&=2+(r+1)\nmpparam+O(\nmpparam^2)\\
      \omega(\nmpparam)&=-r-\frac{r^2}{2r-1}+O(\nmpparam^2).
    \end{aligned}
    \end{equation*}

\textbf{Part c) uniqueness argument.} We showed that there exists a valid solution to~\eqref{eq:example_nmp_DARE_reformulation} of the form~\eqref{eq:example_nmp_asymptotic_Riccati} for a small enough neighborhood of $\nmpparam$ around zero. The solution is also positive definite for a small enough neighborhood of  $\nmpparam$ around zero. Since $(A,C)$ is observable and $(A,W^{1/2})$ is controllable, this solution is the unique positive definite solution to the DARE.

 The result for $\sigma_e$ follows from $\sigma^2_e=\omega^2$ (Lemma~\ref{lem:example_nmp_DARE_reformulation}) and the fact that $\sigma_e>0$, which imply $\sigma_e=-\omega$.
\end{proof}
\begin{lemma}[DARE reformulation]\label{lem:example_nmp_DARE_reformulation}
Let $P=\begin{bmatrix} p_1&p\\p&p_2\end{bmatrix}$ be the solution to the Riccati equation for the Kalman filter. Define $\tilde p_1=\nmpparam^3 p_1$,\, $\tilde p=\nmpparam^2 p,\,\tilde p_2=\nmpparam p_2$, $\omega=-\nmpparam p+p_2$. Then, the DARE is equivalent to the following system of equations
\begin{equation}\label{eq:example_nmp_DARE_reformulation}\begin{aligned}
    2\tilde p+\tilde p_2 \nmpparam+\nmpparam^2&=\tilde p_2^2\\
    -(\tilde p_2-\nmpparam-1)\omega +1&=\omega^2\\
    (\tilde p_2-\nmpparam)\nmpparam\omega&=\tilde p-\tilde p_1\\
    \tilde p_2-\tilde p&=\nmpparam \omega.
\end{aligned}
\end{equation}
Moreover, $\omega^2=\sigma^2_e$.
\end{lemma}
\begin{proof}
   Expanding the DARE, we obtain
\begin{align*}
    p_1&=p_1+2p+p_2+1-\frac{(-\nmpparam p_1+(1-\nmpparam)p+p_2)^2}{\sigma^2_e}\\
    p&=p+p_2-\frac{(-\nmpparam p_1+(1-\nmpparam)p+p_2)(-\nmpparam p+p_2)}{\sigma^2_e}\\
    p_2&=p_2+1-\frac{(-\nmpparam p+p_2)^2}{\sigma^2_e},
\end{align*}
which leads to
\begin{align*}
    2p+p_2+1&=\frac{(-\nmpparam p_1+(1-\nmpparam)p+p_2)^2}{\sigma^2_e}\\
    p_2&=\frac{(-\nmpparam p_1+(1-\nmpparam)p+p_2)(-\nmpparam p+p_2)}{\sigma^2_e}\\
    \sigma^2_e&=(-\nmpparam p+p_2)^2
\end{align*}
We can simplify the equations by introducing variables $\omega=-\nmpparam p+p_2$, $\chi=(-\nmpparam p_1+p)/\omega+1$.
Then, we obtain 
\begin{align*}
    2p+p_2+1&=\chi^2\\
    p_2&=\chi\\
    \sigma^2_e&=\omega^2\\
    (\chi-1)\omega&=p-\nmpparam p_1\\
    p_2-\nmpparam p&=\omega,
\end{align*}
where the last two equations follow from the definition of $\chi,\omega$. Expanding $\sigma^2_e=CPC'+V$ and eliminating $\chi$, leads to
\begin{align*}
    2p+p_2+1&=p_2^2\\
    \nmpparam^2 p_1-2\nmpparam p+p_2+1&=\omega^2\\
    (p_2-1)\omega&=p-\nmpparam p_1\\
    p_2-\nmpparam p&=\omega.
\end{align*}
Plugging the third and fourth equation into the second we finally obtain
\begin{align*}
    2p+p_2+1&=p_2^2\\
   - \nmpparam(p_2-1)\omega+\omega +1&=\omega^2\\
    (p_2-1)\omega&=p-\nmpparam p_1\\
    p_2-\nmpparam p&=\omega.
\end{align*}
To finish the proof, we
multiply the first equation by $\nmpparam^2$, the third by $\nmpparam^2$, and the fourth by $\nmpparam$.
\end{proof}
\subsection*{Proof of Lemma~\ref{lem:example_nmp_gain_derivative}}
\begin{proof}
By the definition of $A(\theta)$, we have \[\dot A=\begin{bmatrix}
0&0\\1&0\end{bmatrix}.\]
Recall the formula
    \begin{align*}
  \dot L &= \left(A_{\mathsf{cl}}^o \Sigma \dot C^\top + A_{\mathsf{cl}}^o \dot \Sigma  C^\top + (\dot A - L \dot C) \Sigma C^\top\right)\Sigma_e^{-1}\\
  &=\left( A_{\mathsf{cl}}^o \dot \Sigma  C^\top + \dot A \Sigma C^\top\right)\Sigma_e^{-1}.
    \end{align*}
    Define
\[
\dot \Sigma=\begin{bmatrix}
    \dot p_1&\dot p\\
    \dot p&\dot p_2
\end{bmatrix}
\]
   and consider the change of variables $q_2=\nmpparam^3(-\nmpparam \dot p_1+\dot p)$, $q_3=\nmpparam^2(-\nmpparam \dot p+\dot p_2)$. Then we have
   \[
A_{\mathsf{cl}}^o \dot \Sigma  C^\top=\begin{bmatrix}
    (a_1q_2+\nmpparam a_2q_3)\nmpparam^{-3}\\
     (\nmpparam^{-1}a_3q_2+ a_4q_3)\nmpparam^{-2}
\end{bmatrix}.
   \]
   By Lemmas~\ref{lem:example_nmp_Riccati_rates},~\ref{lem:example_nmp_Riccati_derivatives}, we obtain
   \[
A_{\mathsf{cl}}^o \dot \Sigma  C^\top\Sigma_e^{-1}=\begin{bmatrix}
    -2r^{-1}\nmpparam^{-3}+O(\nmpparam^{-2})\\ -2r^{-1}\nmpparam^{-2}+O(\nmpparam^{-1})
\end{bmatrix}.
   \]
   Meanwhile, from Lemma~\ref{lem:example_nmp_DARE_reformulation}, the second term is
   \[
\dot A\Sigma C^\top \Sigma^{-1}_e=\begin{bmatrix}
    0\\\frac{-\nmpparam p_1+p}{\sigma^2_e}
\end{bmatrix}=\begin{bmatrix}
    0\\\frac{1-p_2}{\sigma_e}
\end{bmatrix}=\begin{bmatrix}
    0\\O(\nmpparam^{-1})
\end{bmatrix}
   \]
   which can be absorbed in the higher order terms.

   The result for $\dot \Sigma_e$ follows directly from Lemma~\ref{lem:example_nmp_Riccati_derivatives}.
\end{proof}
\begin{lemma}[Derivative of Riccati matrix]\label{lem:example_nmp_Riccati_derivatives}
Recall that $r=\frac{1+\sqrt{5}}{2}$. 
We have
\begin{align*}
    \dot p_1&=4\nmpparam^{-5}+O(\nmpparam^{-4})\\
     -\nmpparam \dot p_1+\dot p&=-2r\nmpparam^{-3}+O(\nmpparam^{-2})\\
    -\nmpparam \dot p+\dot p_2&=-2r\nmpparam^{-2}+O(\nmpparam^{-1}).
\end{align*}
As a result, $\dot\Sigma_e=O(\nmpparam^{-1})$.
\end{lemma}
\begin{proof}
    Consider the change of variables $q_1=\nmpparam^5\dot p_1$, $q_2=-\nmpparam^4\dot p_1+\nmpparam^3\dot p$, $q_3=-\nmpparam^3 \dot p+\nmpparam^2 \dot p_2$. By Lemma~\ref{lem:example_nmp_Riccati_derivatives_change_variables}, Lemma~\ref{lem:example_nmp_Riccati_rates} and the identity $r^2=r+1$, we obtain
    \[
\Bigg(\begin{bmatrix}
    -2&4(3-2r) &-4(2-r)\\-1 &5-3r&-2\\0&2-r&-r
\end{bmatrix}+O(\nmpparam)\Bigg)\begin{bmatrix}
    \dot q_1\\\dot q_2\\\dot q_3
\end{bmatrix}=\begin{bmatrix}
    0\\2\\4
\end{bmatrix}+O(\nmpparam)
    \]
    By the implicit function theorem, it follows that
    \begin{align*}
\begin{bmatrix}
    \dot q_1\\\dot q_2\\\dot q_3
\end{bmatrix}&=\begin{bmatrix}
    -2&4(3-2r) &-4(2-r)\\-1 &5-3r&-2\\0&2-r&-r
\end{bmatrix}^{-1}\begin{bmatrix}
    0\\2\\4
\end{bmatrix}+O(\nmpparam)\\
&=\begin{bmatrix}
    4\\-2r\\-2r
\end{bmatrix}+O(\nmpparam).
    \end{align*}
    The result follows from substituting back $\dot p_1,\dot p,\dot p_2$. Finally,
    $\dot\Sigma_e=C\dot \Sigma C^\top=\nmpparam^2\dot p_1-2\nmpparam\dot p+\dot p_2=O(\nmpparam^{-1})$.
\end{proof}

To compute the derivative $\dot \Sigma$, we need to solve the Lyapunov equation
\begin{equation}\label{eq:example_nmp_derivative_Sigma_Lyapunov}
    \dot \Sigma=(A-LC)\dot{\Sigma}(A-LC)^\top+D,
\end{equation}
where $D=(A-LC)\Sigma \dot A^\top+\dot A\Sigma  (A-LC)^\top$.

\begin{lemma}[Derivative, change of variables]\label{lem:example_nmp_Riccati_derivatives_change_variables}
    Recall that $r=\frac{1+\sqrt{5}}{2}$. Consider the change of variables $q_1=\nmpparam^5\dot p_1$, $q_2=-\nmpparam^4\dot p_1+\nmpparam^3\dot p$, $q_3=-\nmpparam^3 \dot p+\nmpparam^2 \dot p_2$. Then~\eqref{eq:example_nmp_derivative_Sigma_Lyapunov} is equivalent to the equation
    \begin{align}
    &\begin{bmatrix}
        -2-\nmpparam&-2\tfrac{\tilde p_2}{\sigma_e}+\tfrac{\tilde p^2_2}{\sigma^2_e}-2\nmpparam-\nmpparam^2&-(\nmpparam+\tfrac{\tilde p_2}{\sigma_e})^2\\-1&-\tfrac{1}{\sigma_e}+\tfrac{\tilde{p}_2}{\sigma^2_e}-\nmpparam&-\tilde p_2(\tfrac{1+\sigma_e}{\sigma^2_e})-(1+\tfrac{1}{\sigma_e})\nmpparam\\0&\tfrac{1}{\sigma^2_e}&\tfrac{-2\sigma_e-1}{\sigma^2_e}
        \end{bmatrix}\nonumber \\&\times \begin{bmatrix}
            q_1\\q_2\\q_3
        \end{bmatrix}=\begin{bmatrix}
            0\\\tilde p_1- \tilde p \nmpparam-\nmpparam^3\\2\tilde p-2\tilde p_2\nmpparam+2\nmpparam^2 \label{eq:example_nmp_derivative_Sigma_Lyapunov_equivalent_system}
        \end{bmatrix}
    \end{align}
\end{lemma}
\begin{proof}
We have
\begin{align*}
(A-LC)\Sigma \dot A^\top&=\begin{bmatrix}
    0&p_1+p+(-\nmpparam p_1+p)\frac{p_2}{\sigma_e}\\0&p+(-\nmpparam p_1+p)\frac{1}{\sigma_e}
\end{bmatrix}\\
&=\begin{bmatrix}
    0&p_1-p-1\\0&p-p_2+1
\end{bmatrix},
\end{align*}
where the last equality follows from $2p+p_2+1=p^2_2$ and $-(p_2-1)\sigma_e=-\nmpparam p_1+p$, which in turn follow from the proof of Lemma~\ref{lem:example_nmp_DARE_reformulation} and $\sigma_e=-\omega$. 
Rewriting~\eqref{eq:example_nmp_derivative_Sigma_Lyapunov} in vector form, we obtain
\begin{multline}
  \begin{bmatrix}
        1-a_1^2&-2a_1a_2&-a^2_2\\-a_1a_3&1-a_1a_4-a_2a_3&-a_2a_4\\-a^2_3&-2a_3a_4&1-a^2_4
        \end{bmatrix} \\\times \begin{bmatrix}
            \dot p_1\\\dot p\\\dot p_2
        \end{bmatrix}=\begin{bmatrix}
            0\\ p_1- p -1\\2p-2 p_2+2
        \end{bmatrix}.
\end{multline}
The result follows from substituting
\[
\begin{bmatrix}
            \dot p_1\\\dot p\\\dot p_2
        \end{bmatrix}=\begin{bmatrix}
            1&0&0\\\nmpparam&1&0\\\nmpparam^2&\nmpparam&1
        \end{bmatrix}\begin{bmatrix}
            \nmpparam ^{-5} q_1\\\nmpparam ^{-3}  q_2\\\nmpparam ^{-2}  q_3
        \end{bmatrix}
\]
multiplying from the left with
\[
\begin{bmatrix}
    \nmpparam^4&0&0\\0&\nmpparam^3&0\\0&0&\nmpparam^2
\end{bmatrix},
\]
and~\eqref{eq:example_nmp_closed_loop_matrix}.
\end{proof}

\subsection{Proof of Lemma~\ref{lem:example_nmp_Hessian}}
\begin{proof}
Recall the Lyapunov equation    \[\Sigma_H = \dlyap(A_{\mathsf{joint}}^H, B_{\mathsf{joint}}^H \Sigma_e (B_{\mathsf{joint}}^H)^\top)\] where
    \begin{equation*}
    A_{\mathsf{joint}}^H = \bmat{A_{\mathsf{cl}}^c & 0 \\ \dot A - L \dot C + \dot B F & A_{\mathsf{cl}}^o}, \quad
    B_{\mathsf{joint}}^H = \bmat{L \\ \dot L}. 
    \end{equation*}
    The Lyapunov equation is equivalent to
    the following system of equations
    \begin{align*}
        \Sigma_{H,1}&=A_{\mathsf{cl}}^c\Sigma_{H,1}A_{\mathsf{cl}}^{c,\top}+L\Sigma_e L^\top\\
        \Sigma_{H,12}&=A_{\mathsf{cl}}^c\Sigma_{H,1}\dot A+A_{\mathsf{cl}}^c\Sigma_{H,12}A_{\mathsf{cl}}^{o,\top}+L\Sigma_e\dot L^\top\\
        \Sigma_{H,2}&=\dot A \Sigma_{H,1}\dot A^{\top}+\dot A\Sigma_{H,12}A_{\mathsf{cl}}^{o,\top}+A_{\mathsf{cl}}^{o}\Sigma^{\top}_{H,12}\dot A^\top\\&+A_{\mathsf{cl}}^{o}\Sigma^{\top}_{H,2}A_{\mathsf{cl}}^{o,\top}+\dot L\Sigma_e \dot L^\top
    \end{align*}
\textbf{Bounding $\Sigma_{H,1}$.}  The result follows directly from $\|L\Sigma_e L^\top\|=O(\nmpparam^{-2})$ and the fact that $A_{\mathsf{cl}}^c$ is independent of $\nmpparam$ and is always stable. 

\noindent\textbf{Bounding $\Sigma_{H,12}$.} Define the similarity transformation $\tilde{A}^o_{\mathsf{cl}}=TA_{\mathsf{cl}}^o T^{-1}$ where $T=\begin{bmatrix}
    \nmpparam^4&0\\0&\nmpparam^3
\end{bmatrix}$. Define $\tilde{\Sigma}_{H,12}=\Sigma_{H,12} T$, $\tilde{W}= (\dot A\Sigma_{H,1}\dot A^\top+L\Sigma_e \dot L^\top)T$. Multiplying the second equation from the right by $T$, we obtain
\[\tilde{\Sigma}_{H,12}=A_{\mathsf{cl}}^c\tilde \Sigma_{H,12}\tilde A_{\mathsf{cl}}^{o,\top}+\tilde{W}.
\]
Observe that both $\|\tilde W\|=O(1), \|\tilde A_{\mathsf{cl}}^{o,\top}\|=O(1)$, with
\[
\tilde A_{\mathsf{cl}}^{o}=\underbrace{\begin{bmatrix}
    1-\frac{2}{r}&\frac{2}{r}\\-\frac{1}{r}&1+\frac{1}{r}
\end{bmatrix}}_{\tilde A^o_{\mathsf{cl},\infty}}+O(\nmpparam).
\]
Define the linear matrix operator
\[
\mathcal{L}(X)=X-A_{\mathsf{cl}}^c X \tilde A^{o,\top}_{\mathsf{cl},\infty}.
\]
By a vectorization argument, the eigenvalues of $\mathcal{L}$ are equal to the eigenvalues of $I-\tilde A^{o}_{\mathsf{cl},\infty}\otimes A_{\mathsf{cl}}^c $. The eigenvalues of $\tilde A^{o}_{\mathsf{cl},\infty}\otimes A_{\mathsf{cl}}^c$ are the pairwise products of the eigenvalues of $\tilde A^{o}_{\mathsf{cl},\infty}, A_{\mathsf{cl}}^c$. The eigenvalues of $A_{\mathsf{cl}}^c$ lie strictly inside the unit circle while $\tilde A^{o}_{\mathsf{cl},\infty}$ has spectral radius equal to one. Hence, all eigenvalues of $\mathcal{L}$ are non-zero and $\mathcal{L}$ is invertible.
By the implicit function theorem, we obtain
\[
\|\tilde \Sigma_{H,12}\|=\mathcal{L}^{-1}(\tilde W)+O(\nmpparam).
\]
Note that $\mathcal{L}$ is linear, independent of $\nmpparam$. Hence, so is $\mathcal{L}^{-1}$. 
Thus, $\|\tilde \Sigma_{H,12}\|=O(1)$. The result follows from $\Sigma_{H,12}=\tilde \Sigma_{H,12} T^{-1}$.

\noindent\textbf{Bounding $\Sigma_{H,2}$.} Define
\[
W_{H,2}=\dot A \Sigma_{H,1}\dot A^{\top}+\dot A\Sigma_{H,12}A_{\mathsf{cl}}^{o,\top}+A_{\mathsf{cl}}^{o}\Sigma^{\top}_{H,12}\dot A^\top+\dot L\Sigma_e \dot L^\top.
\]
By direct calculation
\[
\dot A\Sigma_{H,12}A_{\mathsf{cl}}^{o,\top}+A_{\mathsf{cl}}^{o}\Sigma^{\top}_{H,12}\dot A^\top=\begin{bmatrix}
    0&O(\nmpparam^{-4})\\O(\nmpparam^{-4})&O(\nmpparam^{-5})
\end{bmatrix}
\]
and $\|\dot A \Sigma_{H,1}\dot A^{\top}\|=O(\nmpparam^{-2})$. Hence, by Lemma~\ref{lem:example_nmp_gain_derivative}
\[
W_{H,2}=4\nmpparam^{-4}\begin{bmatrix}
    \nmpparam^{-2}&\nmpparam^{-1}\\\nmpparam^{-1}&1
\end{bmatrix}+\begin{bmatrix}
    O(\nmpparam^{-5})&O(\nmpparam^{-4})\\O(\nmpparam^{-4})&O(\nmpparam^{-3})
\end{bmatrix}.
\]

To compute $\Sigma_{H,2}$, we need to solve the Lyapunov equation
\[
\Sigma_{H,2}=A_{\mathsf{cl}}^o\Sigma_{H,2}A_{\mathsf{cl}}^{o,\top}+W_{H,2},
\]
which is similar to~\eqref{eq:example_nmp_derivative_Sigma_Lyapunov}. We will follow the exact same steps to solve it.
Using the change of variables 
\begin{align*}
    d_1&=\nmpparam^7[\Sigma_{H,2}]_{11}\\
    d_2&=-\nmpparam^6[\Sigma_{H,2}]_{11}+\nmpparam^5[\Sigma_{H,2}]_{12}\\
    d_3&=-\nmpparam^5[\Sigma_{H,2}]_{12}+\nmpparam^4[\Sigma_{H,2}]_{22}
\end{align*}
and following the same steps as in the proof of Lemma~\ref{lem:example_nmp_Riccati_derivatives}, we arrive at
 \[
\Bigg(\begin{bmatrix}
    -2&4(3-2r) &-4(2-r)\\-1 &5-3r&-2\\0&2-r&-r
\end{bmatrix}+O(\nmpparam)\Bigg)\begin{bmatrix}
     d_1\\ d_2\\ d_3
\end{bmatrix}=\begin{bmatrix}
    4\\4\\4
\end{bmatrix}+O(\nmpparam)
    \]
    By the implicit function theorem
    \[
\begin{bmatrix}
     d_1\\ d_2\\ d_3
\end{bmatrix}=\begin{bmatrix}
     2\\ -2r\\ -2r
\end{bmatrix}+O(\nmpparam).
    \]
    As a result
    \[
\Sigma_{H}=2\begin{bmatrix}
            \nmpparam^{-7}&\nmpparam^{-6}\\\nmpparam^{-6}&\nmpparam^{-5}
        \end{bmatrix}+\begin{bmatrix}
            O(\nmpparam^{-6})&O(\nmpparam^{-5})\\O(\nmpparam^{-5})&O(\nmpparam^{-4})
        \end{bmatrix}.
    \]
   Since $F$ is full rank, we have that $H=\Theta(F\Sigma_{H,2}F')=\Omega(\nmpparam^{-7})$.
   Finally, $C\Sigma_{H,2}C^\top=\nmpparam^{-4}(d_3-d_2)=\nmpparam^{-4}O(\nmpparam)=O(\nmpparam^{-3})$.
   \subsection*{Proof of Lemma~\ref{lem:example_nmp_FI}}
From Lemmas~\ref{lem:example_nmp_gain_derivative} and~\ref{lem:example_nmp_Hessian}, we obtain   \[
\dot\Sigma_e=O(\nmpparam^{-1}),\quad C\Sigma_{H,2}C^\top=O(\nmpparam^{-3}).
\]
The result follows from Proposition~\ref{prop: info}.
\end{proof}

\section{No Margin Example} 

We first study the dominant terms of the Riccati equation solutions for this example. We then proceed to prove our derivation of the Hessian.

\subsection{Riccati Solutions}

\begin{lemma}
    \label{lem: Riccati solution}
    For $\sigma$ sufficently small, the there exist unique positive definite solutions to the Riccati equations for the optimal LQR controller and Kalman Filter corresponding to the system \eqref{eq: doyle example}. These solutions satisfy 
    \begin{align*}
        P &= \bmat{1 & 1 \\ 1 & 1} + \bmat{4 & 2 \\ 2 & 1} \sqrt{\sigma} + o(\sqrt{\sigma})\\
        \Sigma &= \bmat{1 & 1 \\ 1 & 1} + \bmat{1 & 2 \\ 2 & 4} \sqrt{\sigma} + o(\sqrt{\sigma}).
    \end{align*}
\end{lemma}
Given these Riccati solutions, we may compute 
\begin{align*}
    A_{\mathsf{cl}}^c &= \bmat{2 & 1   \\ -2 - 2\sqrt{\sigma} & -1-\sqrt{\sigma} } + \Delta A_{\mathsf{cl}}^c \\
    A_{\mathsf{cl}}^o &= \bmat{-1-\sqrt{\sigma} & 1 \\ -2-2\sqrt{\sigma} & 2 } + \Delta A_{\mathsf{cl}}^o,
\end{align*}
where $\norm{\Delta A_{\mathsf{cl}}^c} = o(\sqrt{\sigma})$ and $\norm{\Delta A_{\mathsf{cl}}^o} = o(\sqrt{\sigma})$. This leads to the characterization of $A_{
\mathsf{joint}}^H$ 
\ifTACmode
above \eqref{eq: hatAjointH}. 
\else
in \eqref{eq: Ajoint decomp}. 
\fi
Each of the closed loop matrices has eigenvalues at $O(\sigma
)$ and $1-\sqrt{\sigma} + O(\sigma)$. The near marginal stability of the closed-loop state feedback dynamics and observer error dynamics render the performance of the system under the learned controller sensitive to errors in the dynamics identification. Furthermore, the choice of the unknown dynamics amplifies the individual sensitivities of the observer and the controller.

\begin{proof}
    We prove the result for the controller. The proof for the observer is identical. The existence of a unique positive definite solution is ensured by the fact that the pair$(A,B)$ is stabilizable, the pair $(A,Q^{1/2})$ is detectable, and $R = \sigma > 0$ \citep{zhou1996robust}. To verify that this solution is characterized by the matrix $P$ given in the statement, first observe that $P$ is positive definite by the fact that its determinant is given by $\sqrt{\sigma} + o(\sqrt{\sigma}) > 0$ for $\sigma$ sufficiently small,  and the top left entry is positive. Next, we verify that $P$ satisfies the Riccati equation. Let
    \begin{align*}
        P = \bmat{p_{11} & p_{12} \\ p_{12} & p_{22}}. 
    \end{align*}
    Any solution to the Riccati equation then satisfies
    \begin{align*}
        0 &= (3 p_{11} +1) (p_{22} +x^2) - 4p_{12}^2 \\
        0 &= (2 p_{11} + 3p_{12} + 1)(p_{22} + x^2)- 2p_{12}(p_{12} + 2 p_{22})  \\
        0 &=  (p_{11} + 4 p_{12} + 3 p_{22} + 1)(p_{22} + x^2) - (p_{12} + 2 p_{22})^2 ,
    \end{align*}
    where $x = \sqrt{\sigma}$.
    Denote the right hand side of the above system of equations as $g_i(p_{11}, p_{12}, p_{22}, x)$ for $i=1,2,3$.
    This system of equations may be written compactly as 
    \begin{align*}
        g(p_{11}, p_{12}, p_{22}, x) \triangleq \bmat{g_1(p_{11}, p_{12}, p_{22}, x) \\ g_2(p_{11}, p_{12}, p_{22}, x) \\g_3(p_{11}, p_{12}, p_{22}, x) } = 0. 
    \end{align*}
    It holds that $g(1,1,1,0) = 0$. Let $p = \bmat{p_{11} & p_{12} & p_{22}}^\top$. Then 
    \begin{align*}
        &D_p g(p, x)=\\& \bmat{3(p_{22} + x^2) & -8 p_{12} & 3p_{11} + 1  \\ 
        2(p_{22} + x^2) & 3 (p_{22} + x^2) - 4p_{12} & 2p_{11} -p_{12} + 1\\
        p_{22} + x^2 & 4x^2 - 2p_{12} - 4p_{22} & 3x^2 + p_{11}+ 1 - 2p_{22} }\\
        &\frac{d}{dx} g(p, x)= \bmat{2x(3p_{11} + 1) \\ 2x (2p_{11} + 3p_{12} + 1)\\ 2x(p_{11} + 4p_{12} + 3p_{22} +1)}
    \end{align*}
    Evaluating $p_{11}, p_{12}, p_{22}, x = 1,1,1,0$ results in
    \begin{align*}
        D_p g(1,1,1,0) = \bmat{3 & -8 & 4 \\ 2 & -5 & 2 \\ 1 & -2 & 0},
    \end{align*}
    and $\frac{d}{dx} g(1,1,1,0)=0$. As $D_p g(1,1,1,0)$ is not invertible, the analytic implicit function theorem does not suffice to ensure a continuously differentiable solution $p(x)$ to the system of equations. However, by treating $p_{22}$ as a parameter, it does suffice to ensure the existence of analytic functions $p_{11}(p_{22}, x)$ and $p_{12}(p_{22}, x)$ over a neighborhood of $p_{22}=1$ and $x=0$ which satisfy $g_1(p_{11}(p_{22}, x), p_{12}(p_{22},x), p_{22}, x) = g_2(p_{11}(p_{22}, x), p_{12}(p_{22},x), p_{22}, x)=0$, and with 
    \begin{align*}
        \bmat{(\nabla p_{11}(p_{22}, x))^\top \\ (\nabla p_{12}(p_{22}, x))^\top } = -\bmat{3 & -8 \\ 2 & -5}^{-1} \bmat{4 & 0 \\ 2 & 0} = \bmat{4 & 0 \\ 2 & 0}. 
    \end{align*}
    We can similarly compute the second order deviations of the solutions $p_{12}(p_{22}, x)$ and $p_{12}(p_{12},x)$ as 
    \begin{align*}
        \nabla^2 p_{11}(0,0) &= \bmat{-8 & 0 \\ 0 & -56} \\
        \nabla^2 p_{12}(0,0) &= \bmat{-4 & 0 \\ 0 & -20}.
    \end{align*}
    Denote $p_{22} - 1 = u$. It holds by a Taylor series expansion that
    \begin{align*}
        p_{11}(1+u, x) =& 1 + 4 u - 4 u^2 -28 x^2 + \sum_{i=0}^2 o\paren{ |x^{i} u^{2-i}|} \\
        p_{12}(1+u,, x) =& 1 + 2 u - 2 u^2 -10 x^2 + \sum_{i=0}^2 o\paren{ |x^{i} u^{2-i}|}.
    \end{align*}
    
    Substituting these solutions to $g_3$  reduces our system of equations to 
    \begin{align*}
        0 = h(u, x) &\triangleq g_3(p_{11}(u+1, x), p_{12}(u+1, x), p_{22}, x).
    \end{align*}
    By the above Taylor expansions for $p_{11}(1+u, x)$ and $p_{12}(1+u,x)$, this constraint is equivalent to
    \begin{align*}
        h(u,x) &= -\left(3 + 4 u - 2u^2 - 10x^2 + \sum_{i=0}^2 o(|x^i u^{2-i}|)\right)^2 \\
        +\big(9 +& 15 u - 12 u^2 - 68x^2 +  \sum_{i=0}^2 o(|x^i u^{2-i}|)\big)(1 + u + x^2) \\
        &= u^2 - x^2 + \sum_{i=0}^2 o(|x^i u^{2-i}|). 
    \end{align*}
    Letting $u = wx$, it holds that
    \begin{align*}
         h(wx, x) = (w^2 - 1) x^2 + \sum_{i=0}^2 o(|w^i x^2|).
    \end{align*}
   Consider the system defined by $0 = \tilde h(w,x) \triangleq \frac{1}{x^2} h(wx, x)$. The function $\tilde h$ is analytic in a neigborhood of $(1,0)$, and satisfies
    \begin{align*}
        \tilde h(w,x) = w^2 -1 + \sum_{i=0}^2 \frac{ o(|w^i x^2|)}{x^2}. 
    \end{align*}
    For $x = 0$, $w = \pm1$ are solutions to the above equation. Consider the positive solution. By the Implicit Function Theorem, there is a continuously differentiable solution $w(x)$ for a neighborhood around $x=0$, with $w(x) = 1 + o(x)$. Then $p_{22}(x) = 1 + w(x) x$ is a continuously differentiable solution to $h(p_{22}, x) = 0$, while $p_{11}(x) = p_{11}(p_{22}(x), x)$, $p_{12}(p_{22}(x), x)$ and $p_{22}(x)$ are continuously differentiable solutions to the system of equations of interest for an interval of $x > 0$. They satisfy $p_{11}(x) = 1 + 4x + o(x^2)$, $p_{12}(x)= 1 + 2x + o(x)$, $p_{22}(x) = 1 + x + o(x)$.
    \end{proof}

\subsection{Proof of \Cref{prop: doyle hessian}}

To prove this result, we first express the dominant term in the Lyapunov solution $\Sigma_H$ defining the problem.
\begin{lemma}
    \label{lem: SigmaH evaluation}
    The quantity $\Sigma_H$ required for the Hessian evaluation in \Cref{prop: hessian} is given by 
    \begin{align*}
        \Sigma_H = 16 \sigma^{-3/2} v_1 v_1^\top + O(\sigma^{-1}).
    \end{align*}
\end{lemma}
The calculation of the above expression relies on the Jordan form of \eqref{eq: hatAjointH}. We present the proof of this result shortly.  
To evaluate the Hessian of \Cref{prop: hessian}, it remains to determine the derivative of the controller, $\dot F$. 

\begin{lemma}
    It holds that 
    \begin{align*}
        \dot F = \bmat{2 & 3} + O(\sqrt{\sigma}).
    \end{align*}
\end{lemma}

\begin{proof}
We may work out $\dot F$ analytically using Lemma B.1 of \citep{simchowitz2020naive} as
\begin{align*}
    \dot F &= - (B^\top P B + R)^{-1}(\dot B^\top P (A+BF) + B^\top P (\dot A \\
    &+ \dot B F) + B^\top P' (A+BF)),
\end{align*}
where $\dot P = \dlyap(A+BF, (A+BF)^\top P ( \dot A + \dot B F) + (\dot A + \dot B F)^\top P (A+BF)$. Substituting $\dot A = 0$ and $\dot B = B$ this simplifies to 
\begin{align*}
    \dot F = -F -2\Psi^{-1} B^\top P (A+BF) - \Psi^{-1} B^\top P' (A+BF). 
\end{align*}
It holds that 
\begin{align*}
    \Psi^{-1} B^\top P (A+BF) = \Psi^{-1} R B^\top P A = O(\sigma). 
\end{align*}
For the remaining term in the expression of $\dot F$, observe that the second argument defining the Lyapunov function simplifies as 
\begin{align*}
    \mathsf{sym}((A+BF)^\top P (\Delta_A + \Delta_B F)) &= \mathsf{sym}((A+BF)^\top P B F)) \\&= \mathsf{sym}(RB^\top P A) \\
    &= O(\sigma). 
\end{align*}
Then expanding the Lyapunov solution as an infinite series and using the eigendecomposition of $A+BF$, we find that $P' = O(\sqrt{\sigma})$, and therefore $-\Psi^{-1} B^\top P' (A+BF)=O(\sqrt{\sigma})$. Thus $\dot F = \bmat{2 & 3} + O(\sqrt{\sigma})$.
\end{proof}

Consider now the Hessian from \Cref{prop: hessian}:
\begin{align*}
    H(\theta^\star) = 2 \Psi \bmat{\dot F & F} \Sigma_H \bmat{\dot F & F}^\top.
\end{align*}
By the fact that $\Sigma_H = 16 \sigma^{-3/2} v_1 v_1^\top+ O(\sigma^{-1})$, we conclude that $H(\theta^\star)$ scales as $\sigma^{-3/2}$, and it is contributed by the term $16 \sigma^{-3/2} \Psi \bmat{\dot F & F} v_1 v_1^\top \bmat{\dot F & F}$. Evaluating this quantity leads to the expression of \Cref{prop: doyle hessian}.

We now proceed to prove \Cref{lem: SigmaH evaluation}.

\begin{proof} Recall the  matrix $\hat A_{\mathsf{joint}}^H$ defined in \eqref{eq: hatAjointH}. The Jordan decomposition of this matrix is given by $\hat A_{\mathsf{joint}}^H = V J V^{-1}$ with
\begin{align}
\label{eq: jordan full}
\begin{aligned}
    J &= \bmat{1 - \sqrt{\sigma} & 1 & 0 & 0 \\ 0 & 1- \sqrt{\sigma} &0 & 0\\ 0&0 &0 &1\\ 0 &0 &0 &0}\\
    V&= \bmat{0 & -\frac{1 - \sqrt{\sigma}}{(1+\sqrt{\sigma})^2} & 0 & \frac{1 - \sqrt{\sigma}}{4} \\ 0 & \frac{1-\sqrt{\sigma}}{1 + \sqrt{\sigma}} & 0 & \frac{-1 + \sqrt{\sigma}}{2}  \\ 1 & 1 & 1 & 0 \\ 2 & 1 & 1 + \sqrt{\sigma} & -1},
\end{aligned}
\end{align}
where the columns of $V$ are denoted by $v_i$ for $i =1,\dots, 4$. %

Let $\hat \Sigma_H =\dlyap(\hat A_{\mathsf{joint}}^H, B_{\mathsf{joint}}^H \Sigma_e (B_{\mathsf{joint}}^H)^\top).$ Define the operator $\calL: \R^{4 \times 4} \to \R^{4 \times 4}$ for an arbitrary $Y \in \R^{4 \times 4}$ as
\begin{align}
    \label{eq: lyapunov difference operator}
    \calL(Y) = 
    \dlyap(\hat A_{\mathsf{joint}}^H, \mathsf{sym}(X Y (\hat A_{\mathsf{joint}}^H)^\top) + X Y X^\top ).
\end{align}

Then it holds that
\begin{align*}
    \Sigma_H &= \dlyap(A_{\mathsf{joint}}^H, B_{\mathsf{joint}}^H \Sigma_e (B_{\mathsf{joint}}^H)^\top) = \hat \Sigma_H + \calL(\Sigma_H). 
\end{align*}

Applying \Cref{lem: hatA dlyap} with $m=n=\frac{1}{\sqrt{2}}\bmat{L \\ 0} \Sigma_e^{1/2}$, we find
\begin{align*}
    \hat \Sigma_H = \alpha_{1} v_1 v_1^\top + \alpha_{2} (v_1 v_2^\top + v_2 v_1^\top) + R, 
\end{align*}
for $\alpha_{1}, \alpha_{2}$ and $R$ as in \Cref{lem: hatA dlyap}. Then it holds that $\hat \Sigma_H$ is dominated by a quantity that grows as $\sigma^{-3/2}$ by the fact that the coefficient $\alpha_1$ is nonzero. To show that $\Sigma_H$ has the same rate of growth, we consider the difference $E$ between $\Sigma_H$ and $\hat \Sigma_H$: 
\begin{align*}
    E \triangleq \Sigma_H - \hat \Sigma_H = \calL(\Sigma_H).
\end{align*} 
Observe that $\calL$ is a linear operator. Then
\begin{align*}
    \calL(\Sigma_H) = \calL(\hat \Sigma_H) + \calL(\Sigma_H - \hat \Sigma_H) = \calL(\hat \Sigma_H) + \calL^2(\Sigma_H). 
\end{align*}
Iterating this recursion, it holds that
\begin{align*}
    E = \sum_{i=1}^\infty \calL^i(\hat \Sigma_H). 
\end{align*}

By \Cref{lem: hot lyap operator}, $\calL^i(\hat \Sigma_H) = G_i + \beta_{2,i} (v_1 v_2^\top + v_2 v_1^\top) + R_i$ for symmetric matrices $R_i$ and $G_i$, with $G_i$ nonzero only in the bottom right $2\times 2$ block, and which satisfy $\norm{G_i} \leq \beta_{1,i}$, $\norm{R_i} \leq \beta_{3,i}$. Here, the sequences $\curly{\beta_{1, i}}_{i=0}^\infty, \curly{\beta_{2,i}}_{i=0}^\infty, \curly{\beta_{3,i}}_{i=0}^\infty$ satisfy the elementwise inequality
\begin{align*}
    \bmat{\beta_{1,i+1} \\ \abs{\beta_{2,i+1}} \\ \beta_{3,i+1}} \leq c\bmat{\sigma^{1/2} & 1 & \sigma^{-1/2} \\ 0 & \sigma^{1/2} & 1 \\ 0 & \sigma & \sigma^{1/2}}\bmat{\beta_{1,i} \\ \abs{\beta_{2,i}} \\ \beta_{3,i}}
\end{align*}
and start from $\beta_{1,0} = \abs{\alpha_1}$, $\beta_{2,0} = \alpha_2$, and $\beta_{3,0} = \norm{R}$. The quantity $c$ is some universal positive constant. 
The eigenvalues of the matrix defining the above inequality are $0, \sigma^{1/2}$ and $2 \sigma^{1/2}$. Therefore, the norm of the series $\sum_{i=1}^{\infty} \calL^i(\hat \Sigma_H)$ can be bounded by 
$\frac{C}{1-c \sqrt{\sigma}} (\abs{\beta_{1,1}} + \abs{\beta_{2,1}} + \abs{\beta_{3,1}})$ for some other universal constant $C$. For $\sigma$ sufficiently small, $\frac{1}{1-c \sqrt{\sigma}} \leq 2$. Additionally, the terms $\beta_{1,1}, \beta_{2,1}$ and $\beta_{3,1}$ arise from the first iteration of \Cref{lem: hot lyap operator} applied to the values $\alpha_1, \alpha_2, \alpha_3$ defining $\hat \Sigma_H$, and are therefore all norm bounded by a universal constant multiplied by $\sigma^{-1}$. Consequently, $E = O(\sigma^{-1})$.  

As a result, it holds that $\Sigma_H =  \alpha_1 v_1 v_1^\top + O(\sigma^{-1}) =\frac{1}{2} (c^m_2)^2 \sigma^{-3/2} v_1 v_1^\top + O(\sigma^{-1})$ for $c^m = V^{-1} \frac{1}{\sqrt{2}}\bmat{L \\ 0} \Sigma_e^{1/2}$. 
    To conclude the proof, we identify $v \gets v_1$, and evaluate $c_2^m = (\frac{1}{\sqrt{2}}V^{-1} \bmat{L \\ 0 } \Sigma_e^{1/2})_2 = \frac{8}{\sqrt{2}} + O(\sqrt{\sigma})$.
\end{proof}

\begin{lemma}
    \label{lem: hot lyap operator}
    Consider a matrix $\Sigma=G + \beta_{2} (v_1 v_2^\top +  v_2 v_1^\top) + R$, for $G$ a symmetric matrix nonzero only in the bottom right $2 \times 2$ block and satisfying $\norm{G}\leq \beta_1$, and $R$ a symmetric matrix satisfying $\norm{R} \leq \beta_3$. Recall the operator $\calL$ from \eqref{eq: lyapunov difference operator}. It holds that
    \begin{align*}
        &\calL(\Sigma) = \tilde G + \tilde \beta_{2} (v_1 v_2^\top + v_2 v_1^\top) + \tilde R,
    \end{align*}
    where $\tilde G$ and $\tilde R$ are symmetric matrices with $\tilde G$ nonzero only in the bottom right $2\times 2$ block with $\tilde G$, $\tilde R$, and $\tilde \beta_2$  satisfying
    \begin{align*}
    \norm{\tilde G} &\leq c (\beta_1 \sigma^{1/2} + \abs{\beta_2} + \beta_3 \sigma^{-1/2})\\
    \abs{\tilde\beta_2} &\leq c (\abs{\beta_2} \sigma^{1/2} + \beta_3)\\
    \norm{\tilde R} &\leq c(\abs{\beta_2} \sigma + \beta_3 \sigma^{1/2})
\end{align*}
for some universal positive constant $c$.
    
\end{lemma}
\begin{proof}
    Substituting the expression for $\Sigma$ into the first term above, leads to the sum of quantities of the form 
\begin{align*}
    \dlyap(\hat A_{\mathsf{joint}}^H, \mathsf{sym}(X M (\hat A_{\mathsf{joint}}^H)^\top) + X M X^\top ),
\end{align*}
with $M$ taking the following possible values $G_1$, $\beta_2 (v_1 v_2^\top +v_2 v_1^\top)$, $R$. Consider each such quantity separately:
\begin{itemize}
    \item 
    $M = G$: Only the bottom right $2 \times 2$ block of G is nonzero. Multiplying by either $X$ or $\hat A_{\mathsf{joint}}^H$ preserves this structure.  Thus  
    \begin{align*}
        &\dlyap(\hat A_{\mathsf{joint}}^H, \mathsf{sym}(X G \hat A_{\mathsf{joint}}^H)^\top) + X G X^\top)\preceq \\& c \beta_1 \sigma \dlyap\paren{\hat A_{\mathsf{joint}}^H, \bmat{0 & 0 \\ 0 & I}},
    \end{align*}
    for a universal constant $c$
    It holds by an eigendecomposition and by evaluation of a geometric series that
    \begin{align*}
        &\dlyap\paren{\hat A_{\mathsf{joint}}^H[2:,2:], I}\\
        &= \sum_{k=0}^\infty \bmat{-1 -\sqrt{\sigma} & 1 \\ -2 - 2 \sqrt{\sigma} & 2 }^k \left(\bmat{-1 -\sqrt{\sigma} & 1 \\ -2 - 2 \sqrt{\sigma} & 2 }^k \right)^\top \\
        &   \preceq c \sigma^{-1/2} I, 
    \end{align*}
    for a universal constant $c$.
    Then the resulting matrix has only nonzero elements in the lower right $2\times 2$ block, and is norm bounded by $c \beta_1 \sigma^{1/2}$ for a universal constant $c$.
    \item 
    $M = \beta_2 (v_1 v_2^\top +v_2 v_1^\top)$: Similar to above, we leverage the fact that $X v_1$ and $\hat A_{\mathsf{joint}}^H v_1$ map to a vector with the first two entries zero. We find that $V^{-1}$ multiplied by a vector that is zero in the first two components results in a vector whose second element is zero. Therefore, application of \Cref{lem: hatA dlyap} with $m$ and $n$ arising from multiplication of $X$ or $\hat A_{\mathsf{joint}}^H$ with $v_1$ and $v_2$ results in  either $c_2^m=0$ or $c_2^n=0$. Therefore, it holds that
    \begin{align*}
        &\dlyap(\hat A_{\mathsf{joint}}^H, \mathsf{sym}(X M (\hat A_{\mathsf{joint}}^H)^\top) + X M X^\top ) \\
        &= \alpha_1' v_1 v_1^\top  + \alpha_2' (v_1 v_2^\top + v_2 v_1^\top) + R',
    \end{align*}
    where $\abs{\alpha_1'} \leq c \abs{\beta_2}$, $\abs{\alpha_2'} \leq c \abs{\beta_2} \sigma^{1/2}$ and $\norm{R'} \leq c \abs{\beta_2} \sigma$ for a universal constant $c$.
    \item 
    $M=R$: It holds that 
    \begin{align*}
        &\dlyap(\hat A_{\mathsf{joint}}^H, \mathsf{sym}(X M (\hat A_{\mathsf{joint}}^H)^\top) + X M X^\top ) \\
        &\preceq \norm{R} (2\norm{X}\norm{\hat A_{\mathsf{joint}}^H} + \norm{X}^2) \dlyap(\hat A_{\mathsf{joint}}^H, I). 
    \end{align*}
    By \Cref{lem: hatA dlyap}, the above may be expressed by summing the solutions for $n=m=1/\sqrt{2} e_i$ as
    \begin{align*}
        \alpha_1'' v_1 v_1^\top  +\alpha_2'' (v_1 v_2^\top + v_2 v_1^\top) + R'',
    \end{align*}
    for $R''$ a symmetric matrix satisfying $\norm{R''} \leq \alpha_3''$ with $\abs{\alpha_1''} \leq c \beta_3 \sigma^{-1/2}$, $\abs{\alpha_2''} \leq c \beta_3$, and $\alpha_3'' \leq c \beta_3 \sigma^{1/2}$ for a universal constant $c$.
\end{itemize}
Summing these three terms concludes the proof.
\end{proof}

\begin{lemma}
    \label{lem: hatA dlyap}
    For $\hat A_{\mathsf{joint}}^H$ defined in \eqref{eq: hatAjointH} with Jordan form defined by $J$ and $V$ in \eqref{eq: jordan}, and for any $m, n \in \R^4$, there exist real numbers $\alpha_1$ and $\alpha_2$ and a symmetric matrix $R$ such that
    \begin{align*}
        &\dlyap(\hat A_{\mathsf{joint}}^H, mn^\top + n m^\top) \\
        &= \alpha_{1} v_1 v_1^\top + \alpha_{2} (v_1 v_2^\top + v_2 v_1^\top) + R,
    \end{align*}
    where
    \begin{align*}
        \alpha_{1} &= \frac{1}{2} c_2^m c_2^n \sigma^{-3/2} + \bar \alpha_1 \sigma^{-1} \\
        \bar \alpha_1 &\leq 16 \norm{V^{-1}}^2 \norm{m} \norm{n}\\
        \abs{\alpha_{2}} &\leq  \abs{c_2^m c_2^n} \sigma^{-1}  + \bar \alpha_2 \sigma^{-1/2} \\
        \bar \alpha_2 &\leq 16 \norm{V^{-1}}^2 \norm{m} \norm{n}\\
        \norm{R} &\leq 16 \norm{V}^2 \norm{V^{-1}}^2 \norm{n}\norm{m} + c_2^m c_2^n \norm{V}^2 \sigma^{-1/2}
    \end{align*}
    and for any matrix $b \in \R^4$,
    \begin{align*}
        c^b \triangleq V^{-1} b. 
    \end{align*}
\end{lemma}
\begin{proof}
For $k \geq 0$,
\begin{align*}
    &(\hat A_{\mathsf{joint}}^H)^k = V J^k V^{-1} \\
    &= V \bmat{(1-\sqrt{\sigma})^k & k (1-\sqrt{\sigma})^{k-1}  & 0 & 0 \\ 0 & (1-\sqrt{\sigma})^k & 0 & 0 \\ 0 & 0 & \mathbf{1}_{k=0} & \mathbf{1}_{k=1} \\ 0 & 0 & 0& \mathbf{1}_{k=0} } V^{-1}.
\end{align*}
Consequently, it holds that
\begin{align*}
    \label{eq: lyap expansion}
    \dlyap(\hat A_{\mathsf{joint}}^H, mn^\top) &= \sum_{k=0}^\infty (\hat A_{\mathsf{joint}}^H)^k mn^\top ((\hat A_{\mathsf{joint}}^H)^k)^\top \\
    &= \sum_{k=0}^\infty  V J^k V^{-1} m n^\top V^{-\top} (J^k)^\top V^\top. 
\end{align*}
Then 
\begin{align*}
    &\dlyap(\hat A_{\mathsf{joint}}^H, m n^\top) \\
    &= \sum_{k=0}^\infty (f_{1}^m(k) v_1  + f_2^m(k) v_2 + f_3^m(k) v_3 + f_4^m(k)v_4)\\
    &\times (f_{1}^n(k) v_1  + f_2^n(k) v_2 + f_3^n(k) v_3 + f_4^n(k)v_4)^\top 
\end{align*}
where for a vector $b\in \R^{4}$,
\begin{align*}
    f_{1}^b(k) &= c_1^b(1-\sqrt{\sigma})^k+  c_2^b k (1-\sqrt{\sigma})^{k-1} \\
    f_{2}^b(k) &= c_2^b (1-\sqrt{\sigma})^k,\\
    f_3^b(k)&= c_3^b\mathbf{1}_{k=0} + c_4^b \mathbf{1}_{k=1}    \\
    f_4^b(k)&= \mathbf{1}_{k=0} c_4^b
\end{align*}
Evaluating the geometric series leads to the expression in the statement.
\end{proof}

\else
    \appendices
    
\fi
\end{document}